\documentclass[10pt,journal,  compsoc]{IEEEtran}

 \usepackage{hyperref}
\usepackage{fancybox}   % For the overlay box
\usepackage{psfig}      % To iclude ".eps" Figures
\usepackage{graphicx}
\usepackage{amsmath}
\usepackage{color}
\usepackage{epsfig}
\usepackage{amssymb}
\usepackage{verbatim}
\usepackage{enumitem}
\usepackage{setspace}
\usepackage{subfigure}
\usepackage[]{algorithm2e}

 \usepackage[T1]{fontenc}
 \usepackage[english]{babel}
 \usepackage[font=small,labelfont=bf]{caption}

\makeatletter
\newcommand{\Rmnum}[1]{\expandafter\@slowromancap\romannumeral#1@}
\makeatother

\newtheorem{theorem}{Theorem}
\newtheorem{definition}{Definition}
\newtheorem{assumption}{Assumption}

\newtheorem{lemma}{Lemma}

\newtheorem{claim}{Claim}

\newcommand{\qed}{\hfill \ensuremath{\Box}}

\DeclareMathOperator*{\argmin}{arg\,min}

\begin{document}

\title{Asynchronous Stochastic Approximation Based 
Learning Algorithms for As-You-Go Deployment of Wireless Relay Networks along a Line\thanks{The research 
reported in this paper was supported by a Department of Electronics and Information 
Technology (DeitY, India) and NSF (USA) funded project on Wireless 
Sensor Networks for Protecting Wildlife and Humans, by an Indo-French Centre for 
Promotion of Advance Research (IFCPAR) funded project, and by the Department of Science and Technology (DST, India), 
 via a  J.C. Bose Fellowship.}
\thanks{Arpan Chattopadhyay is with the Electrical Engineering department, University of Southern California, Los Angeles. Avishek Ghosh is with the EECS department, UC Berkeley. 
Anurag Kumar is with the Department of ECE, 
Indian Institute of Science (IISc), Bangalore. This work was 
done when Arpan Chattopadhyay and Avishek Ghosh were with the Department of ECE, 
IISc. Email: arpanc.ju@gmail.com, avishek.ghosh38@gmail.com, anurag@ece.iisc.ernet.in}
\thanks{{\bf All appendices are provided in the supplementary material.}}
}

\author{
Arpan~Chattopadhyay, Avishek~Ghosh, and Anurag~Kumar\\
}

% \author{K.~P.~Naveen,~\IEEEmembership{Student~Member,~IEEE,} 
% Eitan~Altman,~\IEEEmembership{Fellow,~IEEE,} and
% Anurag~Kumar,~\IEEEmembership{Fellow,~IEEE} 

\IEEEcompsoctitleabstractindextext{
\begin{abstract}
We are motivated by the need, in emergency situations, for impromptu (or ``as-you-go'') 
deployment of multihop wireless networks, by human agents or robots (e.g., unmanned aerial vehicles (UAVs)); the agent moves along a line, makes wireless link quality measurements at regular intervals, and  makes on-line placement decisions using  these measurements. 
As a first step we have formulated such deployment along a line as a sequential decision problem. In our earlier work, reported in 
\cite{chattopadhyay-etal15measurement-based-impromptu-deployment-arxiv-v1}, we proposed   
two possible deployment approaches: (i) the pure as-you-go approach 
where the deployment agent can only move forward, and (ii) the explore-forward approach where the deployment agent explores a few successive steps 
and then selects the best relay placement location among them. The latter was shown to provide better performance 
(in terms of network cost, network performance and power expenditure), 
but at the expense of more measurements and deployment time, which makes explore-forward 
impractical for quick deployment by an energy constrained agent such as a UAV.  
Further, since in emergency situations the terrain would be unknown, the deployment algorithm should not 
require a-priori knowledge of the parameters of the  wireless propagation model. In 
\cite{chattopadhyay-etal15measurement-based-impromptu-deployment-arxiv-v1} we, therefore, developed 
learning algorithms for the explore-forward approach.\\

The current paper fills in an important gap by providing deploy-and-learn algorithms for the pure as-you-go approach. 
We formulate the sequential relay deployment problem as an average cost Markov decision process (MDP), which 
trades off among  power consumption, link outage probabilities, and the number of relay nodes in the deployed network. 
While the pure as-you-go deployment problem was previously formulated as a discounted cost MDP 
(see \cite{chattopadhyay-etal15measurement-based-impromptu-deployment-arxiv-v1}), the discounted cost MDP 
formulation was not amenable for learning algorithms that are proposed in this paper. In this paper, first 
we show  structural results for the optimal policy 
corresponding to the average cost MDP, and provide new insights into the optimal policy. 
Next, by exploiting the special structure of the average cost optimality 
equation and by using the theory of {\em asynchronous} stochastic approximation (in single and two timescale), we develop two learning algorithms  
that asymptotically converge to the set of optimal policies as deployment progresses. Numerical results show reasonably fast 
speed of convergence, and hence the model-free algorithms  
can be useful for practical, fast deployment of emergency wireless networks.
\end{abstract}

\vspace{-2mm}

}

 \maketitle

\section{Introduction}\label{section:introduction}
In emergency situations, such as fires in large buildings or forests, or houses in a flooded neighbourhood 
(without electric power and telecom infrastructure), there is a need to quickly deploy wireless networks 
for situation monitoring. Such networks could be deployed by first responders (e.g., fire-fighters moving 
through a burning building \cite{liu-etal10breadcrumb}), or by robots (e.g., unmanned aerial vehicles (UAVs) hopping over the 
rooftops of flooded homes or flying over a long road \cite{berkeley-project}, \cite{corke04autonomous-deployment-sensor-uav}, 
\cite{anthony-etal14controlled-sensor-installation-uav}), or by  forest guards along 
forest trails (\cite{chattopadhyay-etal15measurement-based-impromptu-deployment-arxiv-v1}).\footnote{See  \cite{dyo-etal10wildlife-wsn} and 
\cite[Section~$5$]{alkhatib14review-forest-fire} for application of multihop wireless 
sensor networks in wildlife monitoring and forest fire detection. \cite{nokia-project} 
illustrates a future possibility where drones deploy high speed, solar-powered 
access points on the roofs of city buildings in order to provide high speed internet connection. 
{\em The drone can land on the ground or on a rooftop for   link quality measurements, 
and can again take off}.} 
Typically, such networks would have one or more \emph{base-stations}, where the command and control would reside, and to which the measurements from the sensors in the field would need to be routed. For example, in the case of the fire-fighting example, the base-station would be in a control truck parked outside the building. 
Evidently, in such emergency situations,  there is a need for ``as-you-go'' deployment algorithms as 
there is no time for network planning. As they move through the affected area, the first-responders would need to deploy wireless relays, in order to provide routes for the wireless sensors for situation monitoring.

With the above motivation for quick deployment of multihop wireless networks, in our work, in the present and earlier papers (\cite{chattopadhyay-etal15measurement-based-impromptu-deployment-arxiv-v1}, \cite{chattopadhyay-etal13measurement-based-impromptu-placement_wiopt}, \cite{chattopadhyay-etal14deployment-experience}), we have considered  the particular situation of as-you-deployment of relays along a line, starting from a base-station, in order to connect a source of data (e.g., a sensor) whose location is revealed (or is itself placed) only during the deployment process. 
Figure~\ref{fig:why-impromptu}  depicts our model for as-you-go deployment along a line, and also illustrates the difference between 
planned deployment and as-you-go deployment. As-you-go deployment along a line is motivated by the need for quick deployment of relay networks along long forest trails by humans or mobile robots, and relay network deployment along a long straight road by human agents or UAVs. In practice, the location of the data source would be a-priori unknown,  as the deployment agent  would also need to select locations at which to place the sensors. 
Yet, as the deployment agent traverses the line, he or she (or it) has to judiciously deploy wireless relays so as to 
end up with a viable network connecting the data source (e.g., the sensor) to the sink. In a planned approach, all possible 
links could be evaluated; in an as-you-go approach, however, the agent needs to make decisions based on 
whatever links can be evaluated as deployment progresses.

\vspace{-0mm}

Motivated by the need for as-you-go deployment of wireless sensor networks (WSNs) over large terrains, such as forest trails, 
in our earlier work \cite{chattopadhyay-etal15measurement-based-impromptu-deployment-arxiv-v1} 
we had considered the problem of multihop wireless network deployment along a line, 
where a single deployment agent  starts  
from a sink node (e.g., a base-station), places relays as the agent  
walks along the line, and finally places a source node (e.g., a sensor) where required.  We formulated this 
problem as a measurement based sequential decision problem with an appropriate additive cost over hops. 
In order to explore the range of possibilities, we considered two alternatives for measurement and deployment: 
(i) the explore-forward approach: after placing a node, the deployment agent explores several potential placement 
locations along the next line segment, and then decides on where to place the next node, and (ii) the pure as-you-go approach: the deployment 
agent only moves forward, making measurements and committing to deploying nodes as he goes. 

\vspace{-0 mm}

As expected, in \cite{chattopadhyay-etal15measurement-based-impromptu-deployment-arxiv-v1} 
we found that the explore-forward approach yields better performance 
(in terms of the additive per hop cost (see 
\cite[Section~V]{chattopadhyay-etal15measurement-based-impromptu-deployment-arxiv-v1}); but, of course, this approach takes more 
time for the completion of deployment. Hence, explore-forward is  prohibitive  when soldiers or robots need to quickly deploy a relay network along a forest trail or along a long road. 
 In addition, a deployment agent such as a UAV would be limited by its fuel, 
and it would be desirable to complete the mission as quickly as possible, without many fuel consuming manoeuvres. 
Thus, pure as-you-go is the only option for network deployment by UAVs along long roads  (see \cite{berkeley-project} for practical network deployment along a road by a UAV). 
Further, in an emergency situation, the algorithm cannot expect to be given the parameters of the propagation environment; 
this gives rise to the need for deploy-and-learn algorithms. 

\vspace{-0 mm}

In \cite{chattopadhyay-etal15measurement-based-impromptu-deployment-arxiv-v1}, 
although we introduced explore-forward and pure as-you-go approaches, 
we developed  learning algorithms for explore-forward alone.
However, with the above motivation, 
our current paper fills in an important gap by proposing online learning algorithms for pure as-you-go deployment.   
We mathematically formulate the  problem  of
{\em pure as-you-go deployment}   
along a line as an {\em optimal sequential decision problem} so 
as to  minimize the expected average cost per step, where the cost 
of a deployment is   a linear combination of  the sum transmit 
power, the sum outage 
probability  and the number of relays deployed. 
We formulate the problem as a Markov decision process (MDP) and  obtain the optimal policy structure. 
Next, we propose two learning algorithms (based on asynchronous stochastic approximation) and 
prove their asymptotic  convergence to the optimal policy 
for the long-run average cost minimization 
problem. Finally, we demonstrate 
the convergence rate of the learning algorithms via numerical exploration.

{\em The new contributions of this paper, in relation to \cite{chattopadhyay-etal15measurement-based-impromptu-deployment-arxiv-v1}, 
are discussed in Section~\ref{subsection:contribution-wrt-prior-work}.}

\begin{figure}[!t]
\begin{center}
\includegraphics[height=2cm, width=8cm]{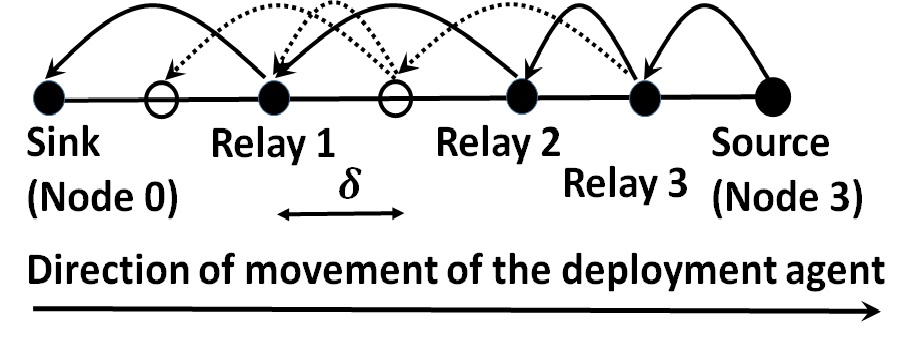}
\end{center}
\caption{A line network connecting a source (e.g., a sensor) 
to a sink (e.g., a control centre) via relay nodes. The dots  in between (filled and unfilled) 
denote potential relay locations, and are spaced $\delta$ meters apart. 
The deployed network consists of three relays (dots labeled Relay~$1$, $2$, and $3$) placed at three potential locations. The solid 
arrows show the multi-hop path from the source to the sink. The unfilled dots represent locations 
where no relay was placed. {\em The dotted arrows represent some other  
possible links between pairs of potential  locations.} In case of planned deployment, link qualities between all potential 
location pairs need to be measured. But, in as-you-go deployment, the agent only measures the qualities of link from his (or its) 
current location to the previously 
placed nodes.}
\label{fig:why-impromptu}
\vspace{-6mm}
\end{figure}

\vspace{-2mm}
\subsection{Related Work}\label{subsecion:related-work}

Prior work on the problem of impromptu deployment of WSN consists of mostly heuristic algorithms 
validated by experimentation. For example, the authors of \cite{souryal-etal07real-time-deployment-range-extension} 
address this problem by studying experimentally the variation in indoor link quality. The authors of 
\cite{aurisch-tlle09relay-placement-emergency-response} also took a similar approach. 
The authors of \cite{howard-etal02incremental-self-deployment-algorithm-mobile-sensor-network} provide 
heuristics for deploying (incrementally) sensors so that a certain area is covered (e.g., self-deployment of 
autonomous robot teams). Bao and Lee, in \cite{bao-lee07rapid-deployment-wireless}, 
address the problem of a group of first responders starting from  a base station (e.g.,  a command center) 
and placing relay nodes while walking 
through a region devoid of communication infrastructure, in order to stay connected among themselves as well as with the 
base station.  Liu et al., in \cite{liu-etal10breadcrumb}, describe a {\em breadcrumbs} system meant for 
firefighters operating inside a building; this paper is in similar spirit with ours, but their goal is just 
to maintain connection with  $k$ previously placed nodes. This work was later extended by them 
in \cite{liu-etal11multiuser-breadcrumb} which provides a reliable  multiuser breadcrumbs system. 
However, all the above works are based on heuristic algorithms, rather than 
on rigorous formulations; hence they do not provide any provable performance guarantee. 
A nice survey on rapid deployment of post-disaster networks is available in \cite{miranda16survey-rapidly-deployable-post-disaster}. 
Sensor network deployment by UAVs have also been studied in literature (see \cite{corke04autonomous-deployment-sensor-uav}, 
\cite{anthony-etal14controlled-sensor-installation-uav}).

In our current paper, we have formulated as-you-go deployment as an MDP, found structural results 
for the optimal policy, and proposed learning algorithms to solve the sequential decision problems 
without using any prior knowledge of the radio propagation parameters. 
The use of MDP to formulate as-you-go deployment was first proposed by Mondal et al. in 
\cite{mondal-etal12impromptu-deployment_NCC}. This work was later extended by 
Sinha et al. in \cite{sinha-etal12optimal-sequential-relay-placement-random-lattice-path}, 
where the authors have provided an algorithm derived from an MDP formulation, so as to create a multi-hop 
wireless relay network between a sink and a source located at an unknown location, 
by placing relay nodes along a random lattice path. However, these papers do not consider spatial variability of 
wireless link qualities due to shadowing, which allows them to develop deployment algorithms that place 
the next relay based on the distance from the previously placed relay. 

The spatial variation of link qualities due to shadowing requires measurement-based 
deployment; here the deployment agent makes placement decision at a given location based on the link quality 
to the previously placed node. Measurement-based as-you-go deployment was formulated  
first in \cite{chattopadhyay-etal13measurement-based-impromptu-placement_wiopt}, and was later 
extended in \cite{chattopadhyay-etal15measurement-based-impromptu-deployment-arxiv-v1}. 
The authors of \cite{chattopadhyay-etal15measurement-based-impromptu-deployment-arxiv-v1} have proposed 
two possible approaches for deployment along a line: (i) the {\em pure as-you-go approach} 
and (ii) the {explore-forward approach}. \cite{chattopadhyay-etal15measurement-based-impromptu-deployment-arxiv-v1} 
has provided MDP formulations and policy structures for both approaches; transition probabilities of the MDPs
depend on the radio propagation parameters in the environment,
and, in practice, these parameters are not known to the agent
prior to deployment. Hence, \cite{chattopadhyay-etal15measurement-based-impromptu-deployment-arxiv-v1} also 
provides {\em learning algorithms for the explore-forward approach}, that converge asymptotically to the set 
of optimal deployment policies as more and more measurements are made in course of deployment. One of these 
learning algorithms was used for actual network deployment (see \cite{chattopadhyay-etal15measurement-based-impromptu-deployment-arxiv-v1} and 
\cite{chattopadhyay-etal14deployment-experience}). Design of a two-connected network to guard against node and link failures was discussed in \cite{ghosh-etal14two-connected}, but it did not provide any learning algorithm.

We also developed, in \cite{chattopadhyay-etal16deploy-as-you-go}, as-you-go deployment algorithms for deploying a multi-relay 
line network, so that the end-to-end achievable rate is maximized; 
but it was done for an information-theoretic, full-duplex, multi-relay channel model where the nodes carry out decode-and-forward 
relaying. However, devices with such sophisticated relaying capability is not yet available for full commercial use. On the other hand, 
our current paper designs deployment algorithms for networks carrying packetized data, which is common in present day wireless standards.

\vspace{-3mm}
\subsection{Contributions of this paper, in relation to \cite{chattopadhyay-etal15measurement-based-impromptu-deployment-arxiv-v1}:}
\label{subsection:contribution-wrt-prior-work}
{\bf (i) New deploy-and-learn algorithms:} Our current paper provides learning algorithms 
for the pure as-you-go approach (Algorithm~\ref{algorithm:OptAsYouGoLearning} and Algorithm~\ref{algorithm:OptAsYouGoAdaptiveLearning}), whereas \cite{chattopadhyay-etal15measurement-based-impromptu-deployment-arxiv-v1} 
provides learning algorithms only for explore-forward. The learning algorithms are required to discover the optimal 
deployment policy as deployment progresses, for the situation where no prior accurate knowledge on the statistical nature of radio 
propagation environment is available. Learning algorithms for pure as-you-go 
deployment is an important requirement since the pure as-you-go deployment approach is more suitable for very fast 
deployment over a large region. In fact, the number of measurements in explore-forward deployment can be double or triple than that of 
pure as-you-go (\cite[Section~V]{chattopadhyay-etal15measurement-based-impromptu-deployment-arxiv-v1}) for practical deployment; this makes 
pure as-you-go a better choice for emergency  network deployment by soldiers or commandos or 
energy-constrained autonomous agents such as robots and UAVs.

Unlike \cite{chattopadhyay-etal15measurement-based-impromptu-deployment-arxiv-v1}, the learning algorithms 
presented in this paper make use of {\em asynchronous stochastic approximation}, where different iterates are updated at different time instants 
(in the learning algorithms proposed in \cite{chattopadhyay-etal15measurement-based-impromptu-deployment-arxiv-v1}, all iterates 
are updated when a new relay is placed). We provide formal proof for the convergence of our proposed 
learning algorithms to the optimal deployment policy for pure as-you-go deployment; 
these proofs require a significant 
and non-trivial novel mathematical analysis (compared to \cite{chattopadhyay-etal15measurement-based-impromptu-deployment-arxiv-v1}) 
in order to address many technical issues that arise in the proofs.

 In other words, the most important contributions of the current paper w.r.t. 
\cite{chattopadhyay-etal15measurement-based-impromptu-deployment-arxiv-v1}, are the newly proposed learning 
algorithms for pure as-you-go deployment and their convergence proofs, which are new to the literature and addresses the 
problem of very fast deployment.

Interestingly, one of the 
learning algorithms proposed in this paper exhibits a nice separation property between estimation and control, which is not 
present in the learning algorithms presented in \cite{chattopadhyay-etal15measurement-based-impromptu-deployment-arxiv-v1}.

{\bf (ii) Average cost MDP formulation:} \cite{chattopadhyay-etal15measurement-based-impromptu-deployment-arxiv-v1} 
formulates the pure-as-you deployment problem for a line having a random length $L \sim Geometric(\theta)$ 
(mean is $\frac{1}{\theta}$), i.e., $\mathbb{P}(L=l)=(1-\theta)^{l-1} \theta$ where $ l \in \{1,2,\cdots,\infty\}$; the average cost optimal policy is obtained by taking $\theta \rightarrow 0$. 
Clearly, this requires value iteration to compute the optimal policy prior to deployment. 
This also requires the knowledge of radio propagation parameters, since they  determine the transition 
probabilities of the MDP. On the other hand, our present paper establishes 
the structure of the optimal policy by using the average cost optimality equation (see \eqref{eqn:average_cost_optimality_equation_no_backtracking}) with 
necessary modification; it turns out that such a formulation along with the special structure of the 
problem enables us to propose very simple learning algorithms to find the optimal policy, irrespective of whether the radio propagation 
parameters are known apriori or not.  Thus, the average cost MDP formulation is a precursor to the learning algorithms (Algorithm~\ref{algorithm:OptAsYouGoLearning} and Algorithm~\ref{algorithm:OptAsYouGoAdaptiveLearning}) 
presented later in this paper.  Some new interesting properties of the value functions and the policy structure 
are also proved in the current paper, which were not present in \cite{chattopadhyay-etal15measurement-based-impromptu-deployment-arxiv-v1} 
because the problem was formulated as discounted cost MDP in \cite{chattopadhyay-etal15measurement-based-impromptu-deployment-arxiv-v1}.

{\bf (iii) Additional measurements to facilitate learning:} The pure-as-you go approach considered in our current paper is not exactly the same as that described in 
\cite{chattopadhyay-etal15measurement-based-impromptu-deployment-arxiv-v1}. In 
\cite{chattopadhyay-etal15measurement-based-impromptu-deployment-arxiv-v1}, the agent makes a link quality 
measurement from the current location to the immediate previous node that he had placed. On the contrary, in the pure 
as-you-go approach described in our present paper, the agent measures link qualities from the current location to all 
previously placed nodes that are located within a certain distance. This is done to facilitate learning the optimal policy. The exact reason behind using this variation of 
pure as-you-go deployment will be explained in Section~\ref{subsection:OptAsYouGoLearning_algorithm}.

{\bf (iv) Bidirectional traffic:} In Section~\ref{subsection:bidirectional-traffic}, we explain how the deployment algorithms presented in this paper can be adapted 
to the case where each link in the network has to carry data packets in both directions.

\vspace{-3mm}
\subsection{Organization}\label{subsection:organization}
The rest of the paper is organized as follows. The system model has been described in 
Section~\ref{section:system-model}. 
MDP formulation for pure-as-you deployment has been provided in Section~\ref{section:mdp-for-pure-as-you-go-deployment}. 
The learning algorithms have been proposed in 
Section~\ref{section:learning-for-pure-as-you-go-deployment-given-xi} 
and Section~\ref{section:learning-for-pure-as-you-go-deployment-constrained-problem}. Convergence speed of the learning 
algorithms are demonstrated numerically in Section~\ref{section:convergence_speed_learning_algorithms}, after which the conclusion follows. 
The proofs of all theorems are provided in the appendices available as supplementary material.

\vspace{-3mm}

\section{System Model}\label{section:system-model}
In this section, we describe the system model assumed in this paper.  It has to be noted that 
the system model and notation used in this paper are similar in many aspects to those of 
\cite{chattopadhyay-etal15measurement-based-impromptu-deployment-arxiv-v1}; a significant difference in the system model will be found 
in the deployment procedure as described in Section~\ref{subsection:deployment-process} (deployment process), and in Section~\ref{subsection:bidirectional-traffic} (bi-directional traffic).  The channel model (Section~\ref{subsection:channel-model}), traffic model (Section~\ref{subsection:traffic_model}) 
and deployment objective (Section~\ref{subsection:network_cost}) subsections are almost similar to the respective 
sections in \cite{chattopadhyay-etal15measurement-based-impromptu-deployment-arxiv-v1}; but we describe the system model here in detail 
to make this paper self-contained.

We assume that the  line (i.e., the road or the forest trail along which the network is deployed) is discretized into steps (starting from the sink), each having length 
$\delta$. The points located at distances $\{k \delta\}_{k \in \{1,2,3,\cdots\} }$ are called potential locations; 
the agent is allowed to place nodes only at these potential locations. 
As the {\em single} deployment agent walks along the line, at each potential location, the agent measures the link quality from the 
current location to the previously placed nodes that are within a certain distance from the current location; 
placement decisions are made based on these measurements. 

After deployment, as shown in Figure~\ref{fig:why-impromptu}, the sink is called Node $0$, and the relays are enumerated 
as nodes $\{1,2,3,\cdots\}$ as we move away from the sink. 
A link whose transmitter is Node $i$ and receiver is Node 
$j$ is called link $(i,j)$.

\subsection{Wireless Channel Model}\label{subsection:channel-model}

We consider a wireless channel model where, for a 
link (i.e., a transmitter-receiver pair) with length $r$  and transmit power $\gamma$, the 
received power of a packet (say the $k$-th packet) is given by:

\begin{equation}
 P_{rcv,k}=\gamma c \bigg(\frac{r}{r_0}\bigg)^{-\eta}H_kW \label{eqn:channel_model}
\end{equation}
Here $c$ is the path-loss at a reference distance $r_0$, 
and $\eta$ is the path-loss exponent. The fading random variable seen by the $k$-th packet 
is $H_k$ (e.g., $H_k$ is exponentially distributed for Rayleigh fading); it takes independent values over different coherent 
times. $W$ denotes the shadowing random variable that captures the (random) spatial variation in path-loss. 
In this paper, $W$ is assumed to take values from a set $\mathcal{W}$, and we denote by $g(w)$ 
the probability mass function or probability density function of $W$, depending on 
whether $\mathcal{W}$ is  countable   or   uncountable. 
We assume that the transmit power of each node comes from a discrete set, $\mathcal{S}:=\{P_1, P_2, \cdots, P_M \}$, where 
the power levels are arranged in ascending order. 

Shadowing becomes spatially uncorrelated if the transmitter or receiver is moved by a certain distance that depends on 
the sizes of the scatterers in the environment (see \cite{agarwal-patwari07correlated-shadow-fading-multihop}).  
 It was shown experimentally that, in a forest-like environment, shadowing 
has log-normal distribution (i.e., $\log_{10}W \sim \mathcal{N}(0,\sigma^2)$ where $\sigma$ is the standard 
deviation of log-normal shadowing) and the shadowing decorrelation distance can be as small as $6$~meters 
(see \cite{chattopadhyay-etal14deployment-experience}). In this paper, 
 we assume that the step size $\delta$ is chosen to be more than the shadowing decorrelation distance; 
this allows us to assume that the shadowing at any two different links in the network 
are independent.

The $k$-th packet is said to see an {\em outage} in the link if 
$P_{rcv,k} \leq P_{rcv-min}$, where $P_{rcv-min}$ is a threshold depending on the modulation scheme and the 
properties of the receiving node. For example, $P_{rcv-min}$ can be chosen to be $-88$~dBm for 
the TelosB ``motes'' (see \cite{bhattacharya-etal13smartconnect-comsnets}), and 
$-97$~dBm for iWiSe motes (see \cite{iwise}).  
For a link with length $r$, transmit power $\gamma$ and  shadowing realization $W=w$, 
the outage probability is denoted by $Q_{out}(r,\gamma,w)$; it is increasing in $r$ and decreasing 
in $\gamma$, $w$. 
$Q_{out}(r,\gamma,w)=\mathbb{P}(P_{rcv,k} \leq P_{rcv-min})$ depends on the 
fading statistics; if $H$ is exponentially distributed with mean $1$ 
(i.e., for Rayleigh fading), then  
$Q_{out}(r,\gamma,w)=\mathbb{P}( \gamma c (\frac{r}{r_0})^{-\eta}wH \leq P_{rcv-min} )=
1-e^{-\frac{P_{rcv-min}(\frac{r}{r_0})^{\eta}}{\gamma c w}}$. 
The outage probability of a randomly chosen link (with given $r$ and $\gamma$) is a 
random variable, with the randomness 
coming from shadowing $W$.   Outage probability can be measured by sending a 
large number of packets over a link and calculating the fraction of packets with received power less than $P_{rcv-min}$.

\subsection{Pure As-You-Go Deployment Process}\label{subsection:deployment-process}
After placing a relay, the agent starts measuring the link qualities from the next $B$ locations {\em one by one}  
(the value of $B$ is fixed prior to deployment). At any given location, the agent uses the measurements from the current location to make a placement decision; the agent does not make measurements from all of those $B$ locations in order to place a new relay.  

At any given location, the agent measures the link qualities 
from the given location to all 
previously placed nodes that are within $B \delta$ distance from the current location; see 
Figure~\ref{fig:pure-as-you-go-learning}. Let us assume that the 
agent is standing at a distance $k \delta$ from the sink. 
Let $\mathcal{I}_k:=\{ r \in \{1,2,\cdots,B \}: \text{a relay was placed at a distance $(k-r)\delta$ from the sink} \}\}$. 
Then, the agent at this location will measure the outage probabilities 
$\{Q_{out}(r,\gamma,w_r)\}_{\gamma \in \mathcal{S},  r \in \mathcal{I}_k}$ ($w_r$ is the 
realization of shadowing in a link of length $r$~steps). 

However, at each location, only the link quality to the {\em immediately} previous node 
is used to decide whether to place a relay there or to move on to the next step. If the decision is to place a relay, 
then the agent also decides which transmit power $\gamma \in \mathcal{S}$ to use at that particular node. If the decision is not to 
place a relay, the agent moves to the next step. In this process, 
if he reaches $B$ steps away from the previous relay, or if the source location is encountered, then he must place a node there.

It is important to note that, while the measurement to the {\em immediately} previous node is used to make a 
placement decision, other measurements made in this process  provide useful information about the statistical characteristics of the 
radio propagation environment (more precisely, the probability distribution of $Q_{out}(r,\gamma,\cdot)$ for 
$r \in \{1,2,\cdots,B\}, \gamma \in \mathcal{S}$), and 
those measurements are used to learn the optimal deployment policy as described in 
Section~\ref{section:learning-for-pure-as-you-go-deployment-given-xi} and 
Section~\ref{section:learning-for-pure-as-you-go-deployment-constrained-problem}. But if the radio propagation parameters 
(such as $\eta$ and $\sigma$) are exactly known, i.e., if the probability distribution of $Q_{out}(r,\gamma,\cdot)$ is known exactly, 
then these additional measurements will not be required (since shadowing is i.i.d. across links, these measurements will not provide any information 
about the link quality between the current location and the immediately previous node); 
this situation has been explored in Section~\ref{section:mdp-for-pure-as-you-go-deployment},  
where measurement is made only to the previously placed relay node.

{\em Choice of $B$:} In general, the choice of  $B$ depends on the constraints 
and requirements for the deployment. Large  $B$ results in   better performance at the expense of more measurements.  One can simply  choose $B$ to be the largest integer such that, the probability that a randomly chosen wireless link with length $B \delta$ respects a certain outage constraint, is larger than some pre-specified target. This will make sure that the probability of obtaining a workable link is small in case the agent reaches a location that is more than $B \delta$ distance away from the previously placed node.

\begin{figure}[!t]
\begin{center}
\includegraphics[height=2.5cm, width=7cm]{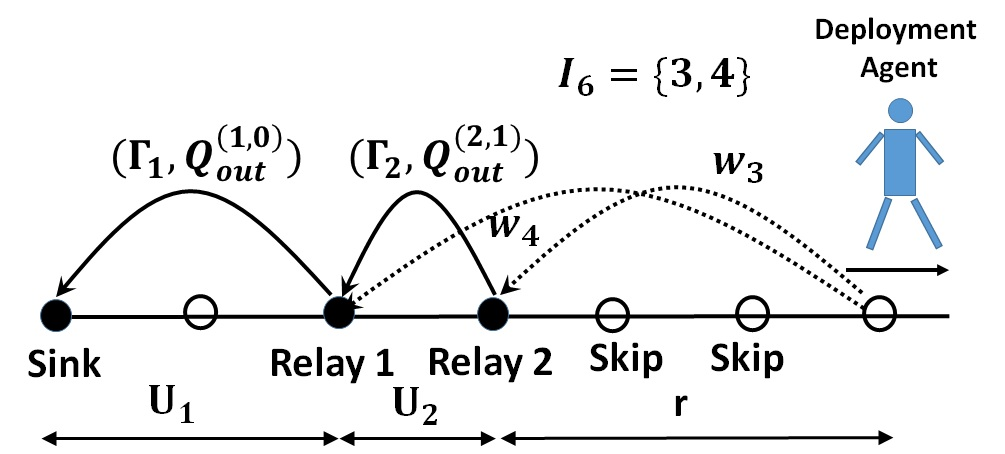}
\end{center}
\caption{Illustration of  pure as-you-go deployment with learning for $B=4$. Here 
the deployment agent has already placed Relay~$1$ and Relay~$2$,  and the corresponding inter-relay distances are 
$U_1$ and $U_2$. The placed relays use transmit powers   $\Gamma_1$ and $\Gamma_2$, 
thereby achieving outage probabilities $Q_{out}^{(1,0)}$ and $Q_{out}^{(2,1)}$ 
(in the links shown by solid arrows). After placing Relay~$2$, the agent measured the link qualities from the next 
location to the sink, Relay~$1$ and Relay~$2$ (since $B=4$) and the algorithm advised him not to place a node there. 
Then the deployment agent moved to the next location (which is at a distance of $2 \delta$ from Relay~$2$) 
and measured the link qualities to Relay~$1$ and Relay~$2$ (but not to the sink since $B=4$). 
In this {\em snap-shot} of the deployment process, the agent is 
evaluating the next location  at $r=3 \delta$ distance from Relay~$2$ (see the dotted arrows). Since $B=4$, the agent  measures the 
link qualities from the current location to both Relay~$1$ and Relay~$2$; this corresponds to $\mathcal{I}_6=\{3,4\}$ 
(see Section~\ref{subsection:deployment-process} for the definition of $\mathcal{I}_6$), since the distances to 
Relay~$2$ and Relay~$1$ from the current location are $3 \delta$ and $4 \delta$ respectively. 
Based on these measurements, the deployment agent will decide whether to place a relay at $r=3 \delta$ or not, and the transmit power of the node 
in case the decision is to place; if the 
decision is not to place a relay here, then a relay must be placed at the next location (since $B=4$),  
and the agent would be at a distance of $B \delta$ from the last placed relay (i.e., Relay~$2$).}
\label{fig:pure-as-you-go-learning}
\end{figure}

\vspace{-4mm}
\subsection{Network Cost Minimization Objective}\label{subsection:network_cost}
We first define the cost that we use to evaluate the performance of any deployment policy. 
A deployment policy $\pi$ takes as input the distance of the current location of the agent from 
the previous relay and the link quality to the previously 
placed node, and provides the placement decision for that location and transmit power (if the decision 
is to place a relay) as output.

We denote the number of relays placed up to $x$~steps from the sink by $N_x$, and let us define $N_0 = 0$. 
Since deployment decisions are based on measurements of (random) outage probabilities, $\{N_x\}_{x \geq 1}$ 
is a random process.

After the deployment is over, let us denote by $\Gamma_i$ the transmit power used by node~$i$, and by $Q_{out}^{(i, i-1)}$ 
the outage probability over the link $(i, i-1)$ (see Figure~\ref{fig:pure-as-you-go-learning}). 
Note that, $\Gamma_i$ and $Q_{out}^{(i, i-1)}$ are random variables since shadowing between various potential 
location pairs are random variables, whose exact realization is known only after measurement. Given the measurement values 
(i.e., the information available to the deployment agent) 
and  the deployment policy, one can find the exact realizations of $\Gamma_i$ and $Q_{out}^{(i, i-1)}$. 

The expected cost of the deployed network up to $x \delta$ distance is given by a sum of hop costs as follows: 

\begin{equation}
\mathbb{E}_{\pi} ( \sum_{i=1}^{N_x} \Gamma_i + \xi_{out} \sum_{i=1}^{N_x}Q_{out}^{(i,i-1)}+ \xi_{relay} N_x )
\label{eqn:cost_function_sum_power_sum_outage}
\end{equation}

which is the expectation (under policy $\pi$) of a linear combination of the sum power $\sum_{i=1}^{N_x} \Gamma_i$, 
the sum outage $\sum_{i=1}^{N_x} Q_{out}^{(i, i-1)}$, and the number of relays $N_x$. 
For small outage probabilities, 
the sum-outage $\sum_{i=1}^{N_x} Q_{out}^{(i, i-1)}$ is
approximately equal to the probability that a packet sent from the point $x$ to the sink 
encounters an outage along the path (see also Section~\ref{subsection:traffic_model} for a better understanding of the 
outage cost in light of 
the traffic model). The sum power 
$\sum_{i=1}^{N_x} \Gamma_i$ is proportional to the battery depletion rate in the network, in case 
wake-on radios are used (see \cite[Section~II]{chattopadhyay-etal15measurement-based-impromptu-deployment-arxiv-v1} for a detailed 
discussion).

The multipliers $\xi_{out} \geq 0$ and $\xi_{relay} \geq 0$ capture the emphasis we place on 
$\sum_{i=1}^{N_x} Q_{out}^{(i, i-1)}$ or $N_x$. A large value of $\xi_{out}$ will aim 
for   deployment with smaller end-to-end expected outage.  $\xi_{relay}$ can be viewed as the cost 
of placing a relay.

Since  the distance $L$ to the source from the sink 
is not known prior to deployment, we simply assume that $L=\infty$. 
 This assumption is practical when the distance of the source from the sink is large (e.g., deployment 
along a long forest trail). $L=\infty$ is also equivalent to the scenario where deployment is done serially  
along multiple trails in a forest, provided that the radio propagation environment in various trails are statistically identical;   
we deploy serially along multiple lines but use this formulation to minimize the per-step cost averaged over 
all the lines.

Next, we define the optimization problems that we seek to address in this paper.

\subsubsection{The Unconstrained Problem}\label{subsubsection:the-unconstrained-problem}
 
We seek to solve the following problem:

\footnotesize
\begin{equation}
 \inf_{\pi \in \Pi} \limsup_{x \rightarrow \infty} \frac{\mathbb{E}_{\pi}\sum_{i=1}^{N_x}(\Gamma_i+\xi_{out}Q_{out}^{(i,i-1)}+\xi_{relay})}{x}
\label{eqn:unconstrained_problem_average_cost_with_outage_cost}
\end{equation}
\normalsize
where $\Pi$ is the set of all possible placement policies. 
We formulate (\ref{eqn:unconstrained_problem_average_cost_with_outage_cost}) as an average cost MDP. 

\subsubsection{The Constrained Problem}\label{subsubsection:the-constrained-problem}
(\ref{eqn:unconstrained_problem_average_cost_with_outage_cost}) is the relaxed version of the following 
constrained problem:

\footnotesize
\begin{eqnarray}
&& \inf_{\pi \in \Pi} \limsup_{x \rightarrow \infty} \frac{\mathbb{E}_{\pi}\sum_{i=1}^{N_x}\Gamma_i}{x} \nonumber\\
&s.t.& \, \limsup_{x \rightarrow \infty} \frac{\mathbb{E}_{\pi}\sum_{i=1}^{N_x}Q_{out}^{(i,i-1)}}{x} \leq \overline{q}, \nonumber\\
&& \text{and  } \limsup_{x \rightarrow \infty} \frac{\mathbb{E}_{\pi}N_x}{x} \leq \overline{N}
\label{eqn:constrained_problem_average_cost_with_outage_cost}
\end{eqnarray}
\normalsize

Here we seek to minimize 
the mean power per step subject to  constraints on the mean outage per step and   the 
mean number of relays per step. 

It turns out that (\ref{eqn:unconstrained_problem_average_cost_with_outage_cost}) is the relaxed version of the constrained 
problem (\ref{eqn:constrained_problem_average_cost_with_outage_cost}), with $\xi_{out}$ and $\xi_{relay}$ as the Lagrange multipliers. 
The constrained problem can be solved by solving the unconstrained problem, under proper choice of the Lagrange multipliers. 
The following  theorem tells us how to choose the {\em Lagrange multipliers} 
$\xi_{out}$ and $\xi_{relay}$ (see \cite{beutler-ross85optimal-policies-controlled-markov-chains-constraint}, 
Theorem~$4.3$):

\begin{theorem}\label{theorem:how-to-choose-optimal-Lagrange-multiplier}
For the constrained problem (\ref{eqn:constrained_problem_average_cost_with_outage_cost}), if there exists a pair 
$\xi_{out}^* \geq 0$, $\xi_{relay}^* \geq 0$ and a policy $\pi^*$ 
such that $\pi^*$ is the optimal policy of the unconstrained problem 
(\ref{eqn:unconstrained_problem_average_cost_with_outage_cost}) under $(\xi_{out}^*, \xi_{relay}^*)$, and if the constraints in 
(\ref{eqn:constrained_problem_average_cost_with_outage_cost}) are met with equality under the policy $\pi^*$, 
then $\pi^*$ is an optimal policy for the constrained problem (\ref{eqn:constrained_problem_average_cost_with_outage_cost}) 
as well.\qed
\end{theorem}

\subsection{Traffic Model}\label{subsection:traffic_model}
Motivated by our prior  work reported in   \cite{mondal-etal12impromptu-deployment_NCC}, 
\cite{sinha-etal12optimal-sequential-relay-placement-random-lattice-path}, 
\cite{chattopadhyay-etal13measurement-based-impromptu-placement_wiopt}, 
\cite{chattopadhyay-etal15measurement-based-impromptu-deployment-arxiv-v1},  we 
assume that the traffic in the network is so light that there is only one packet in the network at a time; this model is called 
the ``lone packet model'' (or the {\em zero traffic} model). 
This model results in collision-free transmissions, since there 
are no simultaneous transmissions in the network. As a result, we can easily write down the
communication cost in the line network as a sum of hop costs (Section~\ref{subsection:network_cost}). 

 It has been formally shown that network design under the lone packet model may be necessary for 
designing a network with positive traffic carrying capability (see \cite[Section~II]{bhattacharya-kumar12qos-aware-survivable-network-design}). 
We can easily adapt the result of \cite[Section~II]{bhattacharya-kumar12qos-aware-survivable-network-design} to show that, for a finite line network, 
if a target end-to-end packet delivery probability has to be achieved under positive traffic, then it is necessary to achieve 
that target  under lone packet traffic. Now, the end-to-end packet error rate under lone packet traffic 
is approximately equal to the sum outage; this justifies the sum outage cost in 
\eqref{eqn:unconstrained_problem_average_cost_with_outage_cost} and the outage constraint in 
\eqref{eqn:constrained_problem_average_cost_with_outage_cost}. Network design for a given positive 
traffic rate is left for future research.

In a line network, if interference-free communication is achieved via multi-channel access and frequency reuse after several hops, then the traffic model essentially becomes lone packet. There have been recent efforts to use multiple channels available in $802.15.4$ radio in WSN; see \cite{lohier-etal12multichannel-wsn}, \cite{abdeddaim-etal12multichannel-cluster-tree-wsn}, 
\cite{toscano-bello12multichannel-superframe-scheduling-wsn}, 
\cite{bardella-etal10experimental-multichannel-transmission-wsn}.

The lone packet traffic model is realistic for WSNs carrying low duty cycle 
measurements, or just  an occasional alarm packet. For example, recently 
developed passive infra-red (PIR) sensor platforms  
can detect and classify human or animal intrusion  
(\cite{raviteja-etal15animation-intrusion-classification-PIR}); such sensors deployed in a forest generate very low data.  The paper 
\cite[Section~3.2]{dyo-etal10wildlife-wsn} 
uses  $1.1\%$ duty cycle for a multi-hop WSN for wildlife monitoring; the sensors gather data from RFID collars tied the animals, and generate light traffic.  Very light traffic model is also realistic for condition monitoring/industrial telemetry applications 
(\cite{aghaei11wsn-water-gas}), where infrequent measurements are taken.   
Very light traffic model is  also 
common in machine-to-machine communication (\cite{adame-etal14m2m}).  The paper 
\cite[Table~$1$, Table~$3$]{mainwaring-etal02wsn-habitat} illustrate sensors with small sampling rate and sampled data size; it shows several   bytes per second data rate requirement for habitat monitoring.

We  assume that data packets traverse the network in a hop-by-hop fashion, without skipping any 
intermediate relay.  Later we will explain in Section~\ref{subsecion:a-note-on-the-unconstrained-objective-function} 
why we do not consider the possibility of relay skipping in this paper; the reason is  increased computational 
complexity without a very significant gain in network performance.

\subsection{Extension to Bi-Directional Traffic Flow}
\label{subsection:bidirectional-traffic}
Let us consider the situation where the traffic is still lone packet, but a packet can flow towards either direction along the line network with 
equal probabilities. In such cases, one can define the cost of link $(i,i-1)$ as 
$\Gamma_{i, forward}+\Gamma_{i-1,reverse}+ \xi_{out} Q_{out}^{(i,i-1, forward)}+\xi_{out} Q_{out}^{(i-1,i, reverse)}+\xi_{relay}$, where 
$\Gamma_{i, forward}$ is the transmit power used  from node~$i$ to node~$(i-1)$, and 
$\Gamma_{i-1, reverse}$ is the transmit power used  from node~$(i-1)$ to node~$i$. Similar meanings apply for 
the outage probabilities $Q_{out}^{(i,i-1, forward)}$ and $Q_{out}^{(i-1,i, reverse)}$, under transmit power levels 
$\Gamma_{i, forward}$ and $\Gamma_{i-1, reverse}$, respectively. It has to be noted that the  shadowing between two potential 
locations in forward and reverse directions, $W_{forward}$ and $W_{reverse}$, 
may not necessarily be independent. But the shadowing random variable pair 
$(W_{forward}, W_{reverse}) \in \mathbb{R}_{+}^2$ between two potential locations have a joint distribution, and 
this pair assumes independent and identically distributed (i.i.d.) value in $\mathbb{R}_{+}^2$ if either the transmitter or the receiver is moved 
beyond the shadowing decorrelation distance (which is smaller than the step size $\delta$). Hence, with this new link cost, our formulation 
\eqref{eqn:unconstrained_problem_average_cost_with_outage_cost} can easily be adapted to deploy a network carrying bi-directional traffic. 
In the process of deployment, the agent has to measure link qualities in both forward and reverse directions in such situation. 
The action at each step is to decide whether to place a relay; if the decision is to place a relay, then the agent also decides 
the transmit power levels used in that link along the forward and the reverse directions. 

Since the design for bi-directional traffic carrying network is mathematically equivalent to the 
design for unidirectional traffic carrying network, {\em we will consider only unidirectional traffic for the rest of this paper.}

\vspace{-2mm}
\section{Formulation for known propagation parameters}\label{section:mdp-for-pure-as-you-go-deployment}
Throughout this section, we will assume that we seek to solve the 
unconstrained problem given in (\ref{eqn:unconstrained_problem_average_cost_with_outage_cost}), and that the radio 
propagation parameters (such as $\eta$ and the standard deviation $\sigma$ for log-normal shadowing) are known prior 
to deployment. We formulate the problem as an average cost MDP, and develop a threshold policy for deployment. 
  In the process, we also discover some interesting properties of the value function, which do not follow 
from the discounted cost formulation.

 Note that, we assume throughout this section that measurement only to the immediately previous node 
is used to make a placement decision at any given location. Measurement to more than one previous nodes will be used later in order 
to develop the learning algorithms.

\subsection{Markov Decision Process (MDP) Formulation}\label{subsection:mdp-formulation}
When the deployment agent is $r$ steps away from the previous node ($r \in \{1,2,\cdots,B\}$), 
the agent measures the outage probabilities $\{Q_{out}(r,\gamma,w)\}_{\gamma \in \mathcal{S}}$ 
on the link from the current location to the previous node,\footnote{Note that, for the time being, we 
will ignore the measurements made to other nodes from the set $\mathcal{I}_k$.} where $w$ is the realization of  shadowing 
in that link. Then the algorithm decides whether to place a relay there, and also the transmit power 
$\gamma \in \mathcal{S}$ in case it decides to place a relay. We formulate the problem as an average cost  MDP 
with state space $\{1,2,\cdots,B\} \times \mathcal{W}$, where a typical state is of the form 
$(r,w), 1 \leq r \leq B, w \in \mathcal{W}$.  
If $r \leq B-1$, the action is either to 
place a relay and select a transmit power, or not to place. If $r=B$, the only 
feasible action is to place and select a transmit power $\gamma \in \mathcal{S}$. 
If a relay is placed at state $(r,w)$ and if a transmit power $\gamma$ is chosen for it, then 
a hop-cost of $\gamma+\xi_{out}Q_{out}(r,\gamma,w)+\xi_{relay}$ is incurred.\footnote{We have taken $(r,w)$ as a typical state for the sake of simplicity in representation;  
for the channel model given by (\ref{eqn:channel_model}), we can also take  
$(r,\{Q_{out}(r,\gamma,w)\}_{\gamma \in \mathcal{S}})$ as a typical state, since the cost of an action 
depends on the state $(r,w)$ only via the outage probabilities.}

A deterministic Markov policy $\pi$ is a sequence of mappings $\{\mu_k\}_{k \geq 1}$ 
from the state space to the action space. The policy $\pi$ is called a stationary policy if $\mu_k=\mu$ for all $k$. 
Given the state (i.e., the measurements), the policy provides the placement decision.

\subsection{Optimal Policy Based on Average Cost Optimality Equation}
\label{subsection:average-cost-no-backtracking-poisson-equation-approach}
We will first derive the structure of an optimal policy based on the average cost optimality equation 
(ACOE). Let $\lambda^*$ (or $\lambda^*(\xi_{out},\xi_{relay})$) be the optimal average cost per step for the unconstrained 
problem (\ref{eqn:unconstrained_problem_average_cost_with_outage_cost}) under the pure as-you-go deployment approach,   
and let $v^*(r,w)$ be the differential cost for the state $(r,w)$, where $1 \leq r \leq B$ and $w \in \mathcal{W}$. 
The average cost optimality equation for our MDP is as follows (by the theory of 
\cite[Chapter~$4$]{bertsekas07dynamic-programming-optimal-control-2}, for the case of finite $\mathcal{W}$, 
and by the theory developed 
in \cite[Chapter~$5$]{lerma-lasserre96mdp-book}, when $\mathcal{W}$ is a Borel subset of the real line): 

\footnotesize
\begin{eqnarray}
 v^*(r,w)&=&\min \bigg \{\min_{\gamma \in \mathcal{S}}(\gamma+ \xi_{out} Q_{out}(r,\gamma, w))+\xi_{relay}-\lambda^*  \nonumber\\
&& +  \sum_{w'}g(w')v^*(1,w'), -\lambda^*+\sum_{w'}g(w')v^*(r+1,w') \bigg \} \nonumber\\
&&  \forall 1 \leq r \leq B-1 \nonumber\\
v^*(B,w)&=& \min_{\gamma \in \mathcal{S}}(\gamma+ \xi_{out} Q_{out}(B,\gamma, w))+\xi_{relay}-\lambda^* \nonumber\\
&& +\sum_{w'}g(w')v^*(1,w') \label{eqn:average_cost_optimality_equation_no_backtracking}
\end{eqnarray}
\normalsize
where $g(w)$ was defined (in Section~\ref{subsection:channel-model}) 
to be the probability mass function or probability density function of shadowing $W$.

The ACOE~(\ref{eqn:average_cost_optimality_equation_no_backtracking}) can be explained as follows. When the state is $(r,w)$, 
the deployment agent can either place or may not place a relay. If he places a relay, he will incur a stage 
cost of $\min_{\gamma \in \mathcal{S}}(\gamma+ \xi_{out} Q_{out}(r,\gamma, w))+\xi_{relay}$ and the next (random) state is $(1, W')$, where 
$W'$ has p.m.f. or p.d.f. $g(w')$. 
If he does not place, then he incurs $0$ cost at that step and the next state is $(r+1, W')$. 
When at state $(B,w)$, he can only place a relay 
and incur a cost of $\min_{\gamma \in \mathcal{S}}(\gamma+ \xi_{out} Q_{out}(B,\gamma, w))+\xi_{relay}$ at that stage and the 
next (random) state is $(1,W')$. Note that, $\min_{\gamma \in \mathcal{S}}$ appears in the single-stage cost because 
choice of transmit power of the placed node is also a part of the action, and a transmit power is chosen so that 
the single-stage cost for a placed relay is minimized.

Note that, by   multiplying both sides of (\ref{eqn:average_cost_optimality_equation_no_backtracking}) with 
$g(w)$ and taking summation over $w$, we obtain the following:

\footnotesize
\begin{eqnarray}
 V(r)&=&\mathbb{E}_{W} \min \bigg \{ \min_{\gamma \in \mathcal{S}}(\gamma+ \xi_{out} Q_{out}(r,\gamma, W))+\xi_{relay}-\lambda^* \nonumber\\
     &  &   + V(1), -\lambda^*+V(r+1) \bigg \} \forall 1 \leq r \leq B-1 \nonumber\\
V(B)&=&\mathbb{E}_{W} \min_{\gamma \in \mathcal{S}}(\gamma+ \xi_{out} Q_{out}(B,\gamma, W))+\xi_{relay}-\lambda^*+ V(1) \nonumber\\
  \label{eqn:average_cost_optimality_equation_no_backtracking-V-equation}
\end{eqnarray}
\normalsize

where $V(r)=\sum_{w}g(w)v^*(r,w) \forall 1 \leq r \leq B$. Now, it is easy to see that if any 
$V(\cdot)$ satisfies (\ref{eqn:average_cost_optimality_equation_no_backtracking-V-equation}), then $V(\cdot)+c$ for any 
constant number $c$ also satisfies (\ref{eqn:average_cost_optimality_equation_no_backtracking-V-equation}). Hence, 
we can put $V(1)=\lambda^*$ in (\ref{eqn:average_cost_optimality_equation_no_backtracking-V-equation}) and obtain:

\footnotesize
\begin{eqnarray}
 V(r)&=&\mathbb{E}_{W} \min \bigg \{ \min_{\gamma \in \mathcal{S}}(\gamma+ \xi_{out} Q_{out}(r,\gamma, W))+\xi_{relay}, \nonumber\\
     &  &    V(r+1)-V(1) \bigg \} \forall 1 \leq r \leq B-1 \nonumber\\
V(B)&=&\mathbb{E}_{W} \min_{\gamma \in \mathcal{S}}(\gamma+ \xi_{out} Q_{out}(B,\gamma, W))+\xi_{relay} 
  \label{eqn:average_cost_optimality_equation_no_backtracking-V-equation-no-lambda}
\end{eqnarray}
\normalsize

{\em Remark:} Let $c(r,W):=\min_{\gamma \in \mathcal{S}}(\gamma+ \xi_{out} Q_{out}(r,\gamma, W))+\xi_{relay}$ be the (random) cost incurred if we place 
a relay at a distance $r$ from the previous relay. \eqref{eqn:average_cost_optimality_equation_no_backtracking-V-equation-no-lambda} 
shows the criteria for optimality to be $V(r)=\mathbb{E}_{W} \min \{c(r,W), V(r+1)-V(1)\}$ for $r \leq B-1$ and 
$V(B)=\mathbb{E}_{W} c(B,W)$. We will see in 
Algorithm~\ref{algorithm:optimal-policy-structure} that, by solving this system of (nonlinear) equations, one can find the optimal policy; there is no need 
to compute the differential cost for each state explicitly. Also, \eqref{eqn:average_cost_optimality_equation_no_backtracking-V-equation-no-lambda} 
will be particularly useful when we develop online deploy-and-learn algorithms in later sections, using the theory 
of stochastic approximation.

\begin{theorem}\label{theorem:uniqueness_of_V}
 There exists a unique vector $\underline{V}^*=[V^*(1), \, V^*(2), \, \cdots \, , V^*(B)]^T$ satisfying 
(\ref{eqn:average_cost_optimality_equation_no_backtracking-V-equation-no-lambda}). Also, 
$V^*(r) \geq r V^*(1)$ for all $r \in \{1,2,\cdots,B-1\}$ and $V^*(r)$ is increasing in $r$.
\end{theorem}

\begin{proof}
 See Appendix~\ref{appendix:mdp-for-pure-as-you-go-deployment}.
\end{proof}

\subsubsection{Policy Structure}
Algorithm~\ref{algorithm:optimal-policy-structure} specifies the optimal decision when the agent is $r$ steps away from the previously placed node and the shadowing realization from the current location to the previously placed node is $w$.
\begin{algorithm}[t!]
\hrule
{\bf Input:} $\xi_{out}$, $\xi_{relay}$, $\underline{V}^*$.\\
{\bf Output:} Placement decision at each step.\\
{\bf Pre-compute:} The threshold values $c_{th}(r):=V^*(r+1)-V^*(1)$ for all $1 \leq r \leq B-1$.\\
{\bf Initialization:} $r=1$ (distance from the previous node)\\
 \While{$1 \leq r \leq B$}{
 Measure  $Q_{out}(r,\gamma,w) \forall {\gamma \in \mathcal{S}}$;\\
  \uIf{$r \leq B-1$ and $\min_{\gamma \in \mathcal{S}}(\gamma+ \xi_{out} Q_{out}(r,\gamma, w))+\xi_{relay} \leq c_{th}(r)$}
{Place a new relay and use transmit power $\arg \min_{\gamma \in \mathcal{S}}(\gamma+ \xi_{out} Q_{out}(r,\gamma, w))$;\\
   Move to next step and set $r=1$;}
    \uElseIf{$r \leq B-1$ and $\min_{\gamma \in \mathcal{S}}(\gamma+ \xi_{out} Q_{out}(r,\gamma, w))+\xi_{relay} > c_{th}(r)$}
    {Do not place a relay and move to next step;\\
    $r=r+1$;\\
    }
    \Else
    {Place a new relay (since $r=B$);\\
    Use transmit power $\arg \min_{\gamma \in \mathcal{S}}(\gamma+ \xi_{out} Q_{out}(B,\gamma, w))$;\\
    Move to next step;\\
    Set $r=1$.
    }
 }

\hrule
\caption{OptAsYouGo Algorithm}
\label{algorithm:optimal-policy-structure}
\end{algorithm}

\begin{theorem}\label{theorem:optimal-policy-structure}
 The policy given by Algorithm~\ref{algorithm:optimal-policy-structure} is optimal 
for the unconstrained problem in (\ref{eqn:unconstrained_problem_average_cost_with_outage_cost}). The threshold $c_{th}(r)$ is 
increasing in $r$. 
\end{theorem}

\begin{proof}
 From (\ref{eqn:average_cost_optimality_equation_no_backtracking}), the optimal policy is to place a relay 
at state $(r,w)$ if the cost of placing is less than the cost of not placing. Hence, the policy structure 
follows from equations (\ref{eqn:average_cost_optimality_equation_no_backtracking}), 
(\ref{eqn:average_cost_optimality_equation_no_backtracking-V-equation}) 
and (\ref{eqn:average_cost_optimality_equation_no_backtracking-V-equation-no-lambda}). $c_{th}(r)$ is increasing 
in $r$ since $V^*(r+1)$ is increasing in $r$.
\end{proof}

{\em We denote the optimal policy given by Algorithm~\ref{algorithm:optimal-policy-structure} by 
$\pi^*(\xi_{out}, \xi_{relay})$.}

\subsection{Some properties of the optimal cost}
\label{subsection:properties-optimal-cost}

Let us consider a sub-class of stationary deployment policies (parameterized by $\underline{V}$, 
$\xi_{out} \geq 0$ and $\xi_{relay} \geq 0$) where $\underline{V}^*(\cdot)$ in 
Algorithm~\ref{algorithm:optimal-policy-structure} is replaced by any vector 
$\underline{V}$. Under this sub-class of policies, 
let us denote by $(U_k,\Gamma_k,Q_{out}^{(k,k-1)}), k \geq 1,$  the sequence of inter-node distances, transmit  
powers and link outage probabilities (see Figure~\ref{fig:pure-as-you-go-learning}).  Since 
shadowing is i.i.d. across links, the deployment process probabilistically 
restarts after each relay placement. Hence,   
$(U_k,\Gamma_k,Q_{out}^{(k,k-1)}), k \geq 1,$ is an i.i.d. sequence. Let 
$\overline{\Gamma}(\underline{V},\xi_{out},\xi_{relay})$, $\overline{Q}_{out}(\underline{V},\xi_{out},\xi_{relay})$ 
and $\overline{U}(\underline{V},\xi_{out},\xi_{relay})$ 
denote the mean power per link, mean outage 
per link and mean placement distance (in steps) respectively, under this sub-class of policies. We denote  
by $\overline{\Gamma}^*(\xi_{out},\xi_{relay})$, $\overline{Q}_{out}^*(\xi_{out},\xi_{relay})$ and $\overline{U}^*(\xi_{out},\xi_{relay})$ 
the optimal mean power per link, the optimal mean outage 
per link and the optimal mean placement distance (in steps) respectively, under 
Algorithm~\ref{algorithm:optimal-policy-structure}, where  $\underline{V}^*$ is used instead of any general $\underline{V}$.

Now, the optimal mean power per step, the optimal mean outage per step, 
and the optimal mean number of relays per step are given by 
$\frac{\overline{\Gamma}^*(\xi_{out},\xi_{relay})}{\overline{U}^*(\xi_{out},\xi_{relay})}$, 
$\frac{\overline{Q}_{out}^*(\xi_{out},\xi_{relay})}{\overline{U}^*(\xi_{out},\xi_{relay})}$ 
and $\frac{1}{\overline{U}^*(\xi_{out},\xi_{relay})}$ (by the Renewal-Reward theorem).

\begin{theorem}\label{theorem:lambda-increasing-concave-continuous-in-xi}
 The optimal average cost per step for problem 
(\ref{eqn:unconstrained_problem_average_cost_with_outage_cost}), $\lambda^*(\xi_{out},\xi_{relay})$, is concave, increasing 
and Lipschitz continuous in $\xi_{out}\geq 0$, $\xi_{relay} \geq 0$.
\end{theorem}

\begin{proof}
 See Appendix~\ref{appendix:mdp-for-pure-as-you-go-deployment}.
\end{proof}

\begin{theorem}\label{theorem:V-continuous-in-xi}
 $\underline{V}^*=(V^*(1),V^*(2), \cdots, V^*(B))$ is Lipschitz continuous in $(\xi_{out},\xi_{relay})$.
\end{theorem}
\begin{proof}
 See Appendix~\ref{appendix:mdp-for-pure-as-you-go-deployment}.
\end{proof}

\begin{theorem}\label{theorem:outage_decreasing_with_xio_placement_rate_decreasing_with_xir}
 For a given $\xi_{out}$, the mean number of relays per step  under 
Algorithm~\ref{algorithm:optimal-policy-structure},  
$\frac{1}{\overline{U}^*(\xi_{out},\xi_{relay})}$, decreases with $\xi_{relay}$. Similarly, 
for a given $\xi_{relay}$, the optimal mean outage per step, $\frac{\overline{Q}_{out}^*(\xi_{out},\xi_{relay})}{\overline{U}^*(\xi_{out},\xi_{relay})}$, 
decreases with $\xi_{out}$.
\end{theorem}

\begin{proof}
 The proof is exactly same as the proof of \cite[Theorem~$5$]{chattopadhyay-etal15measurement-based-impromptu-deployment-arxiv-v1}.
\end{proof}

\subsection{A note on the objective function in \eqref{eqn:unconstrained_problem_average_cost_with_outage_cost}}
\label{subsecion:a-note-on-the-unconstrained-objective-function}
Even though the deployment policy developed in this section uses only the measurements made to the  immediately previous placed node in order 
to make a placement location,  we will see in subsequent sections that measurements to all placed relay nodes located within $B$~steps from the 
current location of the agent will be used for on-line learning of the optimal deployment policy. A question that naturally arises is whether 
we can do better with the additional measurements (when the propagation parameters are known and the optimal policy can be computed 
prior to deployment); this might require skipping some already placed relay nodes after the deployment is over. The 
possibility of relay skipping was considered in 
\cite{chattopadhyay-etal13measurement-based-impromptu-placement_wiopt}; in the current paper, 
we briefly describe a similar formulation in our context and 
explain why we rule out the possibility of relay skipping.

Let us consider deployment up to $x$ steps. 
After the deployment is over, we construct 
a directed acyclic graph over the deployed nodes (including the sink) as follows.  
Links are all directed edges
from each node to every node with smaller index and located within a distance of $B$~steps.  Hence, if $i$ and
$j$ are two nodes with $i>j$ and $\sum_{k=j+1}^i U_k \leq B$, there is a link $(i,j)$ between
them.  Consider all directed acyclic paths from node $N_x$  
to the sink over this graph.  Let us denote by $\mathbf{p}$ any arbitrary directed
acyclic path, and by
$\mathcal{E}(\mathbf{p})$ the set of (directed) links of the path
$\mathbf{p}$.  We also define $\mathcal{P}_x:=\{\mathbf{p}:(i,j) \in \mathcal{E}(\mathbf{p}) \implies N_x \geq i>j \geq 0, \sum_{k=j+1}^i U_k \leq B \}$. 
Let us denote a generic link (edge) on this graph by $e$, and the transmit power and outage probability on edge~$e$ by 
$\Gamma^{(e)}$ and $Q_{out}^{(e)}$. 

Let us consider the following problem: 

\small
\begin{eqnarray}
 && \min_{\pi \in \Pi} \lim \sup_{x \rightarrow \infty}  \nonumber\\
 &&\frac{ \mathbb{E}_{\pi} \bigg( \min_{\mathbf{p} \in \mathcal{P}_x} \sum_{e \in \mathcal{E}(\mathbf{p})} \bigg( \Gamma^{(e)}+ \xi_{out} Q_{out}^{(e)} \bigg) + \xi_{relay} N_x \bigg) }{x} \nonumber\\
 && \label{eqn:shortest-path-problem}
\end{eqnarray}
\normalsize
We call $\sum_{e \in \mathcal{E}(\mathbf{p})} \bigg( \Gamma^{(e)}+ \xi_{out} Q_{out}^{(e)} \bigg)$ the
length of the path $\mathbf{p}$, and $\min_{\mathbf{p} \in
  \mathcal{P}_x} \sum_{e \in \mathcal{E}(\mathbf{p})} \bigg( \Gamma^{(e)}+ \xi_{out} Q_{out}^{(e)} \bigg)$ the
length of the shortest path. 

Formulation of problem~\eqref{eqn:shortest-path-problem} as an MDP will require as the typical state the distance of all nodes located within 
$B$~steps from the current location, the realization of shadowing to all these nodes (through the measured outage probabilities), and 
the lengths of the shortest paths from all these nodes to the sink. A similar situation was considered in 
\cite{chattopadhyay-etal13measurement-based-impromptu-placement_wiopt}. It turns out that the state space becomes very large 
(the number of all possible lengths of shortest paths grows to $\infty$ as $x \rightarrow \infty$, even when the set 
$\mathcal{W}$ of possible values of shadowing is finite), and the 
policy computation becomes numerically intensive; but the numerical results of \cite{chattopadhyay-etal13measurement-based-impromptu-placement_wiopt} 
show that the margin of performance improvement achieved via this formulation (instead of the formulation used 
earlier in this section) is not significant. Hence, in this paper, we only consider 
formulation~\eqref{eqn:unconstrained_problem_average_cost_with_outage_cost} and proceed with it.

\section{OptAsYouGoLearning: Learning with  Deployment for Given  Multipliers}
\label{section:learning-for-pure-as-you-go-deployment-given-xi}

Note that, for any  given values of $\xi_{out}$ and $\xi_{relay}$, the optimal policy  given by Algorithm~\ref{algorithm:optimal-policy-structure} can be completely specified 
by the vector $\underline{V}^*$. 
But, the computation of $\underline{V}^*$ requires the agent to solve a system of nonlinear equations (which is 
computationally intensive), 
and these nonlinear equations can be specified only when the channel model parameters (e.g., path-loss 
exponent $\eta$ and  standard deviation $\sigma$ for log-normal shadowing) are 
known apriori. However, in practice, these parameters may not be available prior to deployment. Under this situation, the 
deployment agent has to {\em learn} the optimal policy as deployment progresses, and  
use the corresponding updated policy at each step to make a placement decision. In order 
to address this requirement, we propose 
an algorithm which will maintain a running estimate of $\underline{V}^*$, and update this estimate at each step 
(using new measurements made at each step). Using   
the theory of Asynchronous Stochastic Approximation (see \cite{bhatnagar11borkar-meyn-theorem-asynchronous-stochastic-approximation}), 
we show that, as the number of deployed relays goes to infinity, the running estimate  
converges to $\underline{V}^*$ almost surely.  
From \eqref{eqn:average_cost_optimality_equation_no_backtracking-V-equation-no-lambda} 
(and the notation defined immediately after \eqref{eqn:average_cost_optimality_equation_no_backtracking-V-equation-no-lambda}), we see that 
the optimal $\underline{V}^*$ is the unique real zero of the system of equations: 
$\mathbb{E}_{W} \min \{c(r,W), V(r+1)-V(1)\}-V(r)=0$ for $r \leq B-1$ and $\mathbb{E}_{W} \, c(B,W)-V(B)=0$. 
We use asynchronous stochastic approximation so that the iterates $\{\underline{V}^{(k)}\}_{k \geq 0}$ converge asymptotically to this unique zero.

\subsection{OptAsYouGoLearning Algorithm}\label{subsection:OptAsYouGoLearning_algorithm}

Suppose that the deployment agent is standing $k$~steps away from the sink node. 
At the $k$-th step, the agent makes a placement decision and then performs a learning operation.  
Let us recall the deployment process (see 
Section~\ref{subsection:deployment-process} and Figure~\ref{fig:pure-as-you-go-learning}) 
and notation: $\mathcal{I}_k:=\{ r \in \{1,2,\cdots,B \}: \text{a relay was placed at a distance $(k-r)\delta$ from the sink} \}\}$. 
For the learning operation, 
$\mathcal{I}_k \subset \{1, \cdots, B\}$ denotes the set of the 
values of $r$ for which links from the current location to the placed relay $r$ 
steps backwards are measured, and for which $V(r)$ is updated, when the agent is at a distance $k \delta$ from the sink. 
Clearly, for each $k \geq 1$, $\mathcal{I}_k$  is a random set. 
Let us denote by $\underline{V}^{(k)}$ the estimate 
of $\underline{V}^*$ after an update (i.e., a learning operation) is made at the $k$-th step from the sink. 
 At step $k$ (after a placement decision is made), $V^{(k-1)}(r)$ for 
$r \in \mathcal{I}_k$ is updated to $V^k(r)$, and it is not updated for $r \notin \mathcal{I}_k$ 
(which means that $V^{(k)}(r)=V^{(k-1)}(r)$ for $r \notin \mathcal{I}_k$). 
Let us define $\nu(r,k):=\sum_{i=1}^k \mathbb{I} \{ r \in \mathcal{I}_i \}$ the number of times the estimate of $V^*(r)$ 
is updated up to the $k$-th step. 

Note that, Algorithm~\ref{algorithm:optimal-policy-structure} 
requires the agent to measure link quality only to the previous node, whereas the learning algorithm presented in this section  
involves link quality measurement to more than one previous nodes (unlike our prior paper \cite{chattopadhyay-etal15measurement-based-impromptu-deployment-arxiv-v1}). {\em This is necessary because, if we make measurement only to last relay, then, depending on the 
initial estimate   $\underline{V}^{(0)}$, there could arise a situation that the inter-relay distance never equals to  
$B$~steps in the entire deployment process, which implies that $V^{(0)}(B)$ will never be updated, thereby converging 
to an unintended policy. Making measurements to all previously placed nodes located at distance less than 
$B\delta$ from the current location ensures that $\liminf_{k \rightarrow \infty}\frac{\nu(r,k)}{k}>0$ almost surely, which is 
required for the convergence proof.}

The OptAsYouGoLearning algorithm is provided in Algorithm~\ref{algorithm:OptAsYouGoLearning}.

\begin{algorithm}[t!]
\hrule
{\bf Input:} $\xi_{out}$, $\xi_{relay}$, and a decreasing positive sequence $\{a(n)\}_{n \geq 1}$ 
such that $\sum_{n=1}^{\infty} a(n)=\infty$, $\sum_{n=1}^{\infty} a^2(n) < \infty$.\\
{\bf Output:} Placement decision at each step.\\
{\bf Initialization:} $r'=1$ (distance from the previous node), $k=1$ (distance of the current location from the sink), initial estimate  $\underline{V}^{(0)}$.\\
 \While{$1 \leq r' \leq B$}{
 Find $\mathcal{I}_k:=\{ r \in \{1,2,\cdots,B \}: \text{relay placement at $(k-r)\delta$ distance from sink} \}\}$;\\
 Find $\nu(r,k):=\sum_{i=1}^k \mathbb{I} \{ r \in \mathcal{I}_i \} \forall r \in \{1,2,\cdots,B\}$ ;\\
 Measure  $Q_{out}(r,\gamma,w_r) \forall {\gamma \in \mathcal{S}}, r \in \mathcal{I}_k$;\\
  \uIf{$r' \leq B-1$ and $\min_{\gamma \in \mathcal{S}}(\gamma+ \xi_{out} Q_{out}(r',\gamma, w_{r'}))+\xi_{relay} \leq -V^{(k-1)}(1)+V^{(k-1)}(r'+1)$}
{Place a new relay and use transmit power $\arg \min_{\gamma \in \mathcal{S}}(\gamma+ \xi_{out} Q_{out}(r',\gamma, w_{r'}))$;\\
Do the following updates:

\small
\begin{eqnarray}
 && V^{(k)}(r)\nonumber\\
 &=&V^{(k-1)}(r)+ a(\nu(r,k)) \mathbb{I}\{r \in \mathcal{I}_k\}  \bigg[ \min \bigg \{ \min_{\gamma}(\gamma+ \nonumber\\
&& \xi_{out} Q_{out}(r,\gamma, w_r)) +\xi_{relay}, -V^{(k-1)}(1) \nonumber\\
&& +V^{(k-1)}(r+1) \bigg \} -V^{(k-1)}(r) \bigg], \forall 1 \leq r \leq B-1 \nonumber\\
&& V^{(k)}(B) \nonumber\\
&=&V^{(k-1)}(B)+ a(\nu(B,k)) \mathbb{I}\{B \in \mathcal{I}_k\}  \bigg[  \min_{\gamma}(\gamma+ \nonumber\\
&& \xi_{out} Q_{out}(B,\gamma, w_B))  +\xi_{relay}-V^{(k-1)}(B) \bigg] 
\label{eqn:learning_no_backtracking_given_xio_xir_update_part}
\end{eqnarray}
\normalsize

   Move to next step and set $r'=1$;}
    \uElseIf{$r' \leq B-1$ and $\min_{\gamma \in \mathcal{S}}(\gamma+ \xi_{out} Q_{out}(r',\gamma, w_{r'}))+\xi_{relay} > -V^{(k-1)}(1)+V^{(k-1)}(r'+1)$}
    {Do not place, do the same updates as \eqref{eqn:learning_no_backtracking_given_xio_xir_update_part};\\
    Move to next step and do $r'=r'+1$;\\
    }
    \Else
    {Place a new relay (since $r'=B$);\\
    Use transmit power $\arg \min_{\gamma \in \mathcal{S}}(\gamma+ \xi_{out} Q_{out}(B,\gamma, w_B))$;\\
        Do the same updates as \eqref{eqn:learning_no_backtracking_given_xio_xir_update_part};\\
    Move to next step and set $r'=1$.
    }
    k=k+1;
 }
\hrule
\caption{OptAsYouGoLearning Algorithm}
\label{algorithm:OptAsYouGoLearning}
\end{algorithm}

\begin{theorem}\label{theorem:OptAsYouGoLearning}
 Under Algorithm~\ref{algorithm:OptAsYouGoLearning}, $V^{(k)}(r) \rightarrow V^*(r)$ almost surely for all $1 \leq r \leq B$.
\end{theorem}
\begin{proof}
 See Appendix~\ref{appendix:learning-for-pure-as-you-go-deployment-given-xi}.
\end{proof}

{\bf Discussion of Algorithm~\ref{algorithm:OptAsYouGoLearning}:}
\begin{enumerate}[label=(\roman{*})]
\item {\em The basic idea:} From \eqref{eqn:average_cost_optimality_equation_no_backtracking-V-equation-no-lambda} 
(and the notation defined immediately after \eqref{eqn:average_cost_optimality_equation_no_backtracking-V-equation-no-lambda}), we see that 
the optimal $\underline{V}^*$ is the unique real zero of the system of equations: 
$\mathbb{E}_{W} \min \{c(r,W), V(r+1)-V(1)\}-V(r)=0$ for $r \leq B-1$ and $\mathbb{E}_{W} \, c(B,W)-V(B)=0$. 
We use asynchronous stochastic approximation so that the iterates converge asymptotically to this unique zero.
\item {\em Asynchronous stochastic approximation:} In standard stochastic approximation techniques, all iterates are updated at the same time. 
However, the pure as-you-go deployment scheme does not allow the deployment agent to update all iterates at each step. Since only a subset 
$\mathcal{I}_k \subset \{1, \cdots, B\}$ of iterates can be updated at step $k$, we have to use {\em asynchronous} stochastic approximation. 
\item The proof of Theorem~\ref{theorem:OptAsYouGoLearning} exhibits a nice separation between the estimation and control. In 
other words, the iterates will asymptotically converge to $\underline{V}^*$ (and the policy will 
converge to the optimal policy) even when the placement decisions are not made according 
to the proposed threshold policy (but the measurement and update scheme should be unchanged); but it may not yield the optimal 
cost for problem~\eqref{eqn:unconstrained_problem_average_cost_with_outage_cost} since we do not use the 
optimal policy at each stage. However, this nice separation property 
will not hold in next section when we vary $\xi_{out}$ and $\xi_{relay}$ in order to solve the constrained problem 
\eqref{eqn:constrained_problem_average_cost_with_outage_cost}. 
 \item Note that, since the state space of the MDP in Section~\ref{section:mdp-for-pure-as-you-go-deployment} is large 
(potentially infinite and even uncountable), it will not be easy to use traditional Q-learning algorithms. In fact, all 
the state action-pairs in a Q-learning algorithm need to repeat comparably often over infinite time horizon 
to guarantee the desired convergence, but this may not happen in case of infinite 
state space (arising out of infinite $\mathcal{W}$). On the other hand, Algorithm~\ref{algorithm:OptAsYouGoLearning} provides a learning 
algorithm with provable convergence guarantee while having only $B$ number of iterates.
\end{enumerate}

\section{OptAsYouGoAdaptiveLearning for the Constrained problem}
\label{section:learning-for-pure-as-you-go-deployment-constrained-problem}

In Section~\ref{section:learning-for-pure-as-you-go-deployment-given-xi}, 
we provided a deploy-and-learn algorithm 
for given   $\xi_{out}$ and $\xi_{relay}$. 
However, Theorem~\ref{theorem:how-to-choose-optimal-Lagrange-multiplier} tells us how to choose the 
Lagrange multipliers $\xi_{out}$ and $\xi_{relay}$ (if they exist) 
in (\ref{eqn:unconstrained_problem_average_cost_with_outage_cost}) 
in order to solve the constrained problem (\ref{eqn:constrained_problem_average_cost_with_outage_cost}). 
But we need to know the radio propagation parameters (e.g., $\eta$ and $\sigma$) in order to 
compute a pair $(\xi_{out}^*, \xi_{relay}^*)$ that satisfies the condition 
given in Theorem~\ref{theorem:how-to-choose-optimal-Lagrange-multiplier}. In practice,  
these parameters may not be known.  
Hence, we provide a sequential placement algorithm such that, as deployment progresses, 
the placement policy (updated at each step) converges to the set of optimal policies for the constrained problem 
(\ref{eqn:constrained_problem_average_cost_with_outage_cost}). We modify the 
OptAsYouGoLearning algorithm so that a running estimate $(\underline{V}^{(k)},\xi_{out}^{(k)},\xi_{relay}^{(k)})$ gets 
updated at each step, and asymptotically converges to 
the set of optimal $(\underline{V}^*(\xi_{out},\xi_{relay}),\xi_{out},\xi_{relay})$ tuples. 
This algorithm is based on two time-scale stochastic approximation 
(see \cite[Chapter~$6$]{borkar08stochastic-approximation-book}).

\subsection{Some Useful Notation and Assumptions}
In this subsection, we will introduce some assumptions and notation (these were provided in 
\cite[Section~VII]{chattopadhyay-etal15measurement-based-impromptu-deployment-arxiv-v1}, but are repeated here for completeness).
\begin{definition}
We denote  by $\gamma^*$ the optimal mean power per step for problem 
(\ref{eqn:constrained_problem_average_cost_with_outage_cost}), for a given constraint pair 
$(\overline{q},\overline{N})$. The set 
$\mathcal{K}(\overline{q},\overline{N})$ is defined as follows:

\footnotesize
\begin{eqnarray*}
&& \mathcal{K}(\overline{q},\overline{N}) := \bigg\{(\underline{V}^*(\xi_{out},\xi_{relay}),\xi_{out},\xi_{relay}): \\
&& \frac{\overline{\Gamma}^*(\xi_{out},\xi_{relay})}{\overline{U}^*(\xi_{out},\xi_{relay})}=\gamma^* , 
\frac{ \overline{Q}_{out}^*(\xi_{out},\xi_{relay}) }{ \overline{U}^*(\xi_{out},\xi_{relay}) } \leq \overline{q} \\
&& \frac{1}{ \overline{U}^*(\xi_{out},\xi_{relay})} \leq \overline{N}, 
\xi_{out} \geq 0, \xi_{relay} \geq 0 \bigg\}
\end{eqnarray*}
\normalsize
\qed 
\end{definition}

Note that, the pair $(\overline{q},\overline{N})$ can be infeasible. For  
example, if $\overline{N}=\frac{1}B$ (i.e., inter-node distance is  $B$)  
and  $\overline{q}< \frac{\mathbb{E}_W Q_{out}(B,P_M,W)}B$ ($P_M$ is the 
maximum available transmit power),  the outage constraint cannot be satisfied 
while meeting the constraint on the mean number of relays per step, even by using the  
maximum transmit power $P_M$.

$\mathcal{K}(\overline{q},\overline{N})$ is empty if  
$(\overline{q},\overline{N})$ is infeasible. 
{\em In this paper, we assume that $\mathcal{K}(\overline{q},\overline{N})$ is non-empty (i.e., 
$(\overline{q},\overline{N})$ is a feasible pair), which is true for feasible pairs of $\mathcal{K}(\overline{q},\overline{N})$:} 

\begin{assumption}\label{assumption:existence_of_xio_xir}
$\overline{q}$ and $\overline{N}$ are such that there exists 
at least one pair $ \xi_{out}^* \geq 0, \xi_{relay}^* \geq 0$ such that   
$(\underline{V}^*(\xi_{out}^*,\xi_{relay}^*),\xi_{out}^*,\xi_{relay}^*) \in \mathcal{K}(\overline{q},\overline{N})$.\qed
\end{assumption}

\begin{assumption}\label{assumption:shadowing_continuous_random_variable}
 The probability density 
function (p.d.f.) of the shadowing random variable $W$ is continuous  over $(0,\infty)$; i.e.,  
$\mathbb{P}(W=w)=0$ for any $w \in (0,\infty)$ (e.g., log-normal shadowing).\qed
\end{assumption}

\begin{theorem}\label{theorem:placement_rate_mean_outage_per_step_continuous_in_xio_and_xir}
Under Assumption~\ref{assumption:shadowing_continuous_random_variable} and 
Algorithm~\ref{algorithm:optimal-policy-structure}, 
the optimal mean power per step $\frac{\overline{\Gamma}^*(\xi_{out},\xi_{relay})}{\overline{U}^*(\xi_{out},\xi_{relay})}$, 
the optimal mean placement rate $\frac{1}{\overline{U}^*(\xi_{out},\xi_{relay})}$ 
and the optimal mean outage per step $\frac{\overline{Q}_{out}^*(\xi_{out},\xi_{relay})}{\overline{U}^*(\xi_{out},\xi_{relay})}$, 
are continuous in $(\xi_{out},\xi_{relay})$.
\end{theorem}
\begin{proof}
  See Appendix~\ref{appendix:learning-for-pure-as-you-go-deployment-constrained-problem}.
\end{proof}

{\em Remark:} Theorem~\ref{theorem:placement_rate_mean_outage_per_step_continuous_in_xio_and_xir} 
implies that there is no need to do any randomization among deterministic policies 
(unlike \cite{ma-makowski88steering-policies-recurrence-condition})  in order to meet the constraints with equality.

\subsection{OptAsYouGoAdaptiveLearning Algorithm}\label{subsection:OptAsYouGoAdaptiveLearning-algorithm-description-discussion}
The basic idea behind this algorithm (Algorithm~\ref{algorithm:OptAsYouGoAdaptiveLearning}; see next page) is to vary $\xi_{out}^{(k)}$ and $\xi_{relay}^{(k)}$ at a much slower rate than $\underline{V}^{(k)}$, as if $\xi_{out}^{(k)}$ and $\xi_{relay}^{(k)}$ are varied in an outer loop and $\underline{V}^{(k)}$ is varied in an inner loop. If the outage in a newly created link is larger than the budgeted outage for a link with that length, then $\xi_{out}$ is increased with the hope that subsequent links will have smaller outage; the opposite is done in case the outage in a newly created link is smaller. On the other hand, if a newly created link is shorter than $\frac{1}{\overline{N}}$, then $\xi_{relay}$ is increased, otherwise it is decreased.

{\em Notation in Algorithm~\ref{algorithm:OptAsYouGoAdaptiveLearning}:} 
$\Lambda_{[0,A_1]}(x)$ denotes the projection of $x$ on the interval $[0,A_1]$. 
Let the power, outage and link length of the new relay (if placed) at the $k$-th step be $\Gamma_{N_k}$, $Q_{out}^{(N_k,N_k-1)}$ and $U_{N_k}$ 
(recall that $N_k$ is the number of nodes placed up to the $k$-th step). Note that,  
$\mathbb{I}\{ N_k=N_{k-1}+1 \}$ is the indicator that a relay is placed at the $k$-th step.

\begin{algorithm}[h!]
\hrule
{\bf Input:} Two positive numbers $A_1$ and $A_2$ appropriately chosen, two decreasing positive sequences $\{a(n)\}_{n \geq 1}$ and $\{b(n)\}_{n \geq 1}$  such that $\sum_{n=1}^{\infty} a(n)=\infty$, $\sum_{n=1}^{\infty} a^2(n) < \infty$, 
$\sum_{n=1}^{\infty} b(n)=\infty$, $\sum_{n=1}^{\infty} b^2(n) < \infty$ 
and $\lim_{n \rightarrow \infty}\frac{b(\lfloor \frac{n}{B} \rfloor)}{a(n)}=0$.  \\
{\bf Output:} Placement decision at each step.\\
{\bf Initialization:}  $r'=1$ (distance from the previous node), $k=1$ (distance of the current location from the sink), initial estimates  $\underline{V}^{(0)}$, $\xi_{out}^{(0)}$, $\xi_{relay}^{(0)}$.\\
 \While{$1 \leq r' \leq B$}{
 Find $\mathcal{I}_k:=\{ r \in \{1,2,\cdots,B \}: \text{relay placed at  $(k-r)\delta$ distance from  sink} \}\}$;\\
 Find $\nu(r,k):=\sum_{i=1}^k \mathbb{I} \{ r \in \mathcal{I}_i \} \forall r \in \{1,2,\cdots,B\}$;\\
 Measure  $Q_{out}(r,\gamma,w_r) \forall {\gamma \in \mathcal{S}}, r \in \mathcal{I}_k$;\\
  \uIf{$r' \leq B-1$ and $\min_{\gamma \in \mathcal{S}}(\gamma+ \xi_{out}^{(k-1)} Q_{out}(r',\gamma, w_{r'}))+\xi_{relay}^{(k-1)} \leq -V^{(k-1)}(1)+V^{(k-1)}(r'+1)$}
{Place a new relay and use transmit power $\arg \min_{\gamma \in \mathcal{S}}(\gamma+ \xi_{out}^{(k-1)} Q_{out}(r',\gamma, w_{r'}))$;\\
Do the following updates:
\small
\begin{eqnarray}
&& V^{(k)}(r)=V^{(k-1)}(r)+ a(\nu(r,k)) \mathbb{I}\{r \in \mathcal{I}_k, r <B\} \nonumber\\
&& \bigg[ \min \bigg \{ \min_{\gamma}(\gamma+  \xi_{out}^{(k-1)} Q_{out}(r,\gamma, w_r)) +\xi_{relay}^{(k-1)},  \nonumber\\
&& -V^{(k-1)}(1)+V^{(k-1)}(r+1) \bigg \} -V^{(k-1)}(r) \bigg] \nonumber\\
 && V^{(k)}(B) =V^{(k-1)}(B)+ a(\nu(B,k)) \mathbb{I}\{B \in \mathcal{I}_k\}  \nonumber\\
 && \bigg[  \min_{\gamma}(\gamma+ \xi_{out}^{(k-1)} Q_{out}(B,\gamma, w_B))  +\xi_{relay}^{(k-1)} \nonumber\\
 && -V^{(k-1)}(B) \bigg] \nonumber\\
&& \xi_{out}^{(k)} = \bigg[\xi_{out}^{(k-1)}+b({N_k}) \mathbb{I}\{ N_k=N_{k-1}+1 \} \nonumber\\
&&  \bigg( Q_{out}^{(N_k,N_k-1)}-\overline{q}U_{N_k} \bigg) \bigg]_{0}^{A_1} \nonumber\\
&& \xi_{relay}^{(k)} =\bigg[\xi_{relay}^{(k-1)}+b({N_k}) \mathbb{I}\{ N_k=N_{k-1}+1 \} \nonumber\\
&& \bigg( 1-\overline{N}U_{N_k} \bigg) \bigg]_{0}^{A_2} 
\label{eqn:OptAsYouGoAdaptiveLearning_update_part}
\end{eqnarray}
\normalsize

   Move to next step and set $r'=1$;}
    \uElseIf{$r' \leq B-1$ and $\min_{\gamma \in \mathcal{S}}(\gamma+ \xi_{out}^{(k-1)} Q_{out}(r',\gamma, w_{r'}))+\xi_{relay}^{(k-1)} > -V^{(k-1)}(1)+V^{(k-1)}(r'+1)$}
    {Do not place, and perform updates as in \eqref{eqn:OptAsYouGoAdaptiveLearning_update_part};\\
    Move to next step and set $r'=r'+1$;\\
    }
    \Else
    {Place a new relay (since $r'=B$);\\
    Use power $\arg \min_{\gamma \in \mathcal{S}}(\gamma+ \xi_{out}^{(k-1)} Q_{out}(B,\gamma, w_B))$;\\
        Do the same updates as \eqref{eqn:OptAsYouGoAdaptiveLearning_update_part};\\
    Move to next step and set $r'=1$.
    }
    k=k+1;.
 }
\hrule
\caption{OptAsYouGoAdaptiveLearning}
\label{algorithm:OptAsYouGoAdaptiveLearning}
\end{algorithm}

\begin{theorem}\label{theorem:convergence_OptAsYouGoAdaptiveLearning}
 Under Assumption~\ref{assumption:existence_of_xio_xir}, Assumption~\ref{assumption:shadowing_continuous_random_variable}   
 and under proper choice of 
$A_1$ and $A_2$, we have 
 $(\underline{V}^{(k)}, \xi_{out}^{(k)}, \xi_{relay}^{(k)}) \rightarrow \mathcal{K}(\overline{q},\overline{N})$ 
 almost surely for 
Algorithm~\ref{algorithm:OptAsYouGoAdaptiveLearning}.
\end{theorem}
\begin{proof}
 See Appendix~\ref{appendix:learning-for-pure-as-you-go-deployment-constrained-problem}.
 
 We complete the proof in four steps. 
First, we show  
that the difference between  $\underline{V}^{(k)}$ and $\underline{V}^*(\xi_{out}^{(k)},\xi_{relay}^{(k)})$ 
converges to $0$ almost surely. This proves  the desired convergence in the faster timescale. 
Next, we pose the slower timescale iteration  as a projected stochastic approximation iteration 
(see \cite[Equation~$5.3.1$]{kushner-clark78SA-constrained-unconstrained}).  
Next, we show that the slower timescale iteration satisfies some conditions 
given in \cite{kushner-clark78SA-constrained-unconstrained} (see   
\cite[Theorem~$5.3.1$]{kushner-clark78SA-constrained-unconstrained}). 
Finally, we argue  
(using Theorem~$5.3.1$ of \cite{kushner-clark78SA-constrained-unconstrained}) 
that the slower timescale iterates converge to the set of stationary points of a suitable ordinary differential 
equation.
 
 It is to be noted that while the proof to some extent follows the outline of the proof of 
\cite[Theorem~$12$]{chattopadhyay-etal15measurement-based-impromptu-deployment-arxiv-v1}, 
significantly new nontrivialities arise in our work as compared to the proof of 
\cite[Theorem~$12$]{chattopadhyay-etal15measurement-based-impromptu-deployment-arxiv-v1}. For example, we had to prove 
the boundedness of the faster timescale iterates separately, since the asynchronous updates in the faster timescale do 
not allow us to mimic the proof of \cite[Theorem~$12$]{chattopadhyay-etal15measurement-based-impromptu-deployment-arxiv-v1}. 
Similarly there are many other steps which require significant novel additional mathematical analysis compared to 
\cite[Theorem~$12$]{chattopadhyay-etal15measurement-based-impromptu-deployment-arxiv-v1}. Hence, in this proof, we 
proved intermediate results wherever necessary, and skipped some steps if they follow from the proof of 
\cite[Theorem~$12$]{chattopadhyay-etal15measurement-based-impromptu-deployment-arxiv-v1}.
 \end{proof}

 {\bf Choice of $A_1$ and $A_2$:} 
$A_1$ and $A_2$ need to be chosen carefully, otherwise the iterates $(\xi_{out}^{(k)}, \xi_{relay}^{(k)})$ 
may converge to undesired points on the boundary of $[0,A_1] \times [0,A_2]$. In general, a stationary point on the boundary of $[0,A_1] \times [0,A_2]$ 
may not correspond to a point in $\mathcal{K}(\overline{q},\overline{N})$. 
Hence, we borrow a scheme from 
\cite{chattopadhyay-etal15measurement-based-impromptu-deployment-arxiv-v1} 
for choosing $A_1$ and $A_2$ which ensures that, if $(\xi_{out}',\xi_{relay}')$ 
is a stationary point of the o.d.e., then 
$(\underline{V}^*(\xi_{out}',\xi_{relay}'),\xi_{out}',\xi_{relay}') \in \mathcal{K}(\overline{q},\overline{N})$. 
The number $A_1$ has to be chosen so large 
that, for all $u \in \{1,2,\cdots,B\}$, we will have 
$\mathbb{P}(\argmin_{\gamma \in \mathcal{S}}(\gamma+A_1 Q_{out}(u,\gamma,W))=P_M)>1-\kappa$ 
for some small enough $\kappa>0$. We also need the condition that   
$\frac{\overline{Q}_{out}^*(A_1,0)}{\overline{U}^*(A_1,0)} \leq \overline{q}$. 
The number $A_2$ has to be chosen so large that,  
for any $\xi_{out} \in [0,A_1]$, we will have 
$\overline{U}^*(\xi_{out},A_2) > \frac{1}{\overline{N}}$ (when $\frac{1}{\overline{N}}<B$). 
The numbers $A_1$ and $A_2$ have to be chosen so large that there exists at least one pair 
$(\xi_{out}',\xi_{relay}')$ for which  
$(\underline{V}^*(\xi_{out}',\xi_{relay}'),\xi_{out}',\xi_{relay}') \in \mathcal{K}(\overline{q},\overline{N})$.\qed

{\bf Discussion of Algorithm~\ref{algorithm:OptAsYouGoAdaptiveLearning}:}
\begin{enumerate}[label=(\roman{*})]
\item {\em Two timescales:} The update scheme (\ref{eqn:OptAsYouGoAdaptiveLearning_update_part}) 
is based on two-timescale 
stochastic approximation (see \cite[Chapter~$6$]{borkar08stochastic-approximation-book}). 
Since $\lim_{n \rightarrow \infty}\frac{b(\lfloor \frac{n}{B} \rfloor)}{a(n)}=0$, 
we can say that $\xi_{out}$ and $\xi_{relay}$ are adapted in a 
{\em slower} timescale, and $\underline{V}$ is updated in a {\em faster} timescale, as if 
$\xi_{out}$ and $\xi_{relay}$ are updated in a slow outer loop, and, $\underline{V}$ is updated in an inner loop. 

\item {\em Structure of the iteration:} The slower timescale iteration involves updating $\xi_{out}$ and $\xi_{relay}$ 
based on whether the corresponding constraints are violated in a link (after placing a relay); if a constraint is violated 
by a newly created link, then the corresponding Lagrange multiplier is increased to counterbalance it in subsequent 
node placements. The goal is to meet both constraints with equality (if possible) in the long run.

\item {\em Asymptotic behaviour of the iterates:}
If $\overline{q}> \frac{\mathbb{E}_W Q_{out}(B,P_1,W)}B$;  we 
will have $\xi_{out}^{(k)} \rightarrow 0$; here 
the policy places all the relays at the $B$-th step and uses the smallest power $P_1$ at each node. 
If the constraints are not feasible, then either $\xi_{out}^{(k)} \rightarrow A_1$ 
or $\xi_{relay}^{(k)} \rightarrow A_2$  or both happens.

{\em   Simulation results show that $\mathcal{K}(\overline{q},\overline{N})$ has only one tuple in case the pair 
$(\overline{q},\overline{N})$ is feasible.}\qed
\end{enumerate}

\subsection{Asymptotic Performance of Algorithm~\ref{algorithm:OptAsYouGoAdaptiveLearning}}
\label{subsection:asymptotic_performance_adaptive_learning}

Though Algorithm~\ref{algorithm:OptAsYouGoAdaptiveLearning} 
induces a nonstationary  
policy, Theorem~\ref{theorem:convergence_OptAsYouGoAdaptiveLearning} 
states that the sequence of policies generated by 
Algorithm~\ref{algorithm:OptAsYouGoAdaptiveLearning} converges to the set 
of optimal stationary, deterministic policies 
for the constrained  problem (\ref{eqn:constrained_problem_average_cost_with_outage_cost}). 
Let  $\pi_{oaygal}$ denote the (nonstationary) deployment policy induced by 
Algorithm~\ref{algorithm:OptAsYouGoAdaptiveLearning}.

\begin{theorem}\label{theorem:expected_average_cost_performance_of_optexplorelimadaptivelearning}
Under Assumption~{\ref{assumption:existence_of_xio_xir}}, 
Assumption~\ref{assumption:shadowing_continuous_random_variable} and proper choice of $A_1$ and $A_2$, we have:

\footnotesize
\begin{eqnarray*}
&& \limsup_{x \rightarrow \infty} \frac{\mathbb{E}_{\pi_{oaygal}}\sum_{i=1}^{N_x}\Gamma_i}{x} = \gamma^* \nonumber\\
&& \, \limsup_{x \rightarrow \infty} \frac{\mathbb{E}_{\pi_{oaygal}}\sum_{i=1}^{N_x}Q_{out}^{(i,i-1)}}{x} \leq \overline{q}, \,\,  
\limsup_{x \rightarrow \infty} \frac{\mathbb{E}_{\pi_{oaygal}}N_x}{x} \leq \overline{N} \nonumber\\
\end{eqnarray*}
\normalsize
\end{theorem}

\begin{proof}
 The proof is similar to \cite[Theorem~$13$]{chattopadhyay-etal15measurement-based-impromptu-deployment-arxiv-v1}.
\end{proof}

\begin{figure*}[t]
\begin{minipage}[r]{0.5\linewidth}
\subfigure{
\includegraphics[width=\linewidth, height=6cm]{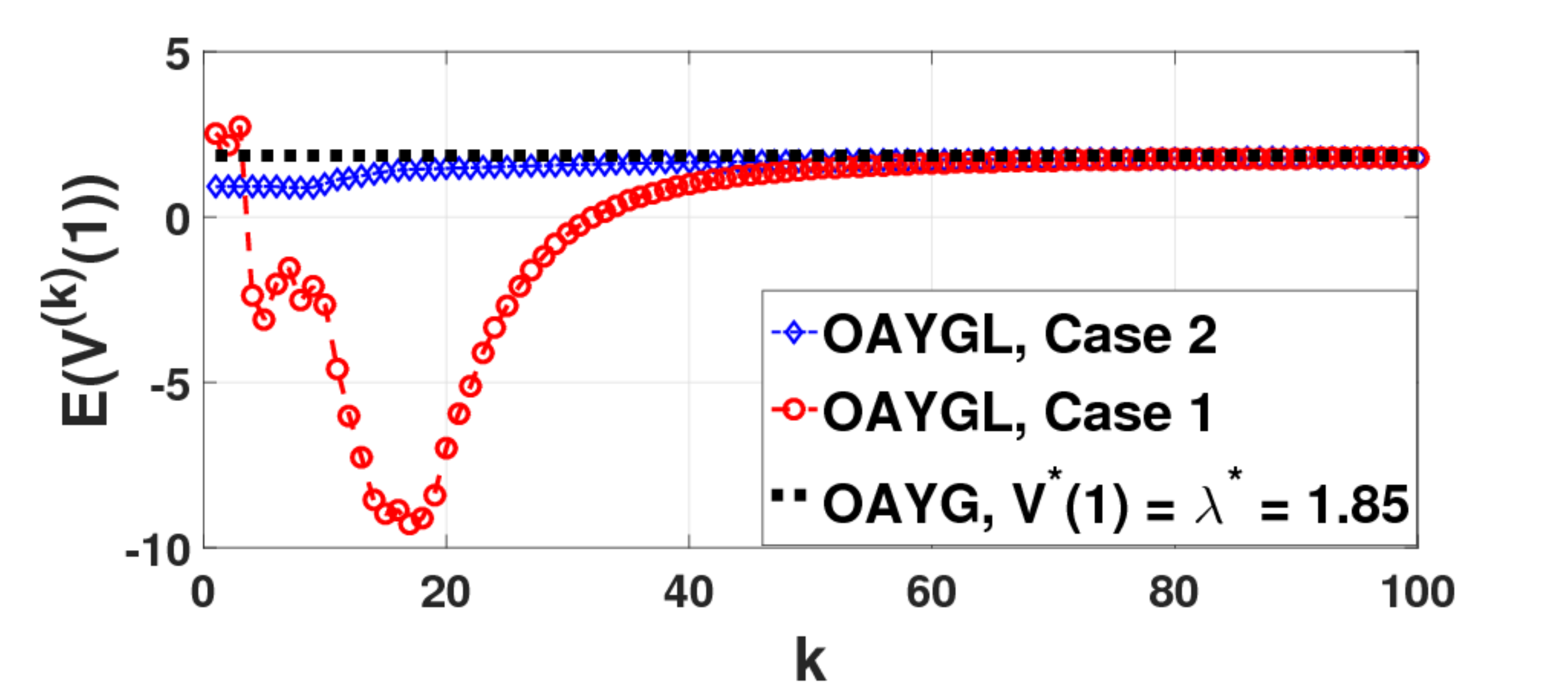}
\includegraphics[width=\linewidth, height=6cm]{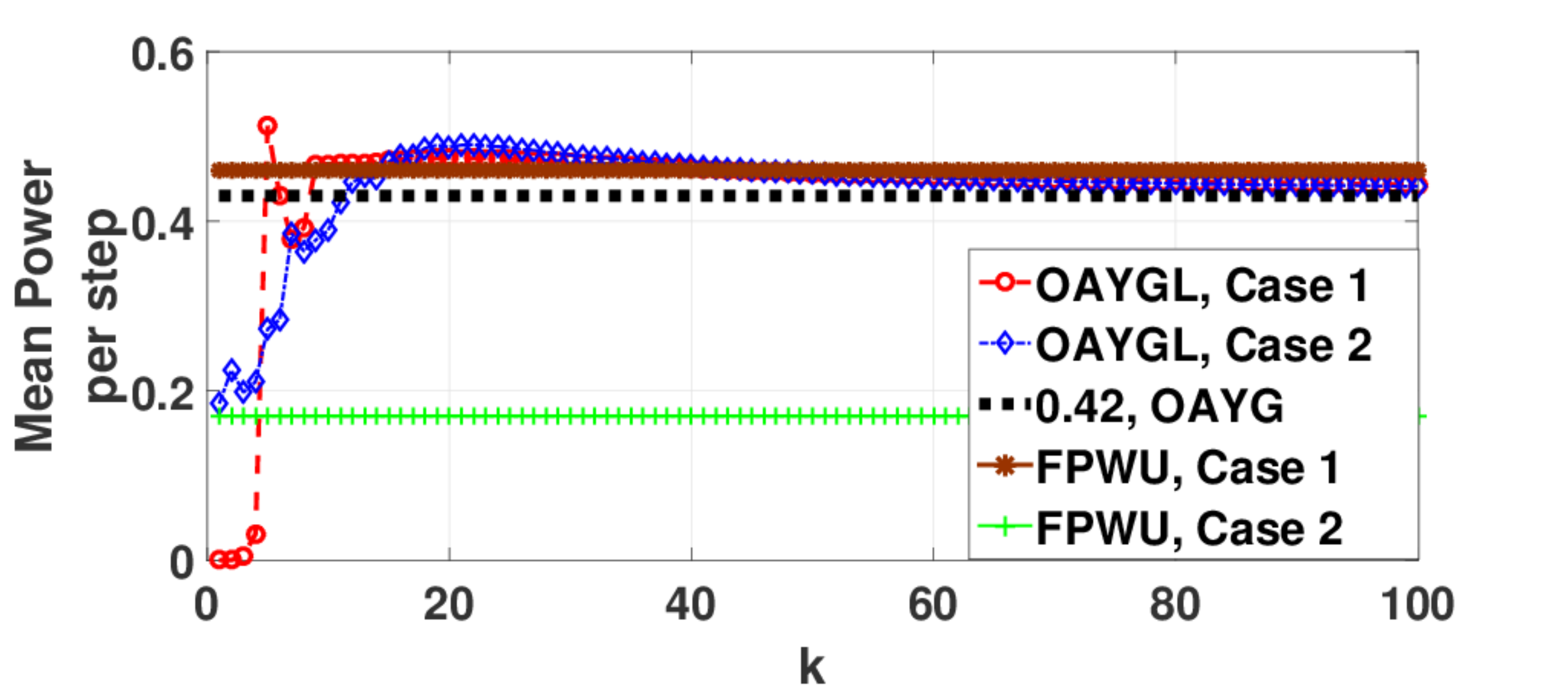}}
\end{minipage}  \hfill
\vspace{-5mm}
\end{figure*}
\begin{figure*}[t]
\begin{minipage}[c]{0.5\linewidth}
\subfigure{
\includegraphics[width=\linewidth, height=6cm]{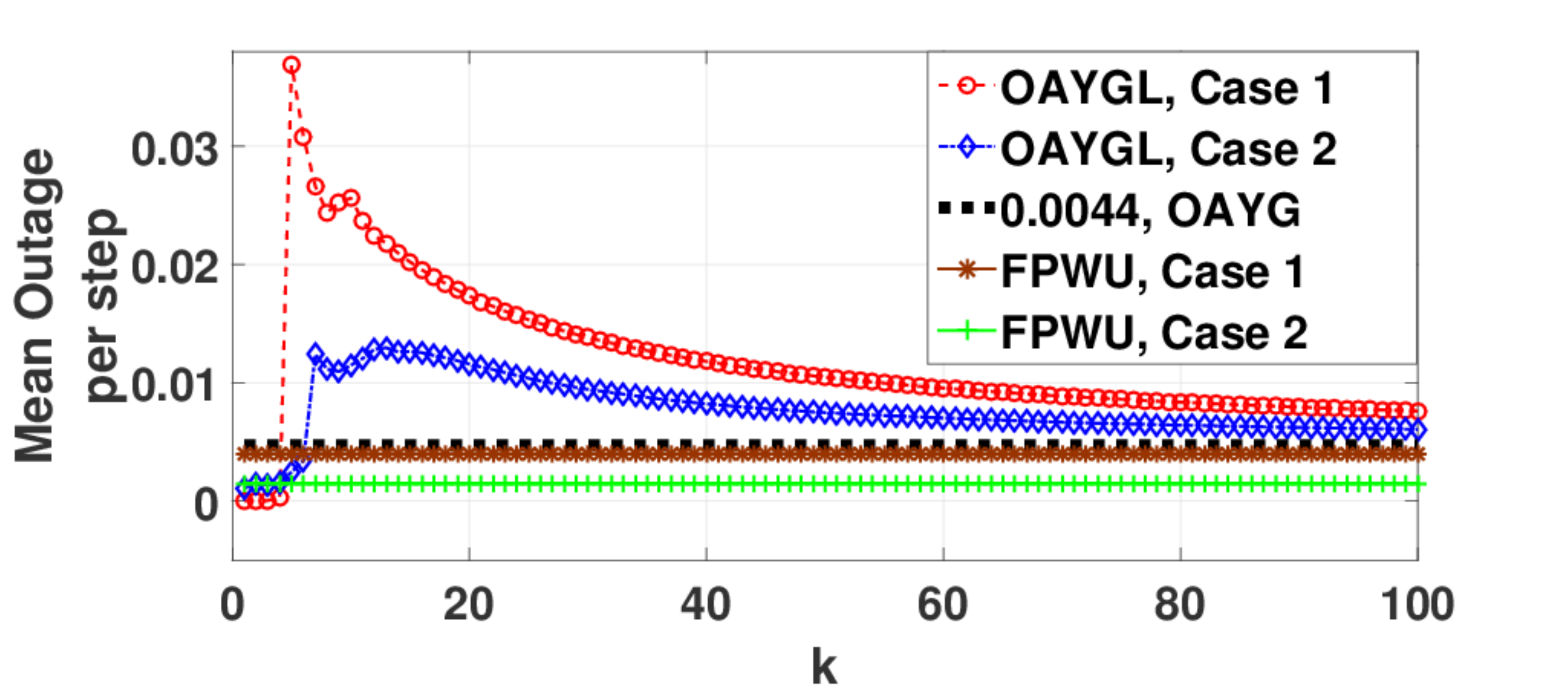}
\includegraphics[width=\linewidth, height=6cm]{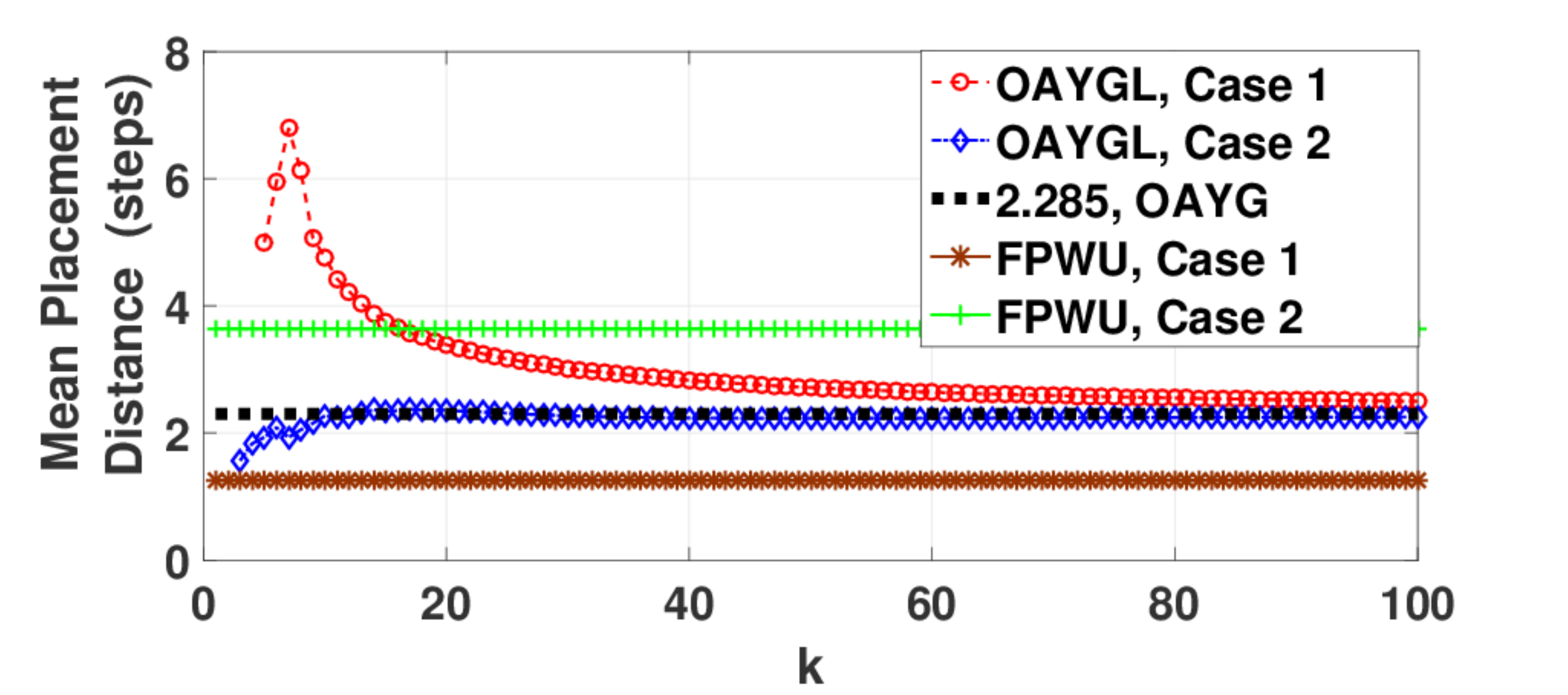}}
\end{minipage} \hfill
\vspace{-3mm}
\caption{Convergence speed of OptAsYouGoLearning (OAYGL) with the number of steps, $k$. 
In the legends, ``OAYG'' refers to the values that are obtained if Algorithm~\ref{algorithm:optimal-policy-structure}  
is used; these are the target values for OptAsYouGoLearning.}
\label{fig:single_timescale_plots}
\vspace{-0mm}
\end{figure*}

\vspace{-3mm}
\section{Convergence Speed of Learning Algorithms: A Simulation Study}\label{section:convergence_speed_learning_algorithms}
\vspace{-1mm}

In this section, we provide a simulation study for the convergence rate of 
Algorithm~\ref{algorithm:OptAsYouGoLearning} 
and  Algorithm~\ref{algorithm:OptAsYouGoAdaptiveLearning}. 

%
%This will give an idea about the performance of these algorithms 
%for deployment over a finite distance.

\vspace{-3mm}
\subsection{Parameter Values  Used in the Simulation}\label{subsection:parameter_values}
\vspace{-1mm}

For simulation, we consider a deployment environment similar to that considered in 
\cite[Section~VIII]{chattopadhyay-etal15measurement-based-impromptu-deployment-arxiv-v1}. The details of the simulation environment are provided below.

We assume that deployment is done with iWiSe motes (\cite{iwise}) equipped with $9$~dBi antennas. 
$\mathcal{S}$, the set of transmit power levels, is taken to be $\{-18,-7,-4,0,5\}$~dBm, which is a subset of available transmit power levels for iWiSe motes. 
Under the channel model as given by  
(\ref{eqn:channel_model}), our measurements 
in a forest-like environment gave 
$\eta=4.7$ and $c=10^{0.17}$ (i.e., $1.7$~dB); the experimental details can be found in 
\cite{chattopadhyay-etal14deployment-experience}. From the statistical analysis of the measurement data, we also showed that 
shadowing $W$  follows log-normal distribution in such a forest-like environment; $W=10^{\frac{Y}{10}}$ with 
$Y \sim \mathcal{N}(0, \sigma^2)$, where $\sigma=7.7$~dB was obtained from our data analysis. Shadowing decorrelation distance was calculated as $6$~meters; hence 
we consider deployment with $\delta=20$~meter. The fading turned out to be Rayleigh fading.

Outage is defined to be the event when  a packet is received at a power level below $P_{rcv-min} = 10^{-9.7}$~mW (i.e., $-97$~dBm); for a commercial implementation of IEEE~$802.15.4$, received power $-97$~dBm   results in a $2\%$ packet loss probability for $127$~byte packets for iWiSe motes (obtained from measurements).

We choose $B$ in the following way. We define a link to be workable if it has an outage probability less than $3\%$. $B$ is chosen to be the largest integer such that the probability of finding a workable link of length $B \delta$ is greater than $20\%$, under $5$~dBm transmit power. For the parameters $\eta=4.7$ and $\sigma=7.7$~dB, and $5$~dBm transmit power, $B$ turned out to be $5$. 

 It is important to note that, the radio propagation parameters (e.g., $\eta$ and $\sigma$) and modeling assumptions (e.g., log-normal shadowing) are obtained and validated using field data collected via extensive measurements in a forest-like environment; the details of these experiments can be found in \cite{chattopadhyay-etal14deployment-experience}. Hence, in this paper, we evaluate our algorithms only via MATLAB simulation of  an environment that has radio propagation model and  parameters obtained from experiments in \cite{chattopadhyay-etal14deployment-experience}. This is done by generating random channel gains in MATLAB, for the wireless links that need to be measured in course of  the deployment process.   
 
 The performance variation of OptAsYouGo algorithm with $(\xi_{out},\xi_{relay})$ has been demonstrated numerically in    \cite[Section~V, Appendix~C]{chattopadhyay-etal15measurement-based-impromptu-deployment-arxiv-v1}, which comply with Theorem~\ref{theorem:lambda-increasing-concave-continuous-in-xi} and  Theorem~\ref{theorem:outage_decreasing_with_xio_placement_rate_decreasing_with_xir}.

\vspace{-3mm}

\subsection{OptAsYouGoLearning for Given  Multipliers}
\label{subsection:convergence_speed_optasyougolearning}
\vspace{-1mm}

Here we study the rate of convergence for OptAsYouGoLearning with $\xi_{out}=125$, $\xi_{relay}=2$. 
Let us assume that the propagation environment,  
in which deployment is being carried out, is characterized 
by the parameters given in Section~\ref{subsection:parameter_values} (i.e., 
$\eta=4.7$, $\sigma=7.7$~dB etc.). The optimal average cost per step, under these parameter values, is
$\lambda^*=V^*(1)=1.85$ (computed numerically).\footnote{These values of $\xi_{out}$ and $\xi_{relay}$ are chosen because they can 
produce reasonable values of placement rate, mean power per step and mean outage per step, which can be used in practical networks. 
However, these values are chosen only for illustration purposes, and the choice will vary depending on the requirement 
for deployment.}

We numerically study the performance of the following three types of algorithms: (i) $\eta$ and $\sigma$ are 
known prior to deployment (the agent uses the fixed optimal policy with $\xi_{relay}=2$ and $\xi_{out}=125$ in this case), 
(ii) the agent has imperfect estimates of $\eta$ and $\sigma$  deployment, and OptAsYouGoLearning is 
used to update the  policy as deployment progresses, and 
(iii) the agent has imperfect estimates of $\eta$ and $\sigma$  deployment, but 
the corresponding suboptimal policy is used along the infinite line without any update. We use 
the abbreviations OAYGL and 
OAYG for OptAsYouGoLearning and Optimal Algorithm for As-You-Go deployment 
(i.e., Algorithm~\ref{algorithm:optimal-policy-structure}), respectively. Also, following 
the terminology in \cite{chattopadhyay-etal15measurement-based-impromptu-deployment-arxiv-v1}, 
we use the abbreviation FPWU for ``Fixed Policy without Update.''

Next, we formally explain the various cases  considered in our simulations:

\begin{enumerate}[label=(\roman{*})]

{\bf \item OAYG:} Here the agent  knows  
$\eta=4.7$, $\sigma=7.7$~dB prior to deployment, and uses  Algorithm~\ref{algorithm:optimal-policy-structure} 
with $\xi_{out}=125$, $\xi_{relay}=2$. 

 {\bf \item  OAYGL Case~$1$:} Here the true $\eta=4.7$ and $\sigma=7.7$~dB  are unknown to 
the deployment agent. But the agent has an initial estimate 
$\eta=5$, $\sigma=8$~dB. Hence, he starts deploying using a $\underline{V}^{(0)}$ which is optimal for 
these imperfect estimates of $\eta$ and $\sigma$, and $\xi_{out}=125$, $\xi_{relay}=2$. He updates 
the policy using the OptAsYouGoLearning algorithm as deployment progresses.

 {\bf \item  OAYGL Case~$2$:} This is different from OAYGL Case~$1$ only in the aspect that here 
 deployment starts with the optimal policy for $\eta=4$, $\sigma=7$~dB.

{\bf \item FPWU Case~$1$:} Here the true $\eta$ and $\sigma$ are unknown prior to deployment, and 
the agent has an initial estimate 
$\eta=5$, $\sigma=8$~dB. The agent computes $\underline{V}^*$ for these imperfect initial estimates and  
$\xi_{out}=125$, $\xi_{relay}=2$, and uses this policy throughout the 
deployment process without any update. This case will 
demonstrate the gain in performance by updating the policy under OptAsYouGoLearning, w.r.t. the case 
where the suboptimal policy  is used throughout the deployment process.

{\bf \item FPWU Case~$2$:} It differs from 
FPWU Case~$1$ only in the aspect that here the agent has initial estimates 
$\eta=4$, $\sigma=7$~dB.
\end{enumerate}

For simulation of OAYGL, we chose $a(k)=\frac{120}{k}$. 
We simulated $2000$ independent network deployments (i.e., $2000$ sample paths of the 
deployment process) with OptAsYouGoLearning, and 
estimated (by averaging over $2000$ deployments) the expectation of $V^{(k)}(1)$, 
mean power per step (i.e., $\frac{\sum_{j=1}^{N_k} \Gamma_j}{k}$), mean outage per step  (i.e., $\frac{\sum_{j=1}^{N_k} Q_{out}^{(j,j-1)}}{k}$) 
and mean placement distance (i.e., $\frac{k}{N_k}$), in the part of the network between the sink node 
to the $k$-th step. The results are summarized in 
Figure~\ref{fig:single_timescale_plots}. Asymptotically the estimates are supposed to converge to the values provided by 
OAYG.

{\bf Observations:} We observe that the estimate of 
$\mathbb{E} (V^{(k)}(1))$ approaches the optimal cost 
$\lambda^*=V^*(1)=1.85$ (for the actual propagation 
parameters), as $k$ increases, and gets to within $10\%$ of the optimal 
cost by the time where $k=35$ to $40$ (within a distance of $800$~meters), while 
starting with two widely different initial guesses of 
the propagation parameters. The estimates of mean power per step, mean outage per step and  
mean placement distance also converges very fast to the corresponding values achieved by OAYG. It also shows that, 
if the performance of the initial imperfect policy (FPWU) is significantly different than that of OAYG, then 
OptAsYouGoLearning will provide closer performance to OAYG, as compared to FPWU (see the mean placement distance 
plot).

Note that, even though Theorem~\ref{theorem:OptAsYouGoLearning} guarantees almost sure convergence, the 
convergence speed will vary across sample paths. But here we demonstrate speed of convergence after averaging over 
$2000$ sample paths.

\begin{figure*}[t]
\begin{minipage}[r]{0.5\linewidth}
\subfigure{
\includegraphics[width=\linewidth, height=6cm]{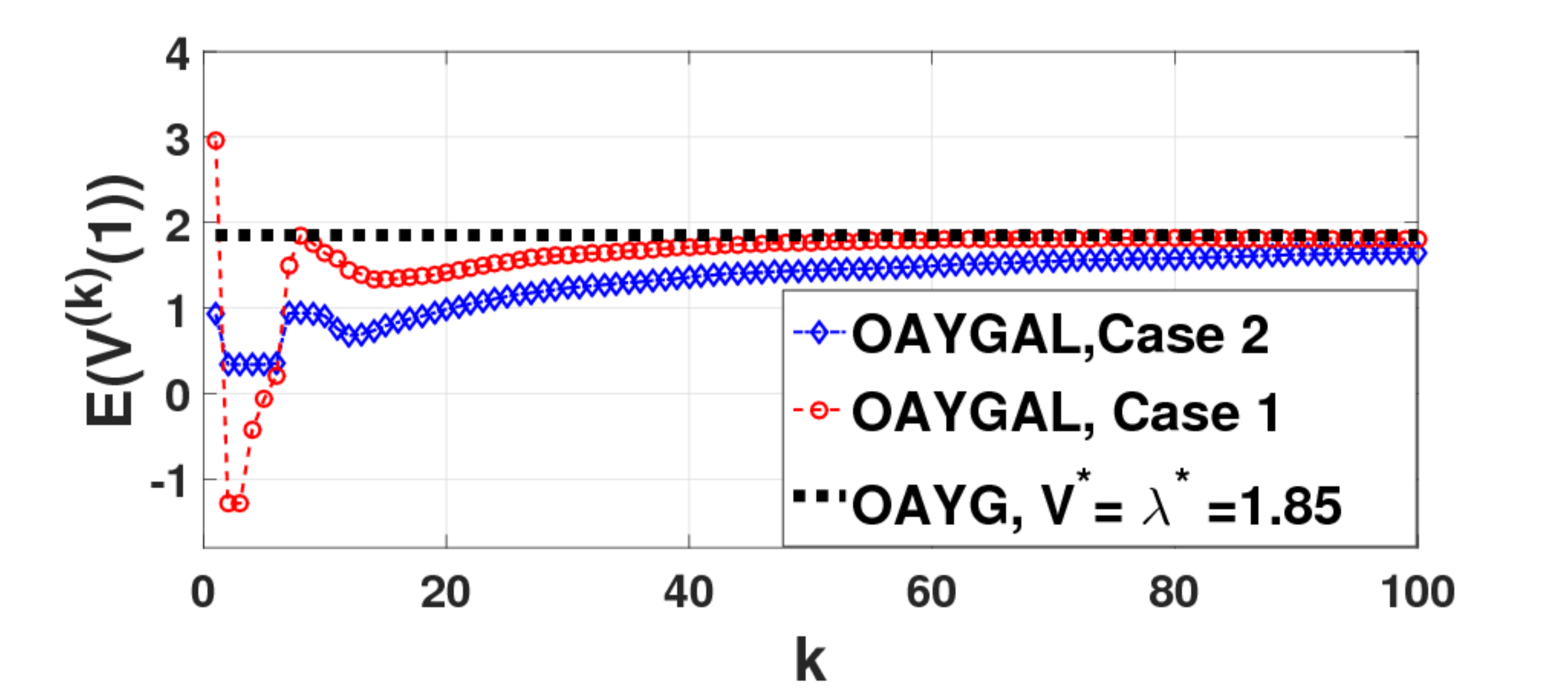}
\includegraphics[width=\linewidth, height=6cm]{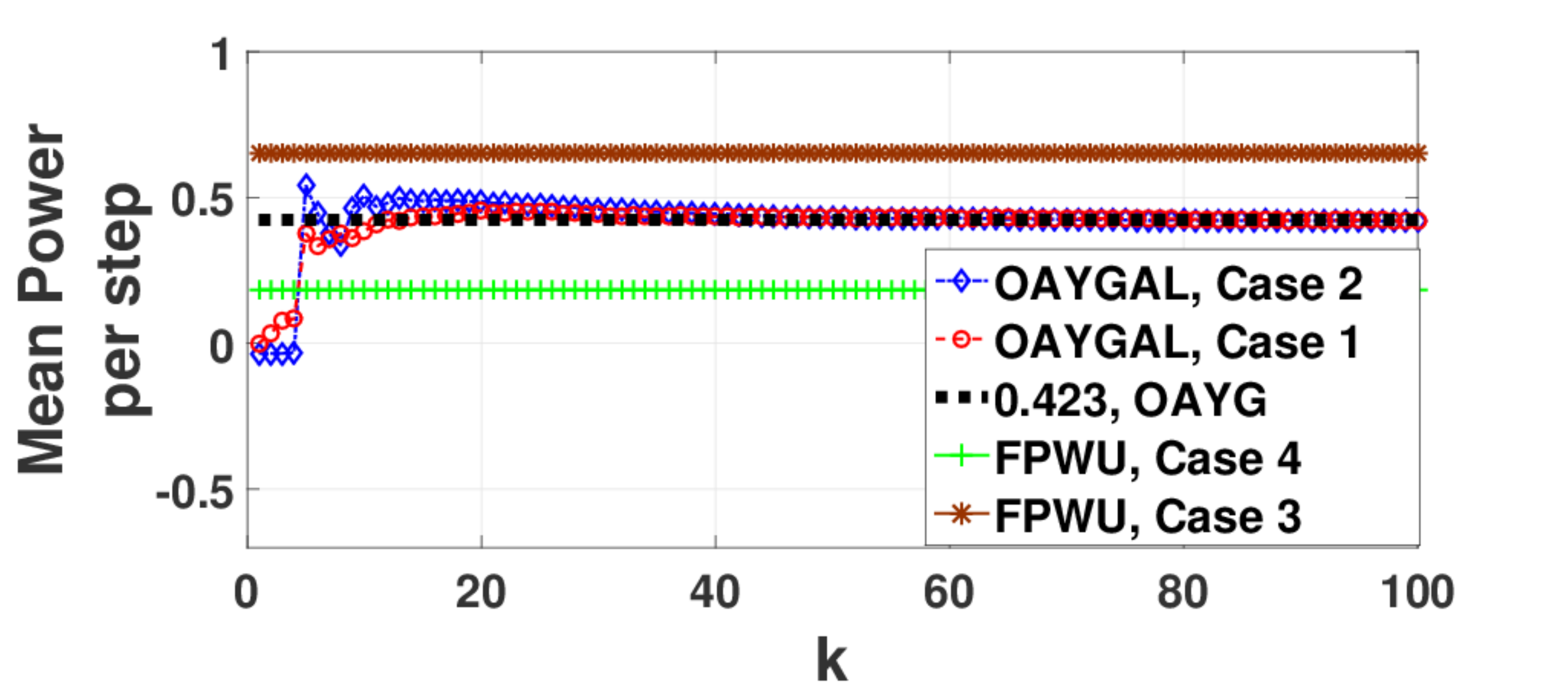}}
\end{minipage} \\ \hfill
\vspace{-3mm}
\begin{minipage}[c]{0.5\linewidth}
\subfigure{
\includegraphics[width=\linewidth, height=6cm]{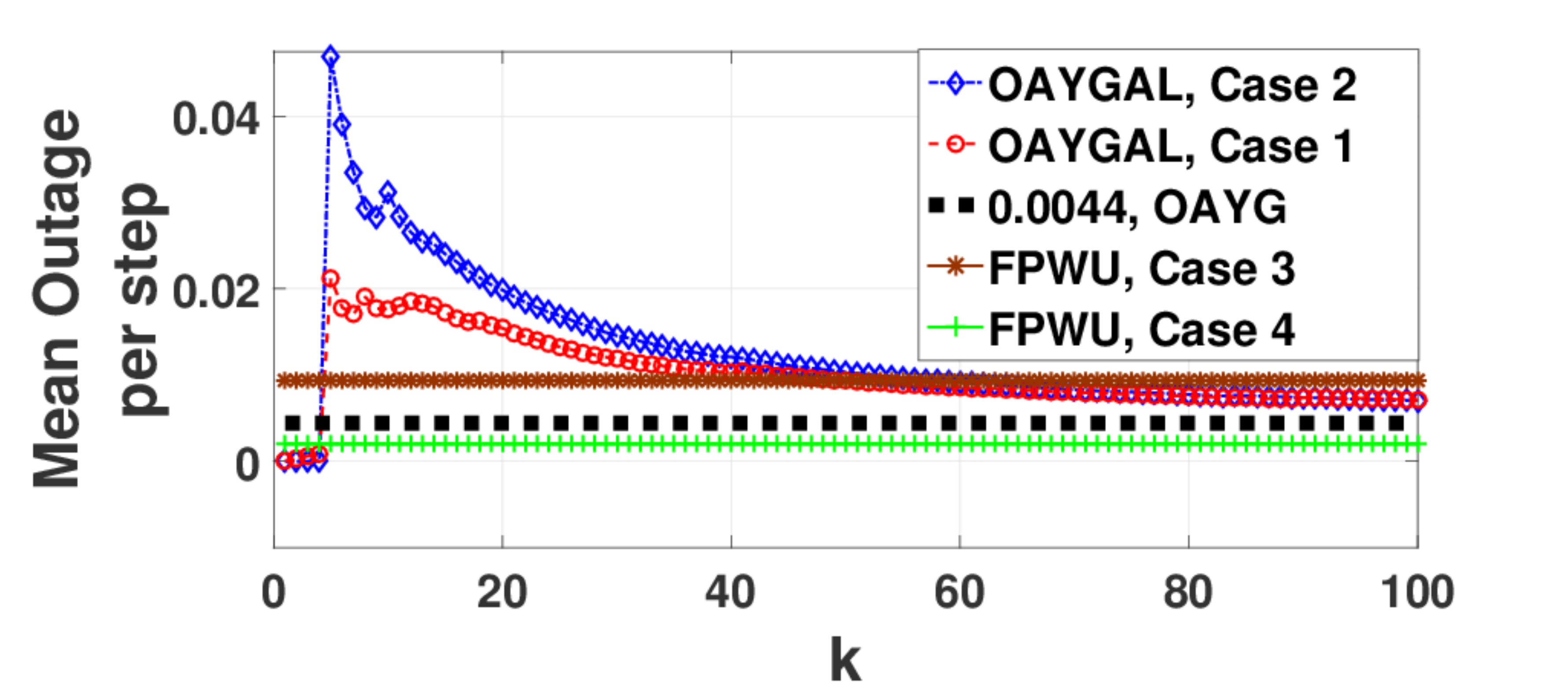}
\includegraphics[width=\linewidth, height=6cm]{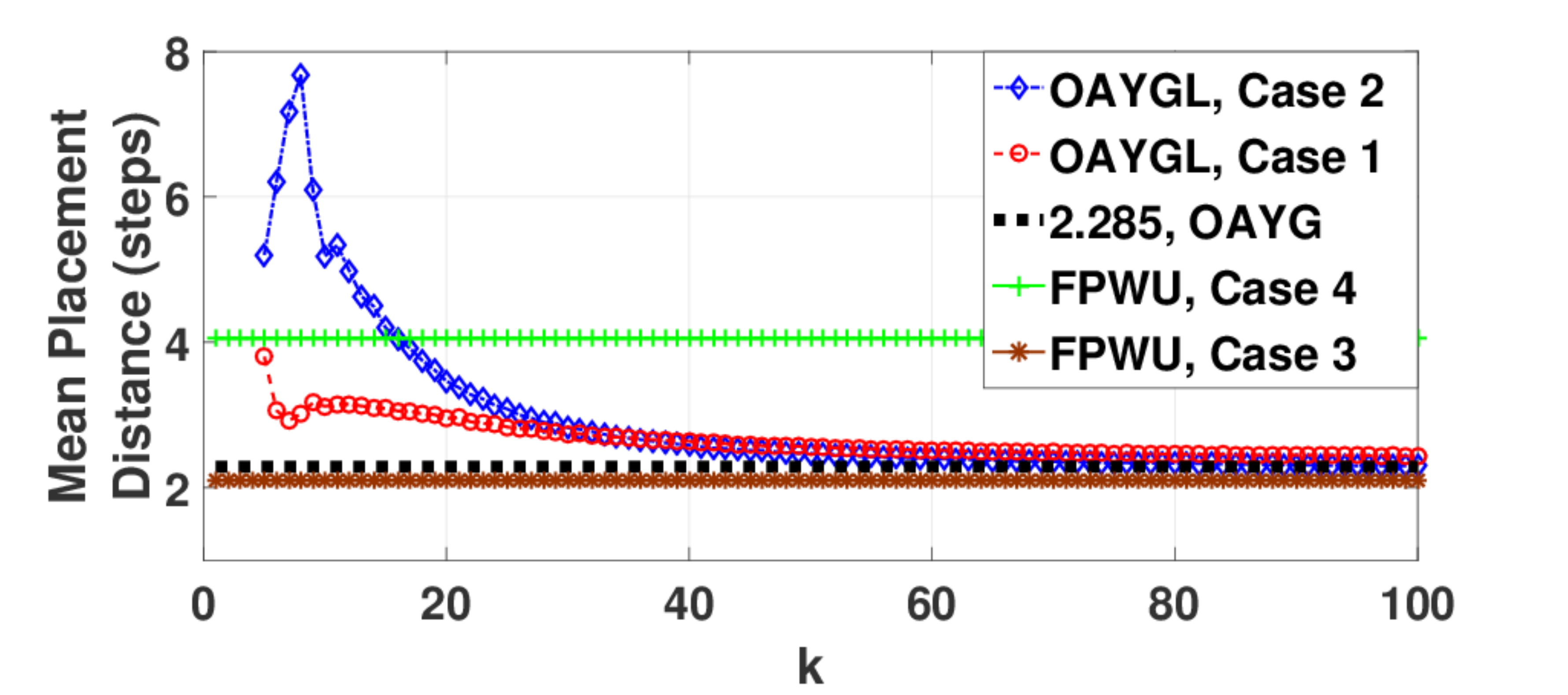}}
\end{minipage}\\ \hfill
\vspace{-2mm}
\begin{minipage}[r]{0.5\linewidth}
\subfigure{
\includegraphics[width=\linewidth, height=6cm]{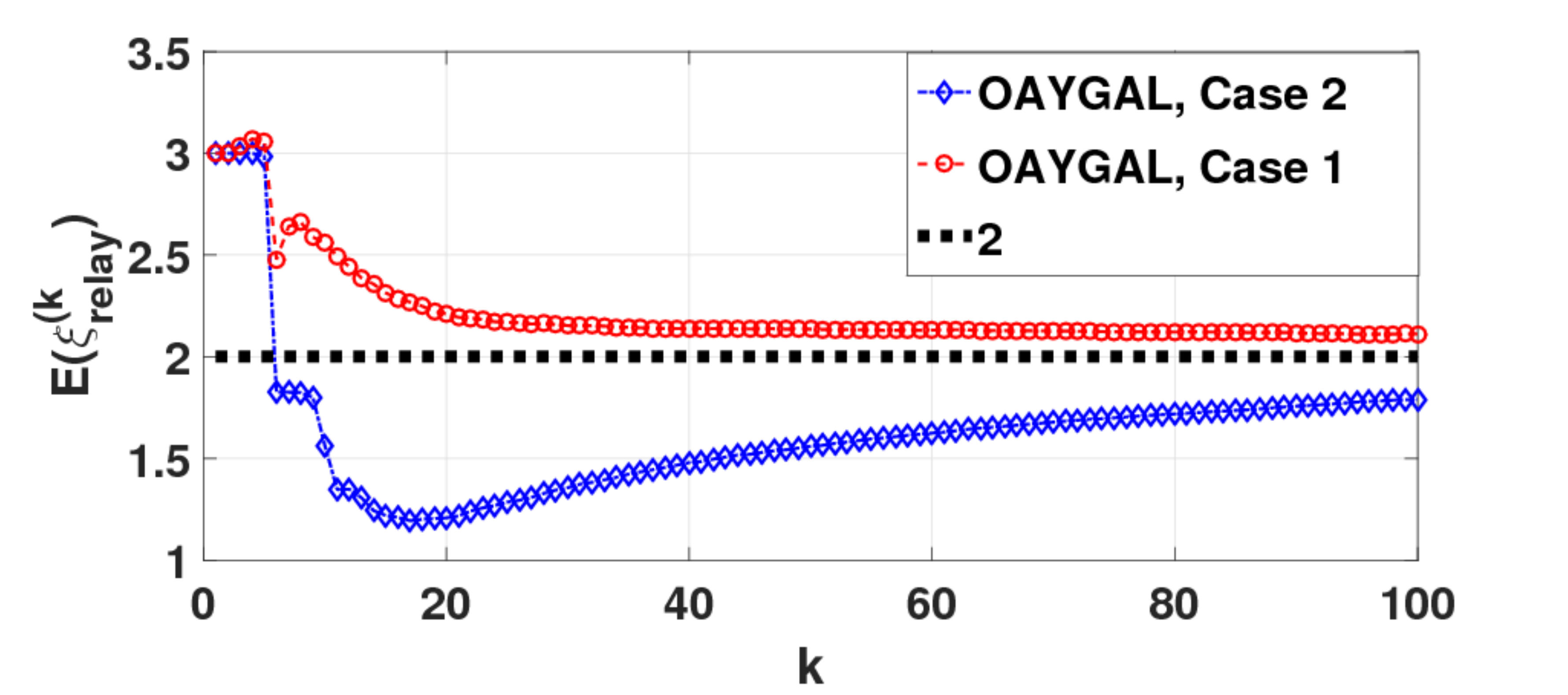}
\includegraphics[width=\linewidth, height=6cm]{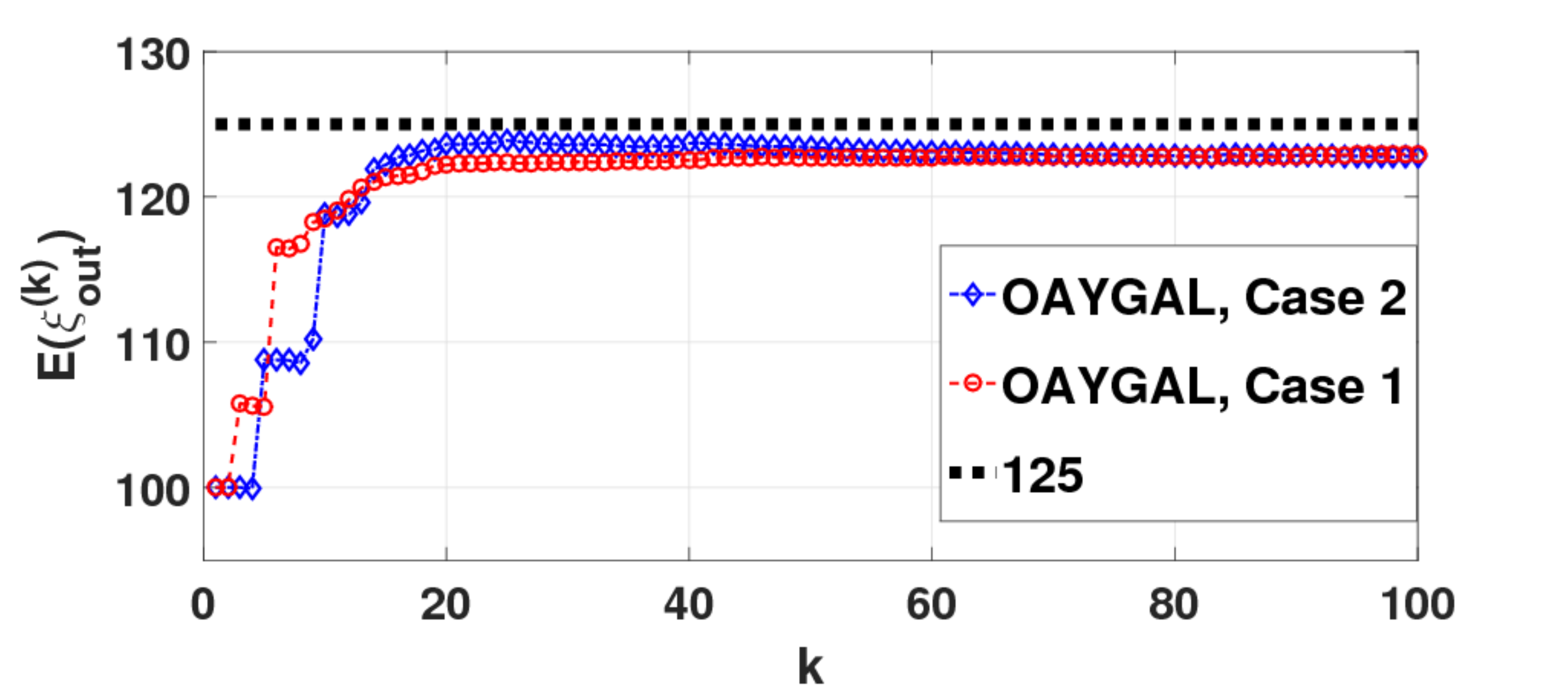}}
\end{minipage}  \hfill
\vspace{6mm}
\caption{Convergence speed of OptAsYouGoAdaptiveLearning (OAYGAL) with the number of steps, $k$. 
In the legends, ``OAYG'' refers to the values that are obtained if Algorithm~\ref{algorithm:optimal-policy-structure}  is used; 
these are the target values for OptAsYouGoAdaptiveLearning. Evolution of $\xi_{out}^{(k)}$ and $\xi_{relay}^{(k)}$ are shown 
for a longer time, since they converge slowly to their respective target values.}
\label{fig:two_timescale_plots}
\vspace{-5mm}
\end{figure*}

\vspace{-3mm}

\subsection{OptAsYouGoAdaptiveLearning}
\label{subsection:convergence_speed_optasyougoadaptivelearning}
\vspace{-1mm}

Now we will demonstrate the performance of OptAsYouGoAdaptiveLearning 
(Algorithm~\ref{algorithm:OptAsYouGoAdaptiveLearning}) 
for deployment over a 
finite distance under an unknown propagation environment. 
We again assume that the true propagation parameters are given by  $\eta=4.7$, $\sigma=7.7$~dB. 
For these parameters, under the choice $\xi_{relay}=2$ and $\xi_{out}=125$, 
the optimal average cost per step will be $\lambda^*=1.85$, which can be achieved by 
OAYG (Algorithm~\ref{algorithm:optimal-policy-structure}). OAYG in this case will 
yield a mean placement distance of $2.285$~steps, a mean outage per step 
of $\frac{0.0101}{2.285}=0.0044$, and a mean power per step 
of $0.423$~mW.

Now, let us suppose that we need to solve the constrained problem 
in (\ref{eqn:constrained_problem_average_cost_with_outage_cost}) with 
the targets $\overline{q}=0.0044$ and $\overline{N}=\frac{1}{2.285}$, 
but the true $\eta$ and $\sigma$ of the environment are unknown to us. Hence, we need to employ 
OptAsYouGoAdaptiveLearning (we use the abbreviation OAYGAL for it); as compared to OptAsYouGoLearning, we need 
to make  an additional choice of $\xi_{out}^{(0)}$ and $\xi_{relay}^{(0)}$. 

We consider the following cases  in our simulations:

\begin{enumerate}[label=(\roman{*})]

{\bf \item OAYG:} This is same as in Section~\ref{subsection:convergence_speed_optasyougolearning}

 {\bf \item  OAYGAL Case~$1$:} Here the true $\eta=4.7$ and $\sigma=7.7$~dB  are unknown to 
the deployment agent. But the agent has an initial estimate 
$\eta=5$, $\sigma=8$~dB. Hence, he starts deploying using a $\underline{V}^{(0)}$ which is optimal for 
these imperfect estimates of $\eta$ and $\sigma$, and $\xi_{out}^{(0)}=100$, $\xi_{relay}^{(0)}=3$. He updates 
the policy using the OptAsYouGoAdaptiveLearning algorithm as deployment progresses.

 {\bf \item  OAYGAL Case~$2$:} This is same as OAYGAL Case~$1$, except that the agent starts deploying 
 using a policy corresponding to the wrong initial estimate 
$\eta=4$, $\sigma=7$~dB (under $\xi_{out}^{(0)}=100$, $\xi_{relay}^{(0)}=3$).

{\bf \item FPWU Case~$3$:} Here the agent uses $\xi_{out}=100$, $\xi_{relay}=3$, and uses 
the corresponding optimal policy for the imperfect estimates $\eta=5$, $\sigma=8$~dB, throughout the 
deployment process.

{\bf \item FPWU Case~$4$:} This is similar to  
FPWU Case~$3$; the only difference is that the optimal policy for the imperfect estimates $\eta=4$, $\sigma=7$~dB 
is used throughout deployment.
\end{enumerate}

For simulation of OAYGAL, we chose the step sizes as  follows. We took $a(k)=\frac{1}{k^{0.55}}$, 
$b(k)=\frac{100}{k^{0.8}}$ for the $\xi_{out}$ update and $b(k)=\frac{1}{k^{0.8}}$ for the $\xi_{relay}$ update 
(however, both $\xi_{out}$ and $\xi_{relay}$ are updated in the same timescale). 
We simulated $2000$ independent network deployments (i.e., $2000$ sample paths of the 
deployment process) with OptAsYouGoLearning, and 
estimated (by averaging over $2000$ deployments) the expectations of $V^{(k)}(1)$,  
mean power per step, mean outage per step  
mean placement distance, $\xi_{out}^{(k)}$ and $\xi_{relay}^{(k)}$, in the part of the network between the sink node 
to the $k$-th step. The results are summarized in 
Figure~\ref{fig:two_timescale_plots} (see previous page).

{\bf Observations:} Under OAYGAL Case~$1$ the estimates of 
the expectations of $V^{(2000)}(1)$, $\xi_{out}^{(2000)}$, $\xi_{relay}^{(2000)}$, 
mean power per step up to the $2000^{th}$ step, mean outage per step up to the $2000^{th}$ step, and 
mean placement distance over $2000$ steps are $1.8479$, $124.89$, $2.01$, $0.4222$, $0.04403$  and 
$2.2852$, whereas 
the corresponding target values are $1.85$, $125$, $2$, $0.4223$, $0.00441$ and $2.2857$, respectively. 
Similarly, for OAYGAL Case~$2$ also, the quantities converge   close to the target values. In practice, the performance metrics are 
reasonably close to their respective target values within $100$~steps (i.e., $2$~kms). 

FPWU Case~$3$ and  FPWU Case~$4$ either violate some constraint or 
uses significantly higher per-step power compared to  
OAYG. But, by using OptAsYouGoAdaptiveLearning, 
we can achieve mean power per step  close to the optimal  
while (possibly) violating the constraints by small amount. 
However,  performance of OAYGAL is significantly closer to the target compared to FPWU.\qed

The speed of convergence will depend on the choice of  
$a(k)$ and $b(k)$, of $\xi_{out}^{(0)}$, $\xi_{relay}^{(0)}$ and the initial 
estimates of $\eta$ and $\sigma$. However, optimizing   convergence speed over step size 
sequences is left   for future research.

\vspace{-2mm}
\section{Conclusion}\label{section:conclusion}
In this paper, we have formulated the problem of pure-as-you-go deployment along a line, 
under a very light traffic assumption. The problem was  
formulated as an average cost MDP, and its optimal policy structure was studied analytically. 
We also proposed two learning algorithms that   asymptotically converge to the 
corresponding optimal policies. 
Numerical  results have been provided to illustrate the speed of convergence of the learning algorithms. 

While this paper provides an interesting set of results, it can be extended or modified in several ways: 
(i) One can attempt to develop deployment algorithms for 2 dimensional regions, 
where multiple agents cooperate to carry out the deployment. 
(ii) One can also attempt to develop deployment algorithms that can provide theoretical guarantees on the data rate 
supported by the deployed networks (instead of assuming that the traffic is lone packet). (iii) The optimization of 
the rate of convergence for the learning algorithms by proper choice of the step sizes is also a challenging problem. 
We leave these issues for future research endeavours.

\vspace{-3mm}
\bibliographystyle{IEEEtran}
\bibliography{IEEEabrv,arpan-techreport}

\vspace{-15mm}

\begin{IEEEbiography}[{\includegraphics[width=1in,height=1in,clip,keepaspectratio]{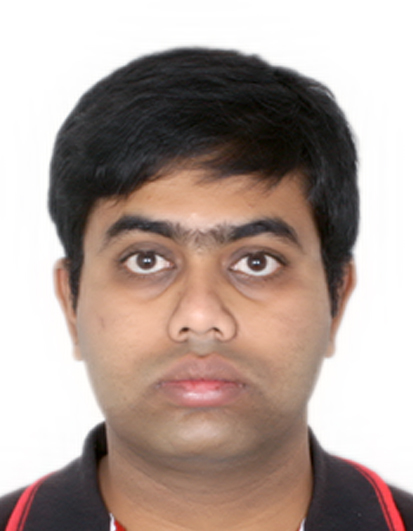}}]{Arpan 
Chattopadhyay} obtained his B.E. in Electronics and Telecommunication Engineering from Jadavpur University, 
India in the year 2008, and M.E. and Ph.D in Telecommunication Engineering from Indian Institute of Science, 
Bangalore, India in the year 2010 and 2015, respectively. He is currently working in Electrical Engineering department, University of Southern California, as a postdoctoral researcher. 
His research interests include  design, resource allocation, control and learning in wireless networks and cyber-physical systems.
    \end{IEEEbiography}

    \vspace{-15mm}

\begin{IEEEbiography}[{\includegraphics[width=1in,height=1in,clip,keepaspectratio]{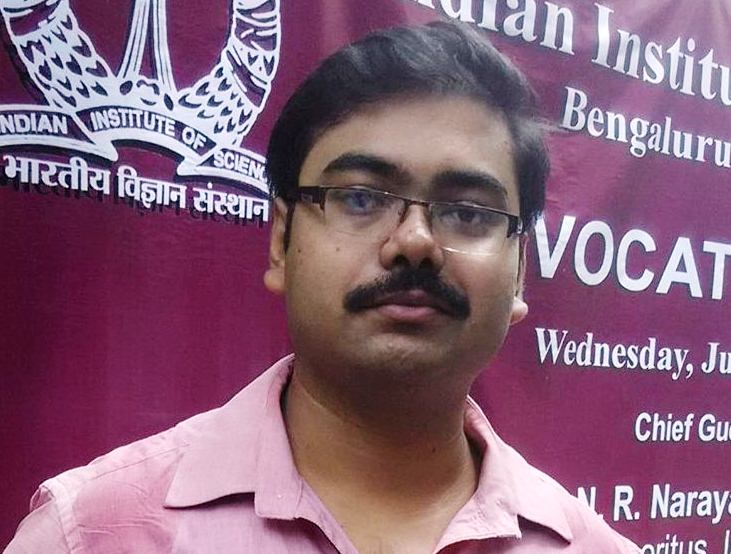}}]{Avishek Ghosh} 
obtained his B.E. in Electronics and Telecommunication Engineering from Jadavpur University, 
 India in 2012, and M.E.  in Telecommunication from Indian Institute of Science, 
Bangalore, India in the year 2014. He is currently doing his Ph.D in the department of EECS of UC Berkeley. 
His research interests include  networks and machine learning.
    \end{IEEEbiography}

   \vspace{-20mm}
   
    \begin{IEEEbiography}[{\includegraphics[width=1in,height=1in,clip,keepaspectratio]{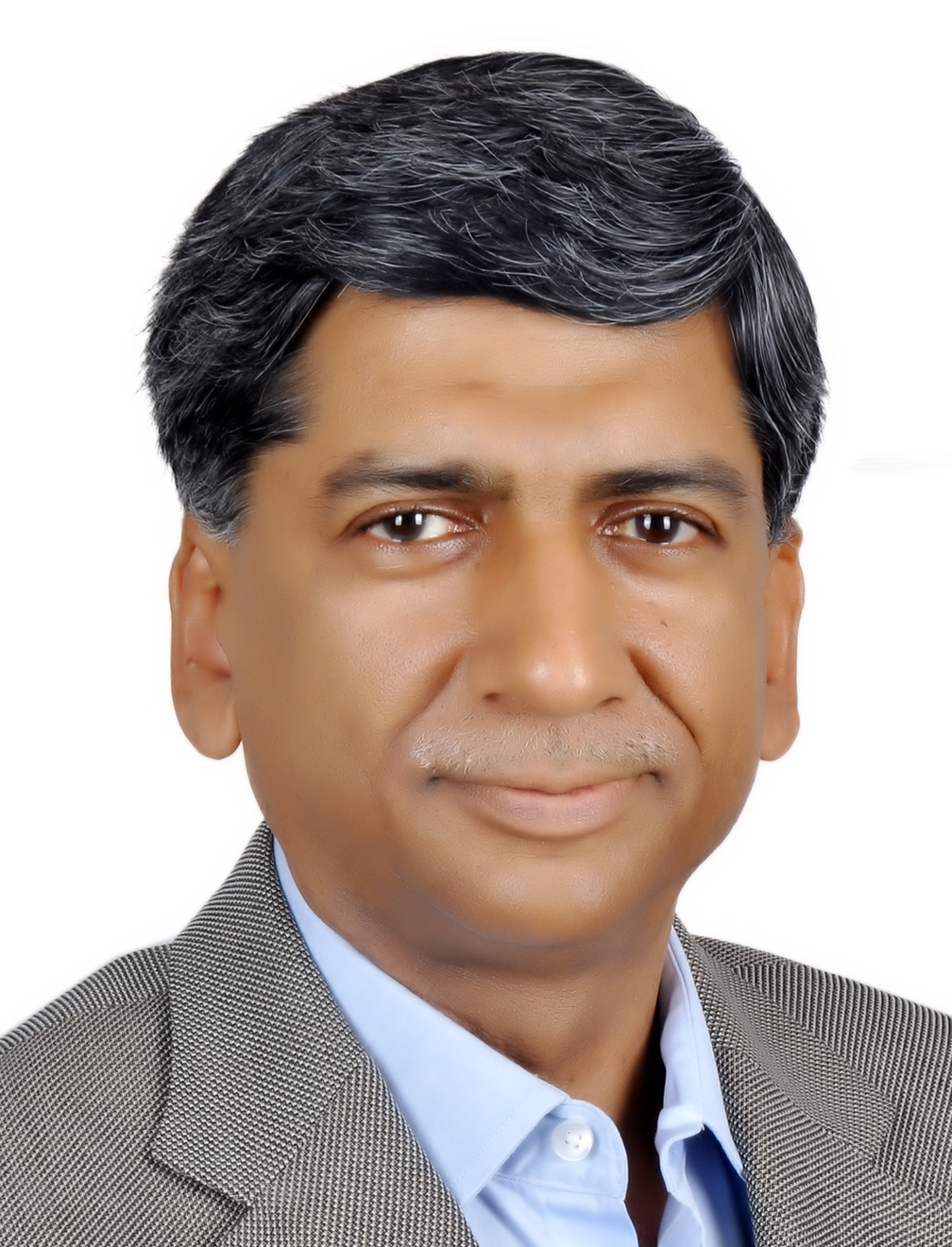}}]
    {Anurag Kumar} (B.Tech., Indian Institute of Technology (IIT)
Kanpur; PhD, Cornell University, both in Electrical Engineering) was
with Bell Labs, Holmdel, N.J., for over 6 years.  Since then he has
been on the faculty of the ECE Department at the Indian Institute of
Science (IISc), Bangalore; he is at present the Director of the
Institute.  His area of research is communication networking, and he
has recently focused primarily on wireless networking. He is a Fellow
of the IEEE, the Indian National Science Academy (INSA), the Indian
National Academy of Engineering (INAE), and the Indian Academy of
Sciences (IASc).  He was an associate editor of IEEE Transactions on
Networking, and of IEEE Communications Surveys and Tutorials.
 \end{IEEEbiography}

\newpage

\renewcommand{\thesubsection}{\Alph{subsection}}

\appendices
\hrule
{\bf Supplementary Material} \\
{\bf Title:} ``Asynchronous Stochastic Approximation Based 
Learning Algorithms for As-You-Go Deployment of Wireless Relay Networks along a Line'' \\
{\bf Authors:} Arpan Chattopadhyay, Avishek Ghosh, anurag Kumar\\
\hrule
\section{Formulation for known propagation parameters}
\label{appendix:mdp-for-pure-as-you-go-deployment}

\textbf{Proof of Theorem~\ref{theorem:uniqueness_of_V}:} 
From (\ref{eqn:average_cost_optimality_equation_no_backtracking-V-equation-no-lambda}), 
$V(B)$ is unique for fixed $\xi_{out}$ and $\xi_{relay}$. 
Hence, we can say that $V(B)$ is a continuous and decreasing function of 
$V(1)$. Now, let us assume that $V(r+1)$ is continuous and 
decreasing in $V(1)$ for some $r, 1 \leq r \leq B-1$. Let us recall 
(\ref{eqn:average_cost_optimality_equation_no_backtracking-V-equation-no-lambda}) for $V(r)$. 
Since $V(r+1)$ is continuous and decreasing in $V(1)$ by our induction hypothesis, it is evident 
from (\ref{eqn:average_cost_optimality_equation_no_backtracking-V-equation-no-lambda}) that $V(r)$ is 
also continuous and decreasing in $V(1)$. Proceeding in this way, we can write $V(1)=\phi(V(1))$ 
where $\phi(\cdot)$ is continuous and decreasing in $V(1)$. But $V(1)$ is continuous and strictly increasing 
in $V(1)$. Hence, $V(1)=\phi(V(1))$ has a unique fixed point $V^*(1)$. Now, 
from (\ref{eqn:average_cost_optimality_equation_no_backtracking-V-equation-no-lambda}), $V(B-1)$ is unique since 
$V(1)=V^*(1)$ is unique and $V(B)$ is unique. Proceeding backwards in this way, we can show that we have 
a unique $V^*(r)$ for all $r$.

Now, from (\ref{eqn:average_cost_optimality_equation_no_backtracking-V-equation-no-lambda}), we find 
that $V^*(r) \leq -V^*(1)+V^*(r+1)$, i.e., $V^*(r+1) \geq V^*(r)+V^*(1)$ for all $r \in \{1,2,\cdots,B-1\}$. Also, $V^*(1)=\lambda^*>0$ and 
it is unique. This proves the second part of the theorem.
\qed

\textbf{Proof of Theorem~\ref{theorem:lambda-increasing-concave-continuous-in-xi}:} 
Let us denote the mean power per link, mean outage per link and mean placement distance (in steps) under a 
stationary policy $\pi$ by $\overline{\Gamma}_{\pi}$, $\overline{Q}_{out,\pi}$ and $\overline{U}_{\pi}$. 
Then, by Renewal-Reward Theorem, we have 
$\lambda^{*}(\xi_{out},\xi_{relay}) =\inf_{\pi}\frac{\Gamma_{\pi}+\xi_{out}\overline{Q}_{out,\pi}+\xi_{relay}}{\overline{U}_{\pi}}$. 
The numerator is affine and increasing in  $\xi_{out}$ and $\xi_{relay}$,   
and the denominator is independent of  $\xi_{out}$ and $\xi_{relay}$.  Hence, 
$\lambda^{*}(\xi_{out},\xi_{relay})$ is concave, increasing in $\xi_{out}$ and $\xi_{relay}$, 
since the pointwise infimum of increasing affine functions of  $(\xi_{out},\xi_{relay})$ is increasing and 
jointly concave in  $(\xi_{out},\xi_{relay})$. Now, any increasing, concave function is continuous. Hence, 
$\lambda^{*}(\xi_{out},\xi_{relay})$ is continuous in $(\xi_{out},\xi_{relay})$. Also, it is easy to see that 
$\lambda^{*}(\xi_{out},\xi_{relay})$ is Lipschitz in each argument with Lipschitz constant $1$.

\textbf{Proof of Theorem~\ref{theorem:V-continuous-in-xi}:} By 
Theorem~\ref{theorem:lambda-increasing-concave-continuous-in-xi}, 
$V^*(1):=\lambda^*$ is Lipschitz continuous in $(\xi_{out}, \xi_{relay})$. By 
(\ref{eqn:average_cost_optimality_equation_no_backtracking-V-equation-no-lambda}), $V^*(B)$ is Lipschitz continuous in 
$(\xi_{out},\xi_{relay})$. Hence, by (\ref{eqn:average_cost_optimality_equation_no_backtracking-V-equation-no-lambda}), 
$V^*(B-1)$ is also Lipschitz continuous in $(\xi_{out}, \xi_{relay})$. Thus, by using backward induction, 
we can show that $V^*(r)$ is Lipschitz continuous in $(\xi_{out}, \xi_{relay})$ for all 
$1 \leq r \leq B$.

\section{OptAsYouGoLearning: Learning with Pure As-You-Go Deployment, for Given Lagrange Multipliers}
\label{appendix:learning-for-pure-as-you-go-deployment-given-xi}

\textbf{Proof of Theorem~\ref{theorem:OptAsYouGoLearning}:}
We can rewrite (\ref{eqn:learning_no_backtracking_given_xio_xir_update_part}) as follows: 

\footnotesize
\begin{eqnarray}
 V^{(k)}(r)&=&V^{(k-1)}(r)+ a(\nu(r,k)) \mathbb{I}\{r \in \mathcal{I}_k\} \bigg[ f_r(\underline{V}^{(k-1)})+ M_k(r) \bigg]  \nonumber\\ 
&& \label{eqn:learning_no_backtracking_given_xio_xir_update_part_in_standard_asynchronous_SAA_form}
\end{eqnarray}
\normalsize

where, for all $1 \leq r \leq B-1$

\footnotesize
\begin{eqnarray*}
 f_r(\underline{V}^{(k-1)})&=&\mathbb{E}_W \bigg[ \min \bigg \{ \min_{\gamma}(\gamma+ \xi_{out} Q_{out}(r,\gamma, W))+\xi_{relay}, \\
 && -V^{(k-1)}(1)+V^{(k-1)}(r+1) \bigg \} -V^{(k-1)}(r) \bigg] 
\end{eqnarray*}

\begin{eqnarray*}
 M_k(r)&=& \bigg[\min \bigg \{ \min_{\gamma}(\gamma+ \xi_{out} Q_{out}(r,\gamma, w_r))+\xi_{relay},  \\ 
&& -V^{(k-1)}(1)+V^{(k-1)}(r+1) \bigg \} -V^{(k-1)}(r) \bigg]-f_r(\underline{V}^{(k-1)})
\end{eqnarray*}

and

\begin{eqnarray*}
 f_B(\underline{V}^{(k-1)})=\mathbb{E}_W \bigg[  \min_{\gamma}(\gamma+ \xi_{out} Q_{out}(B,\gamma, W)) +\xi_{relay} -V^{(k-1)}(B) \bigg]
\end{eqnarray*}

\begin{eqnarray*}
 M_k (B)&=&\bigg[  \min_{\gamma}(\gamma+ \xi_{out} Q_{out}(B,\gamma, w_B)) +\xi_{relay}-V^{(k-1)}(B) \bigg] \\
&& -f_B(\underline{V}^{(k-1)})
\end{eqnarray*}
\normalsize

Let $\underline{M}_k:=(M_k(1), \cdots, M_k(B))$. Let us denote the $\sigma$-field 
$\mathcal{F}_k:=\sigma(\underline{V}_i, \mathcal{I}_i, \underline{M}_i, i \leq k-1)$; it is 
the information available to the deployment agent before making any decision at the $k$-th step. 
Clearly, the update equations fall under the category of Asynchronous Stochastic Approximation 
algorithms (see \cite{bhatnagar11borkar-meyn-theorem-asynchronous-stochastic-approximation}). In order to see whether 
$\underline{V}^{(k)} \rightarrow \underline{V}^*$ almost surely, we will first check whether the five assumptions 
mentioned in \cite{bhatnagar11borkar-meyn-theorem-asynchronous-stochastic-approximation} are satisfied. 

\textbf{Checking Assumption~$1$ of \cite{bhatnagar11borkar-meyn-theorem-asynchronous-stochastic-approximation}:} 
For each $r, 1 \leq r \leq B$, $V(r)$ gets updated at least once in every $B$~steps. 
Hence, $\lim \inf_{k \rightarrow \infty}\frac{\nu(r,k)}{k} \geq \frac{1}B>0 $ almost surely. Hence, the assumption is satisfied.

\textbf{Checking Assumption~$2$ of \cite{bhatnagar11borkar-meyn-theorem-asynchronous-stochastic-approximation}:} 
If we choose $\{a(k)\}_{k \geq 1}$ to be a bounded, decreasing sequence with $\sum_k a(k)=\infty$ 
and $\sum_k a^2(k) < \infty$, this condition will be satisfied.

\textbf{Checking Assumption~$3$ of \cite{bhatnagar11borkar-meyn-theorem-asynchronous-stochastic-approximation}:} 
Not applicable to our problem since before updating $\underline{V}^{(k)}$ the deployment agent knows 
$\underline{V}^{(k-1)}$.

Before checking the other two conditions, we will establish a lemma. Let us consider the following system of o.d.e-s:
\begin{eqnarray}\label{eqn:system-of-ode}
 \dot{V}_t (r) = \kappa_t(r) f_r(\underline{V}_t) \,\,\, \forall r \in \{1,2,\cdots,B\}
\end{eqnarray}
where $\kappa_t(r) \in (0,1]$ for all $r$ and $t$. By Theorem~\ref{theorem:uniqueness_of_V}, this 
system of o.d.e-s has an unique stationary point $\underline{V}^*(\xi_{out},\xi_{relay})$. 

\begin{lemma}\label{lemma:ode-globally-asymptotically-stable}
 $\underline{V}^*(\xi_{out},\xi_{relay})$ is a globally asymptotically stable equilibrium for the system of o.d.e-s \eqref{eqn:system-of-ode}. 
 Also, $\underline{V}=0$ is a globally asymptotically stable equilibrium  for \eqref{eqn:system-of-ode} when 
 $\gamma$, $\xi_{out}$ and $\xi_{relay}$ are replaced by 
 $0$ in the definition of $f_r(\underline{V})$ for all $r \in \{1,2,\cdots,B \}$.
\end{lemma}
\begin{proof}
Note that, by Theorem~\ref{theorem:uniqueness_of_V}, $\underline{V}^*(\xi_{out},\xi_{relay})$ is the unique stationary point for 
\eqref{eqn:system-of-ode}. Now, the proof for this lemma follows from similar line of arguments as in the appendix of 
 \cite{salodkar-etal08online-learning-delay-constrained-scheduling} (which uses results from \cite{abounadi-etal01learning-algorithm-mdp-average-cost}
and \cite{borkar-soumyanath97analog-fixed-point-computation}).
\end{proof}

\textbf{Checking Assumption~$4$ of \cite{bhatnagar11borkar-meyn-theorem-asynchronous-stochastic-approximation}:} 
It is easy to see that $f_r (\underline{V})$ is Lipschitz in $\underline{V}$ for each $r$; this satisfies 
Assumption~4(i). Let us consider the ODE \eqref{eqn:system-of-ode} with 
$0<\kappa_t(r) \leq 1$ corresponds to the relative rate at which $V(r)$ is updated. By Lemma~\ref{lemma:ode-globally-asymptotically-stable}, 
$\underline{V}^*(\xi_{out},\xi_{relay})$ is a globally asymptotically stable equilibrium for the system of o.d.e-s \eqref{eqn:system-of-ode}. 
Hence, Assumption~4(ii) is satisfied.

Consider the functions $\frac{f_r(c\underline{V})}{c}, c \geq 1$ for all $r$. Clearly, 
$\lim_{c \rightarrow \infty} \frac{f_r(c\underline{V})}{c}=\min\{0,-V(1)+V(r+1)\}-V(r)$ for $r \neq B$, and 
$\lim_{c \rightarrow \infty} \frac{f_B(c\underline{V})}{c}=-V(B)$. Note that $\frac{f_r(c\underline{V})}{c}$ for all 
$r$ and $\lim_{c \rightarrow \infty} \frac{f_r(c\underline{V})}{c}$ all are continuous in 
$\underline{V}$, and $\frac{f_r(c\underline{V})}{c}$ is decreasing in $c$. Hence, 
by Theorem~$7.13$ of \cite{rudin76principles-of-mathematical-analysis}, convergence of $\frac{f_r(c\underline{V})}{c}$ 
over compacts is uniform. Hence, Assumption~4(iii) is satisfied.

Consider the ODE: $\dot{V}_t(r)=\kappa_t(r)(\min\{0,-V_t(1)+V_t(r+1)\}-V_t(r))$ for $r \neq B$ and 
$\dot{V}_t(B)=\kappa_B(t)(-V_t(B))$. Clearly, by the second part of Lemma~\ref{lemma:ode-globally-asymptotically-stable}, 
there is a unique globally asymptotically stable equilibrium 
$\underline{V}=\underline{0}$. Hence, Assumption~4(iv) is satisfied.

\textbf{Checking Assumption~$5$ of \cite{bhatnagar11borkar-meyn-theorem-asynchronous-stochastic-approximation}:} 
It is easy to see that, $\{\underline{M}_k\}_{k \geq 1}$ is a Martingale difference sequence adapted 
to $\mathcal{F}_k$. Hence, Assumption~5(i) is satisfied.

Now,

\footnotesize
\begin{eqnarray*}
 |M_{k+1}(r)| & \leq & 2 \bigg| \bigg(\min\{P_M+\xi_{out}+\xi_{relay},-V^{(k)}(1) \\
&&  +V^{(k)}(r+1)\}-V^{(k)}(r) \bigg) \bigg|
\end{eqnarray*}
\normalsize

and 
\begin{eqnarray*}
 |M_{k+1}(B)| \leq  \bigg| \bigg(P_M+\xi_{out}+\xi_{relay} -V^{(k)}(B) \bigg) \bigg|
\end{eqnarray*}

Hence, $||M_{k+1}|| \leq C_0 (1+||\underline{V}^{(k)}||)$ for some $C_0>0$. 
Hence, Assumption~5(ii) is satisfied. Now, by  
\cite[Theorem~$3$]{bhatnagar11borkar-meyn-theorem-asynchronous-stochastic-approximation}, 
$\underline{V}^{(k)} \rightarrow \underline{V}^*$.\qed

\section{OptAsYouGoAdaptiveLearning with Constraints on Outage Probability and Relay Placement Rate}
\label{appendix:learning-for-pure-as-you-go-deployment-constrained-problem}

\subsection{Proof of Theorem~\ref{theorem:placement_rate_mean_outage_per_step_continuous_in_xio_and_xir}}
\label{subsection:proof_of_placement_rate_mean_outage_per_step_continuous_in_xio_and_xir}

Let us denote by $g(r, \gamma), 
r \in \{1,2,\cdots,B\}, \gamma \in \mathcal{S}$ 
the joint distribution of $(U_k, \Gamma_k)$ under Algorithm~\ref{algorithm:OptAsYouGoLearning}. 
For the time being, let us assume that $g(r,\gamma)$ is continuous in  $(\xi_{out},\xi_{relay})$. 
Then, the mean placement distance 
$\overline{U}^*(\xi_{out},\xi_{relay})= \sum_{r=1}^B \sum_{\gamma \in \mathcal{S}} r g (r,\gamma)$, and  
the mean power per link 
$\overline{\Gamma}^*(\xi_{out},\xi_{relay})=\sum_{r=1}^B\sum_{\gamma \in \mathcal{S}} \gamma g (r,\gamma)$ 
are both continuous in $(\xi_{out},\xi_{relay})$.

Now, by Renewal-Reward Theorem, 

\footnotesize
\begin{eqnarray*}
\lambda^*(\xi_{out}, \xi_{relay}) = \frac{ \overline{\Gamma}^*(\xi_{out},\xi_{relay})+\xi_{out}\overline{Q}_{out}^*(\xi_{out},\xi_{relay})+\xi_{relay}  }{\overline{U}^*(\xi_{out},\xi_{relay})}
\end{eqnarray*}
\normalsize

Since $\lambda^*(\xi_{out}, \xi_{relay})$ is continuous in $(\xi_{out},\xi_{relay})$ 
(by Theorem~\ref{theorem:lambda-increasing-concave-continuous-in-xi}), 
we conclude that $\overline{Q}_{out}^*(\xi_{out},\xi_{relay})$ 
is continuous in $\xi_{out}$ and $\xi_{relay}$. Hence, 
$\frac{\overline{\Gamma}^*(\xi_{out},\xi_{relay})}{\overline{U}^*(\xi_{out},\xi_{relay})}$, 
$\frac{\overline{Q}_{out}^*(\xi_{out},\xi_{relay})}{\overline{U}^*(\xi_{out},\xi_{relay})}$ and 
$\frac{1}{\overline{U}^*(\xi_{out},\xi_{relay})}$ are 
continuous in $(\xi_{out},\xi_{relay})$. \qed

Now, the proof of the theorem is completed by the following lemma.

\begin{lemma}\label{lemma:grgamma_continuous_in_xio_xir}
 Under Assumption~\ref{assumption:shadowing_continuous_random_variable}, 
$g(r,\gamma)$ is continuous in $(\xi_{out},\xi_{relay})$.
\end{lemma}
\begin{proof}
We will first prove the result for $r \leq B-1$. Let us fix an $r \in \{1,\cdots,B-1\}$ and any $\gamma \in \mathcal{S}$. 
We will only show that $g(r,\gamma)$ is continuous in $\xi_{out}$; the proof for continuity of $g(r,\gamma)$ w.r.t. 
$\xi_{relay}$ will be similar. 

Let us consider a sequence $\{\xi_n\}_{n \geq 1}$ such that $\xi_n \rightarrow \xi_{out}$. Let us 
denote the joint probability distribution of $(U_k, \Gamma_k)$  
by $g_n(r,\gamma)$, if  Algorithm~\ref{algorithm:optimal-policy-structure} is used with $\xi_n$ as the cost 
for unit outage. 
We will show that $\lim_{n \rightarrow \infty} g_n(r,\gamma) \rightarrow g(r,\gamma)$.

Define the sets 
$\mathcal{E}_{r,\gamma'}=\bigg\{w_r: \gamma+\xi_{out}Q_{out}(r,\gamma,w_r) < \gamma'+\xi_{out}Q_{out}(r,\gamma',w_r) \bigg\}$
and $\mathcal{E}_u=\bigg\{w_u: \min_{\gamma \in \mathcal{S}}(\gamma+\xi_{out}Q_{out}(u,\gamma,w_u))
>-\xi_{relay}-V^*(1)+V^*(u+1) \bigg\}$ for all $1 \leq u \leq r$.

Let us  define $\mathcal{E}=\cap_{\gamma' \neq \gamma}\mathcal{E}_{r,\gamma'}  
\cap_{u \leq r-1}  \mathcal{E}_u \cap \overline{\mathcal{E}_r}$, where 
$\overline{\mathcal{E}_r}$ is the set complement of $\mathcal{E}_r$. 

Now, $g(r,\gamma)=\mathbb{P}(\mathcal{E})=\mathbb{E}(\mathbb{I}_{\mathcal{E}})$, where $\mathbb{I}$ denotes 
the indicator function. The expectation is over the joint distribution of $(W_1,W_2, \cdots, W_r)$ 
(shadowing random variables from $r$ locations).

Now, for any $\gamma' \neq \gamma$, 
we have $\mathbb{P}\bigg( \gamma+\xi_{out}Q_{out}(r,\gamma,W_r) = \gamma'+\xi_{out}Q_{out}(r,\gamma',W_r) \bigg)=0$, 
and $\mathbb{P}\bigg( \min_{\gamma \in \mathcal{S}}(\gamma+\xi_{out}Q_{out}(u,\gamma,W_u))
=-\xi_{relay}-V^*(1)+V^*(u+1) \bigg)=0$ for all $u \leq r$; 
these two assertions follow from Assumption~\ref{assumption:shadowing_continuous_random_variable} and 
from the continuity of $Q_{out}(r,\gamma,w)$  in $w$. 
Hence, we can safely assume the following:

\begin{itemize}
 \item $\overline{\mathcal{E}_{r,\gamma'}}$ has the same expression as 
$\mathcal{E}_{r,\gamma'}$ except that the $<$ sign is replaced by $>$ sign.
\item $\overline{\mathcal{E}_u}$ has the same expression as 
$\mathcal{E}_u$ except that the $>$ sign is replaced by $<$ sign.
\end{itemize}

Let $\mathcal{E}_{r,\gamma'}^{(n)}$, $\mathcal{E}_u^{(n)}$ and 
$\mathcal{E}^{(n)}$ be the sets obtained by replacing $\xi_{out}$ by $\xi_n$ in the expressions of the sets 
$\mathcal{E}_{r,\gamma'}$, $\mathcal{E}_u$ and $\mathcal{E}$ respectively (also $\underline{V}^*$ 
has to be replaced by the corresponding optimal $\underline{V}^{(n,*)}$). 
Clearly, we can make similar claims for  $\mathcal{E}_{r,\gamma'}^{(n)}$, $\mathcal{E}_u^{(n)}$.

Now, if we can show that $\mathbb{E}(\mathbb{I}_{\mathcal{E}^{(n)}}) \rightarrow \mathbb{E}(\mathbb{I}_{\mathcal{E}})$, 
the lemma will be proved, because $g(r,\gamma)=\mathbb{P}(\mathcal{E})=\mathbb{E}(\mathbb{I}_{\mathcal{E}})$.

\begin{claim}\label{claim:convergence_of_indicator_function}
 $\lim_{n \rightarrow \infty}\mathbb{I}_{ \mathcal{E}_u^{(n)} } \rightarrow \mathbb{I}_{ \mathcal{E}_u} $,  
and $\lim_{n \rightarrow \infty} \mathbb{I}_{ \mathcal{E}_{r,\gamma'}^{(n)} } \rightarrow \mathbb{I}_{ \mathcal{E}_{r,\gamma'} } $ 
almost surely, for $\gamma' \neq \gamma$.
\end{claim}
\begin{proof}
 Suppose that, for some value of $w_u$, 
$\mathbb{I}_{ \mathcal{E}_u}(w_u)=1$, i.e., 
$\min_{\gamma \in \mathcal{S}}(\gamma+\xi_{out}Q_{out}(u,\gamma,w_u))
>-\xi_{relay}-V^*(1)+V^*(u+1)$. 
Now, 
$V^*(1)$ and $V^*(u+1)$ are  continuous in 
$(\xi_{out},\xi_{relay})$ for all $1 \leq u \leq r$ (see Theorem~$\ref{theorem:V-continuous-in-xi}$). 
Hence, there exists an integer $n_0$ large enough, such that 
for all $n > n_0$, we have $\min_{\gamma \in \mathcal{S}}(\gamma+\xi_n Q_{out}(u,\gamma,w_u))
>-\xi_{relay}-\bigg( V^{(n,*)}(1)+V^{(n,*)}(u+1) \bigg)\bigg|_{\xi_{out}=\xi_n}$, i.e., 
$\mathbb{I}_{ \mathcal{E}_u^{(n)} } (w_u)=1$ for all $n > n_0$. Hence, 
$\mathbb{I}_{ \mathcal{E}_u^{(n)} } (w_u) \rightarrow \mathbb{I}_{ \mathcal{E}_u } 
(w_u) $ if $\mathbb{I}_{ \mathcal{E}_u } 
(w_u)=1$. For the case $\mathbb{I}_{ \mathcal{E}_u } (w_u)=0$, we can have similar arguments. This 
proves the first part of the claim, and second part can be proved by similar arguments.
\end{proof}

Now, $\mathbb{I}_{\mathcal{E}^{(n)}}=\prod_{\gamma' \neq \gamma} \mathbb{I}_{\mathcal{E}_{r,\gamma'}^{(n)}} 
\prod_{u \leq r-1} \mathbb{I}_{\mathcal{E}_u^{(n)}} \times \mathbb{I}_{\overline{\mathcal{E}_r^{(n)}}}$. 
By Claim~\ref{claim:convergence_of_indicator_function}, $\mathbb{I}_{\mathcal{E}^{(n)}} \rightarrow \mathbb{I}_{\mathcal{E}}$ 
almost surely as $n \rightarrow \infty$. Hence, by 
Dominated Convergence Theorem, 
we have $\mathbb{E}(\mathbb{I}_{\mathcal{E}^{(n)}}) \rightarrow \mathbb{E}(\mathbb{I}_{\mathcal{E}})$.

We can prove the same statement for $r=B$ in a similar method; but we need to define 
$\mathcal{E}=\cap_{\gamma' \neq \gamma}\mathcal{E}_{B,\gamma'}  
\cap_{u \leq B-1}  \mathcal{E}_u $.

Hence, the lemma is proved. 
\end{proof}

\begin{figure*}[!t]
\footnotesize
\begin{eqnarray}
 V^{(k)}(r)&=&V^{(k-1)}(r)+ \overline{a}(k) \frac{a(\nu(r,k))}{\overline{a}(k)} \mathbb{I}\{r \in \mathcal{I}_k\} \bigg[ \min \bigg \{ \min_{\gamma}(\gamma+ 
\xi_{out}^{(k-1)} Q_{out}(r,\gamma, w_r)) +\xi_{relay}^{(k-1)}, -V^{(k-1)}(1) 
 +V^{(k-1)}(r+1) \bigg \} -V^{(k-1)}(r) \bigg],\nonumber\\ 
&& \forall 1 \leq r \leq B-1 \nonumber\\
 V^{(k)}(B)&=&V^{(k-1)}(B)+ \overline{a}(k) \frac{a(\nu(r,B))}{\overline{a}(k)} \mathbb{I}\{B \in \mathcal{I}_k\}  \bigg[  \min_{\gamma}(\gamma+ 
 \xi_{out}^{(k-1)} Q_{out}(B,\gamma, w_B))  +\xi_{relay}^{(k-1)}-V^{(k-1)}(B) \bigg] \nonumber\\
\xi_{out}^{(k)}&=& \xi_{out}^{(k-1)}+   \mathbb{I}\{N_k=N_{k-1}+1\} \bigg( b(N_k) \lim_{\beta \downarrow 0} \frac{\Lambda_{[0,A_1]}\bigg(\xi_{out}^{(k-1)}+ \beta (Q_{out}^{(N_k,N_{k-1})}-\overline{q}U_{N_k}) \bigg)-\xi_{out}^{(k-1)}}{\beta} + o(b(N_k)) \bigg)\nonumber\\
&=& \xi_{out}^{(k-1)}+   \mathbb{I}\{N_k=N_{k-1}+1\} \overline{a}(k)  \frac{b(N_k)}{\overline{a}(k)} \bigg( \lim_{\beta \downarrow 0} \frac{\Lambda_{[0,A_1]}\bigg(\xi_{out}^{(k-1)}+ \beta (Q_{out}^{(N_k,N_{k-1})}-\overline{q}U_{N_k}) \bigg)-\xi_{out}^{(k-1)}}{\beta}+ \frac{o(b(N_k))}{b(N_k)} \bigg)   \nonumber\\
\xi_{relay}^{(k)}&=& \xi_{relay}^{(k-1)} +   \mathbb{I}\{N_k=N_{k-1}+1\} \bigg( b(N_k) \lim_{\beta \downarrow 0} \frac{\Lambda_{[0,A_2]}\bigg(\xi_{relay}^{(k-1)} + \beta (1-\overline{N}U_{N_k}) \bigg)-\xi_{relay}^{(k-1)}}{\beta} + o(b(N_k)) \bigg) \nonumber\\
&=& \xi_{relay}^{(k-1)} +   \mathbb{I}\{N_k=N_{k-1}+1\} \overline{a}(k)  \frac{b(N_k)}{\overline{a}(k)} \bigg( \lim_{\beta \downarrow 0} \frac{\Lambda_{[0,A_2]}\bigg(\xi_{relay}^{(k-1)} + \beta (1-\overline{N}U_{N_k}) \bigg)-\xi_{relay}^{(k-1)}}{\beta} + \frac{o(b(N_k))}{b(N_k)} \bigg) 
\label{eqn:convergence_optasyougoadaptivelearning_first_approximation}
\end{eqnarray}
\normalsize
\hrule
\end{figure*}

\subsection{Proof of Theorem~\ref{theorem:convergence_OptAsYouGoAdaptiveLearning}}
\label{subsection:proof_of_convergence_OptAsYouGoAdaptiveLearning}

We denote the shadowing in the link between the potential locations 
 located at distances $i \delta$ and $j \delta$ from the sink node, by the random variable $W_{i,j}$. 
 The sample space $\Omega$  is defined to be the collection of all $\omega$  such that each $\omega$ 
 corresponds to a fixed realization $\{w_{i,j}: i \geq 0, j \geq 0, i>j, 1 \leq i-j \leq B \}$ 
 of   shadowing   that could be encountered in the deployment process over infinite horizon. 
 Let $\mathcal{F}$ be the Borel $\sigma$-algebra on  $\Omega$. 
 We also define a sequence of sub-$\sigma$ fields
 $\mathcal{F}_k:=\sigma \bigg(W_{i,j}: i \geq 0, j \geq 0, k \geq i>j, 1 \leq i-j \leq B \bigg)$;  $\mathcal{F}_k$ is increasing in $k$, and 
captures the history of the deployment process 
up to  $k \delta$ distance.

Let us recall the outline of the proof of 
Theorem~\ref{theorem:convergence_OptAsYouGoAdaptiveLearning} in 
Section~\ref{subsection:OptAsYouGoAdaptiveLearning-algorithm-description-discussion}.

\subsubsection{\textbf{The Faster Time-Scale Iteration of $\underline{V}^{(k)}$}} 
\label{subsubsection:lemma_used_in_convergence_proof_of_optasyougoadaptive_learning}

Let us denote by $\underline{V}^*(\xi_{out}, \xi_{relay})$ the value of 
$\underline{V}^*$, for given $\xi_{out}$ and $\xi_{relay}$. 
Let us also define $\overline{a}(k):=\max_{r \in \mathcal{I}_k} a ( \nu(r,k) )$.

Using the first order Taylor series expansion of the function $\Lambda_{[0,A_1]}(\cdot)$, and 
using the fact that $\Lambda_{[0,A_1]}(\xi_{out}^{(k-1)})=\xi_{out}^{(k-1)}$ (since $\xi_{out}^{(k-1)} \in [0,A_1]$), 
we rewrite the update equation (\ref{eqn:OptAsYouGoAdaptiveLearning_update_part}) 
as (\ref{eqn:convergence_optasyougoadaptivelearning_first_approximation}). 
Now, for the update equation for $\xi_{relay}$ in (\ref{eqn:convergence_optasyougoadaptivelearning_first_approximation}), 
we can write:

\footnotesize
\begin{eqnarray*}
&& \lim_{\beta \downarrow 0} \frac{\Lambda_{[0,A_2]}\bigg(\xi_{relay}^{(k-1)} + \beta (1-\overline{N}U_{N_k}) \bigg)-\xi_{relay}^{(k-1)}}{\beta} \\
&=& (1-\overline{N}U_{N_k}) \mathbb{I} \{0< \xi_{relay}^{(k-1)} <A_2\}\\
&+& (1-\overline{N}U_{N_k})^+ \mathbb{I} \{\xi_{relay}^{(k-1)} =0 \} \\
&-& (1-\overline{N}U_{N_k})^- \mathbb{I} \{\xi_{relay}^{(k-1)} =A_2 \} \\
\end{eqnarray*}
\normalsize

where $y^+=\max\{y,0\}$ and $y^-=-\min \{y,0\}$. 
We can write similar expression for the  
$\xi_{out}^{(k)}$ update. Since outage probabilities and placement distances are bounded quantities, and since 
$N_k \geq \lfloor \frac{k}{B} \rfloor$ and $\lim_{k \rightarrow 0}\frac{b(\lfloor \frac{k}{B} \rfloor)}{\overline{a}(k)}=0$, we have:

\footnotesize
\begin{eqnarray*}
&& \lim_{k \rightarrow \infty}  \bigg(\frac{b(N_k)}{\overline{a}(k)} \bigg( \lim_{\beta \downarrow 0} \bigg(\Lambda_{[0,A_1]}\bigg(\xi_{out}^{(k-1)}+ \beta (Q_{out}^{(N_k,N_{k-1})} \\
&& -\overline{q}U_{N_k}) \bigg) -\xi_{out}^{(k-1)}\bigg)/ \beta  + \frac{o(b(N_k))}{b(N_k)} \bigg) \bigg)=0
\end{eqnarray*}
\normalsize

Similar claim can be made for $\xi_{relay}$ update.

\begin{lemma}\label{lemma:boundedness_V_iteration_optasyougoadaptivelearning}
Under Algorithm~\ref{algorithm:OptAsYouGoAdaptiveLearning}, the faster timescale iterates 
$\{\underline{V}^{(k)}\}_{k \geq 1}$ are almost surely bounded.
\end{lemma}
\begin{proof}
  Note that, (\ref{eqn:convergence_optasyougoadaptivelearning_first_approximation}) combines the faster 
  and slower timescale iterations in a single timescale where the step size is $\overline{a}(n)$. 
  We will now use the theory from \cite[Section~$3$]{lakshminarayanan-bhatnagar14stability-two-timescale} 
  to prove this lemma. 
  
  Note that, the R.H.S. of the faster timescale iteration in 
  (\ref{eqn:convergence_optasyougoadaptivelearning_first_approximation}) is Lipschitz continuous in both faster and slower 
  timescale iterates. Hence, the first part of \cite[Assumption~$2.1$]{lakshminarayanan-bhatnagar14stability-two-timescale} 
  is satisfied.

  \cite[Assumption~$2.2$]{lakshminarayanan-bhatnagar14stability-two-timescale} 
  can be checked, using similar arguments as in checking 
  \cite[Assumption~$5(ii)$]{bhatnagar11borkar-meyn-theorem-asynchronous-stochastic-approximation} 
  in the proof of Theorem~\ref{theorem:OptAsYouGoLearning}.

  Also, $\sum_{n=1}^{\infty}\overline{a}(n) \geq \sum_{n=1}^{\infty}a(n)=\infty$ and 
  $\sum_{n=1}^{\infty}\overline{a}^2(n)\leq \sum_{n=1}^{\infty} a^2( \lfloor \frac{n}{B} \rfloor) <\infty$, which satisfies 
  \cite[Assumption~$2.3$]{lakshminarayanan-bhatnagar14stability-two-timescale}.
  
  {\em Checking \cite[Assumption~$2.4$]{lakshminarayanan-bhatnagar14stability-two-timescale}:} Let us consider 
  the following set of o.d.e. (similar to what we considered in the proof of Theorem~\ref{theorem:OptAsYouGoLearning}):
  $ \dot{V}_t (r) = \kappa_t(r) f_r(\underline{V}_t,\xi_{out}(t),\xi_{relay}(t))$ for $r \in \{1,2,\cdots,B\}$, 
  $\dot{\xi}_{out}(t)=0$ and $\dot{\xi}_{relay}(t)=0$ (recall the interpretation of $\kappa_t(r)$ 
  from Appendix~\ref{appendix:learning-for-pure-as-you-go-deployment-given-xi}). Note that, 
$\lim_{c \rightarrow \infty} \frac{f_r(c\underline{V}, c \xi_{out}, c \xi_{relay})}{c}=\mathbb{E}_{W}\min\{\xi_{out}Q_{out}(r,\gamma,W)+\xi_{relay},-V(1)+V(r+1)\}-V(r)$ for $r \neq B$, and 
$\lim_{c \rightarrow \infty} \frac{f_B(c\underline{V},  c \xi_{out}, c \xi_{relay})}{c}=\xi_{out} \mathbb{E}_{W}Q_{out}(B,\gamma,W)+\xi_{relay}-V(B)$. Note that $\frac{f_r(c\underline{V})}{c}$ for all 
$r$ and $\lim_{c \rightarrow \infty} \frac{f_r(c\underline{V}, c \xi_{out}, c \xi_{relay})}{c}$ all are continuous in 
$(\underline{V}, \xi_{out},\xi_{relay})$, and $\frac{f_r(c\underline{V},  c \xi_{out}, c \xi_{relay})}{c}$ 
is decreasing in $c$. Hence, 
by Theorem~$7.13$ of \cite{rudin76principles-of-mathematical-analysis}, convergence of $\frac{f_r(c\underline{V}, c \xi_{out}, c \xi_{relay})}{c}$ 
over compacts is uniform. Hence, one part of \cite[Assumption~$2.4$]{lakshminarayanan-bhatnagar14stability-two-timescale} 
is proved. 
Next, by similar analysis done while checking 
\cite[Assumption~4]{bhatnagar11borkar-meyn-theorem-asynchronous-stochastic-approximation}
in the proof of Theorem~\ref{theorem:OptAsYouGoLearning} (using Lemma~\ref{lemma:ode-globally-asymptotically-stable}), 
we can verify the second part of 
\cite[Assumption~$2.4$]{lakshminarayanan-bhatnagar14stability-two-timescale}.

Hence, using similar analysis as in 
\cite[Section~$3$, Theorem~$11$]{lakshminarayanan-bhatnagar14stability-two-timescale} (adapted to the case of asynchronous 
stochastic approximation), we can claim that 
$||\underline{V}^{(k)}|| \leq C^* (1+\xi_{out}^{(k)}+\xi_{relay}^{(k)})$ for all $k \geq 1$, for some 
$C^*>0$. Now, since the slower timescale iterates are bounded in our problem, the faster timescale iterates are 
also bounded. This completes the proof of Lemma~\ref{lemma:boundedness_V_iteration_optasyougoadaptivelearning}.
\end{proof}

\begin{lemma}\label{lemma:adaptive_learning_faster_timescale_convergence}
\label{lemma:lemma_used_in_convergence_proof_of_optasyougoadaptive_learning}
 For Algorithm~\ref{algorithm:OptAsYouGoAdaptiveLearning}, we have 
 $(\underline{V}^{(k)},\xi_{out}^{(k)},\xi_{relay}^{(k)}) \rightarrow \{(\underline{V}^*(\xi_{out},\xi_{relay}),\xi_{out},\xi_{relay}): 
(\xi_{out},\xi_{relay}) \in [0,A_1] \times [0,A_2] \}$ almost surely, i.e.,  
$\lim_{k \rightarrow \infty}||\underline{V}^{(k)}-\underline{V}^*(\xi_{out}^{(k)},\xi_{relay}^{(k)})||=0$ almost surely.
\end{lemma}

\begin{proof}
Note that, the functions  
$f_r(\underline{V},\xi_{out},\xi_{relay})=
\mathbb{E}_W \bigg[ \min \bigg \{ \min_{\gamma}(\gamma+ \xi_{out} Q_{out}(r,\gamma, W))+\xi_{relay}, 
-V(1)+V(r+1) \bigg \} -V(r) \bigg]$ and  
$f_B(\underline{V},\xi_{out},\xi_{relay})=
\mathbb{E}_W \bigg[ \min_{\gamma}(\gamma+ \xi_{out} Q_{out}(B,\gamma, W))+\xi_{relay}-V(B)  \bigg]$ 
are Lipschitz continuous in all arguments (by Theorem~$\ref{theorem:V-continuous-in-xi}$), and the 
collection of 
o.d.e. $\dot{\underline{V}}_r(t)=\kappa_t(r) f_r(\underline{V}(t),\xi_{out},\xi_{relay})$ for all 
$r \in \{1,2,\cdots,B\}$ 
(see \cite[Theorem~$2$, Chapter~$7$]{borkar08stochastic-approximation-book} and 
the proof of Theorem~\ref{theorem:OptAsYouGoLearning} for an interpretation of 
$\kappa_t(r)$) has a unique 
globally asymptotically stable equilibrium $\underline{V}^*(\xi_{out},\xi_{relay})$ for any 
$\xi_{out} \geq 0$, $\xi_{relay} \geq 0$ (see Lemma~\ref{lemma:ode-globally-asymptotically-stable} in the proof of 
Theorem~\ref{theorem:OptAsYouGoLearning}). Also, by Theorem~\ref{theorem:V-continuous-in-xi}, 
$\underline{V}^*(\xi_{out},\xi_{relay})$ 
is Lipschitz continuous in $\xi_{out}$ and $\xi_{relay}$. On the other hand, by 
Lemma~\ref{lemma:boundedness_V_iteration_optasyougoadaptivelearning} and the 
projection in the slower timescale, the iterates are almost surely bounded.

Hence, by a similar argument as in the proof  \cite[Lemma~$1$, Chapter~$6$]{borkar08stochastic-approximation-book}, 
and by Theorem~\ref{theorem:OptAsYouGoLearning}, 
$(\underline{V}^{(k)},\xi_{out}^{(k)},\xi_{relay}^{(k)})$ 
converges to the internally chain transitive invariant sets of the collection of o.d.e. given by 
$\dot{V}_r(t)=\kappa_t(r) f_r(\underline{V}(t),\xi_{out},\xi_{relay})$ for all $r \in \{1,2,\cdots,B\}$, 
$\dot{\xi}_{out}(t)=0$, $\dot{\xi}_{relay}(t)=0$ (where $\underline{V}(t):=\{V_1(t),V_2(t),\cdots,V_B(t) \}$). Hence, 
$(\underline{V}^{(k)},\xi_{out}^{(k)},\xi_{relay}^{(k)}) \rightarrow \{(\underline{V}^*(\xi_{out},\xi_{relay}),\xi_{out},\xi_{relay}): 
(\xi_{out},\xi_{relay}) \in [0,A_1] \times [0,A_2] \}$ and 
$\lim_{k \rightarrow \infty}||\underline{V}^{(k)}-\underline{V}^*(\xi_{out}^{(k)},\xi_{relay}^{(k)})||=0$. 
\end{proof}

{\em Remark:} Lemma~\ref{lemma:adaptive_learning_faster_timescale_convergence} does not guarantee the convergence 
of the slower timescale iterates.

\begin{figure*}[!t]
\begin{footnotesize}
\begin{eqnarray}
 \xi_{out}^{(k)} &=& \Lambda_{\mathcal{G}} \bigg( \xi_{out}^{(k-1)}+b(k) \bigg( Q_{out}(U_k,\Gamma_k,W_{U_k})-\overline{q}U_k \bigg) \bigg) \nonumber\\
 &=& \Lambda_{\mathcal{G}} \bigg( \xi_{out}^{(k-1)}+b(k) \bigg( \underbrace{\overline{Q}_{out}^*(\xi_{out}^{(k-1)}, \xi_{relay}^{(k-1)})-\overline{q}\overline{U}^*(\xi_{out}^{(k-1)}, \xi_{relay}^{(k-1)})}_{:=f_1(\xi_{out}^{(k-1)},\xi_{relay}^{(k-1)})}        \nonumber\\
 && +   \underbrace{\overline{Q}_{out}(\underline{V}^{(k-1)},\xi_{out}^{(k-1)}, \xi_{relay}^{(k-1)})-\overline{q}\overline{U}(\underline{V}^{(k-1)},\xi_{out}^{(k-1)}, \xi_{relay}^{(k-1)})-f_1(\xi_{out}^{(k-1)},\xi_{relay}^{(k-1)})}_{:=g_1(\underline{V}^{(k-1)},\xi_{out}^{(k-1)},\xi_{relay}^{(k-1)})}    \nonumber\\
 && +   \underbrace{\overline{Q}_{out}'(\underline{V}^{(k-1)},\xi_{out}^{(k-1)}, \xi_{relay}^{(k-1)})-\overline{q}\overline{U}'(\underline{V}^{(k-1)},\xi_{out}^{(k-1)}, \xi_{relay}^{(k-1)})- \bigg( \overline{Q}_{out}(\underline{V}^{(k-1)},\xi_{out}^{(k-1)}, \xi_{relay}^{(k-1)})-\overline{q}\overline{U}(\underline{V}^{(k-1)},\xi_{out}^{(k-1)}, \xi_{relay}^{(k-1)}) \bigg)}_{:=l_1(\underline{V}^{(k-1)},\xi_{out}^{(k-1)},\xi_{relay}^{(k-1)})}    \nonumber\\
 &&  +  \underbrace{Q_{out}(U_k,\Gamma_k,W_{U_k})-\overline{q}U_k-\bigg( \overline{Q}_{out}'(\underline{V}^{(k-1)},\xi_{out}^{(k-1)}, \xi_{relay}^{(k-1)})-\overline{q}\overline{U}'(\underline{V}^{(k-1)},\xi_{out}^{(k-1)}, \xi_{relay}^{(k-1)}) \bigg) }_{:=M_1^{(k)}}   \bigg)      \bigg) \nonumber\\
 &=& \Lambda_{\mathcal{G}} \bigg( \xi_{out}^{(k-1)}+b(k) \bigg( f_1(\xi_{out}^{(k-1)},\xi_{relay}^{(k-1)})+ g_1(\underline{V}^{(k-1)},\xi_{out}^{(k-1)},\xi_{relay}^{(k-1)})+l_1(\underline{V}^{(k-1)},\xi_{out}^{(k-1)},\xi_{relay}^{(k-1)})+ M_{1}^{(k)} \bigg) \bigg) \nonumber\\
\xi_{relay}^{(k)}&=& \Lambda_{\mathcal{G}} \bigg( \xi_{out}^{(k-1)}+b(k) \bigg( 1-\overline{N}U_k \bigg) \bigg) \nonumber\\   
&=& \Lambda_{\mathcal{G}} \bigg( \xi_{relay}^{(k-1)}+b(k) \bigg( \underbrace{1-\overline{N}\overline{U}^*(\xi_{out}^{(k-1)}, \xi_{relay}^{(k-1)})}_{:=f_2(\xi_{out}^{(k-1)},\xi_{relay}^{(k-1)})}        \nonumber\\
 && +   \underbrace{1-\overline{N}\overline{U}(\underline{V}^{(k-1)},\xi_{out}^{(k-1)}, \xi_{relay}^{(k-1)})-f_2(\xi_{out}^{(k-1)},\xi_{relay}^{(k-1)})}_{:=g_2(\underline{V}^{(k-1)},\xi_{out}^{(k-1)},\xi_{relay}^{(k-1)})}    \nonumber\\
  && +   \underbrace{1-\overline{N}\overline{U}'(\underline{V}^{(k-1)},\xi_{out}^{(k-1)}, \xi_{relay}^{(k-1)})- \bigg( 1-\overline{N}\overline{U}(\underline{V}^{(k-1)},\xi_{out}^{(k-1)}, \xi_{relay}^{(k-1)}) \bigg)}_{:=l_2(\underline{V}^{(k-1)},\xi_{out}^{(k-1)},\xi_{relay}^{(k-1)})}    \nonumber\\
 &&  +  \underbrace{ 1-\overline{N}U_k-\bigg( 1-\overline{N}\overline{U}'(\underline{V}^{(k-1)},\xi_{out}^{(k-1)}, \xi_{relay}^{(k-1)})  \bigg) }_{:=M_2^{(k)}}     \bigg)    \bigg) \nonumber\\
&=& \Lambda_{\mathcal{G}} \bigg( \xi_{relay}^{(k-1)}+b(k) \bigg( f_2(\xi_{relay}^{(k-1)},\xi_{relay}^{(k-1)})+ g_2(\underline{V}^{(k-1)},\xi_{relay}^{(k-1)},\xi_{relay}^{(k-1)})+l_2(\underline{V}^{(k-1)},\xi_{out}^{(k-1)},\xi_{relay}^{(k-1)})+ M_{2}^{(k)} \bigg) \bigg)
\label{eqn:optasyougoadaptivelearning_slower_timescale_in_kushner_form}
\end{eqnarray}
\end{footnotesize}
\hrule
\end{figure*}

\subsubsection{\textbf{The slower timescale iteration}}
\label{subsubsection:posing_slower_timescale_as_projected_stochastic_approximation_by_kushner}

We will pose the slower timescale update  as a 
projected stochastic approximation (see \cite[Equation~$5.3.1$]{kushner-clark78SA-constrained-unconstrained}). 
{\bf In order to do that and to avoid complicated notation, for the rest of this appendix we will denote 
by $\underline{V}^{(k)}$, $\xi_{out}^{(k)}$ and $\xi_{relay}^{(k)}$ the values of the corresponding variable 
after placing the $k$-th relay and performing the update (earlier they were defined to be the iterates after a decision is made at the $k$-th step).} 
Let us also recall the definition of the functions $\overline{Q}_{out}(\cdot,\cdot,\cdot)$, 
$\overline{Q}_{out}^*(\cdot,\cdot)$, $\overline{U}(\cdot,\cdot,\cdot)$, $\overline{U}^*(\cdot,\cdot)$. 
Let us define the functions $\overline{Q}_{out}'(\underline{V}^{(k-1)}, \xi_{out}^{(k-1)},\xi_{relay}^{(k-1)})$ and 
$\overline{U}'(\underline{V}^{(k-1)}, \xi_{out}^{(k-1)},\xi_{relay}^{(k-1)})$ to be the mean link outage and mean length 
of the $k$-th link that is created by 
Algorithm~\ref{algorithm:OptAsYouGoAdaptiveLearning} (using the two-timescale update) 
starting with $\underline{V}^{(k-1)}$, $\xi_{out}^{(k-1)}$ and $\xi_{relay}^{(k-1)}$ (which are obtained by the algorithm 
after placing the $(k-1)$-st relay and and doing the learning/update operation; 
note that, these quantities are obtained after placing $(k-1)$ nodes and not at the $(k-1)$-th step). 

{\em The difference between $\overline{U}'(\underline{V}^{(k-1)}, \xi_{out}^{(k-1)},\xi_{relay}^{(k-1)})$ and 
$\overline{U}(\underline{V}^{(k-1)}, \xi_{out}^{(k-1)},\xi_{relay}^{(k-1)})$ can be explained as follows. 
$\overline{U}(\underline{V}^{(k-1)}, \xi_{out}^{(k-1)},\xi_{relay}^{(k-1)})$ is the mean length of the $k$-th link where no quantity 
is updated in the process of measurements made to create the $k$-th link; hence, $\overline{U}(\underline{V}^{(k-1)}, \xi_{out}^{(k-1)},\xi_{relay}^{(k-1)})$ 
is the mean placement distance of a {\em stationary} policy which is similar to Algorithm~\ref{algorithm:optimal-policy-structure} 
except that $\xi_{out}$, $\xi_{relay}$ and $\underline{V}^*$ are replaced by 
 $\xi_{out}^{(k-1)}$, $\xi_{relay}^{(k-1)}$ and $\underline{V}^{(k-1)}$ respectively. On the other hand, 
$\overline{U}'(\underline{V}^{(k-1)}, \xi_{out}^{(k-1)},\xi_{relay}^{(k-1)})$ 
is the mean length of the $k$-th link created under Algorithm~\ref{algorithm:OptAsYouGoAdaptiveLearning} (with 
$(\underline{V}^{(k-1)}, \xi_{out}^{(k-1)},\xi_{relay}^{(k-1)})$ as starting parameters), 
where the iterates are updated at each step  between placement of the $(k-1)$-th node and 
the $k$-th node.}

Let us denote by $\mathcal{G}$ the set $[0,A_1] \times [0,A_2]$, defined by the following constraints:
\begin{eqnarray}\label{eqn:constraint_equations_in_the_slower_timescale}
 -\xi_{out} \leq 0, \xi_{out} \leq A_1, -\xi_{relay} \leq 0, \xi_{relay} \leq A_2
\end{eqnarray}

{\em Clearly, projection onto the set 
$\mathcal{G}$ is nothing but coordinate wise projection.}

We rewrite the slower timescale iteration in (\ref{eqn:OptAsYouGoAdaptiveLearning_update_part}) 
as (\ref{eqn:optasyougoadaptivelearning_slower_timescale_in_kushner_form}) (note the 
definitions of the functions $f_1(\xi_{out},\xi_{relay})$, $f_2(\xi_{out},\xi_{relay})$, $g_1(\underline{V},\xi_{out},\xi_{relay})$,  
$g_2(\underline{V},\xi_{out},\xi_{relay})$, $l_1(\underline{V},\xi_{out},\xi_{relay})$ 
and $l_2(\underline{V},\xi_{out},\xi_{relay})$  in 
(\ref{eqn:optasyougoadaptivelearning_slower_timescale_in_kushner_form})). 
The random variables $M_{1}^{(k)}$ and $M_{2}^{(k)}$ are two zero mean Martingale difference noise sequences w.r.t.  
$\mathcal{F}_{k-1}$ (information available up to the $(k-1)$-st placement instant); this happens due to i.i.d. shadowing 
across links.

{\em (\ref{eqn:optasyougoadaptivelearning_slower_timescale_in_kushner_form}) has the form of a 
projected stochastic approximation 
(see \cite[Equation~$5.3.1$]{kushner-clark78SA-constrained-unconstrained}). In order to show the desired convergence of 
the iterates in (\ref{eqn:optasyougoadaptivelearning_slower_timescale_in_kushner_form}), 
we will  use   \cite[Theorem~$5.3.1$]{kushner-clark78SA-constrained-unconstrained}; 
this requires us to check five conditions 
from \cite{kushner-clark78SA-constrained-unconstrained}, which is done in the next subsection.} 
\qed

\subsubsection{\bf{\em Checking the five conditions from \cite{kushner-clark78SA-constrained-unconstrained}}}
\label{subsubsection:checking_five_conditions_kushner_optexplorelimadaptivelearning}
We will first present a lemma that will be useful for checking one condition.
\begin{lemma}\label{lemma:placement_rate_mean_outage_per_step_continuous_in_V_xio_and_xir}
Under Assumption~\ref{assumption:shadowing_continuous_random_variable}, 
the quantities $\overline{\Gamma}(\underline{V},\xi_{out},\xi_{relay})$, $\overline{Q}_{out}(\underline{V},\xi_{out},\xi_{relay})$ 
and $\overline{U}(\underline{V},\xi_{out},\xi_{relay})$ 
are continuous in $\underline{V}$, $\xi_{out}$ and $\xi_{relay}$.
\end{lemma}
\begin{proof}
  The proof is similar to that of Theorem~\ref{theorem:placement_rate_mean_outage_per_step_continuous_in_xio_and_xir}.
\end{proof}

Now, we will check   conditions $A5.1.3$, $A5.1.4$, $A5.1.5$, $A5.3.1.$ and $A5.3.2$ from 
\cite{kushner-clark78SA-constrained-unconstrained}.

{\em Checking Condition~$A5.1.3$:} We need $f_1(\cdot,\cdot)$ and $f_2(\cdot,\cdot)$ to be continuous functions; 
this holds by Theorem~\ref{theorem:placement_rate_mean_outage_per_step_continuous_in_xio_and_xir}.\qed

{\em Checking Condition~$A5.1.4$:} This condition is satisfied 
by the choice of the sequence $\{b(k)\}_{k \geq 1}$.\qed

{\em Checking Condition~$A5.1.5$:} This 
condition requires that $\lim_{k \rightarrow \infty}g_1(\underline{V}^{(k-1)},\xi_{out}^{(k-1)},\xi_{relay}^{(k-1)})=0$,  
$\lim_{k \rightarrow \infty}g_2(\underline{V}^{(k-1)},\xi_{out}^{(k-1)},\xi_{relay}^{(k-1)})=0$, 
$\lim_{k \rightarrow \infty}l_1(\underline{V}^{(k-1)},\xi_{out}^{(k-1)},\xi_{relay}^{(k-1)})=0$ and 
$\lim_{k \rightarrow \infty}l_2(\underline{V}^{(k-1)},\xi_{out}^{(k-1)},\xi_{relay}^{(k-1)})=0$ almost surely.

We can find a probability $1$ subset of the sample space $\Omega$, 
such that for any sample path in this subset the conclusions of 
Lemma~\ref{lemma:boundedness_V_iteration_optasyougoadaptivelearning} and 
Lemma~\ref{lemma:lemma_used_in_convergence_proof_of_optasyougoadaptive_learning} hold. Take one such sample 
path $\omega$. By Lemma~\ref{lemma:boundedness_V_iteration_optasyougoadaptivelearning}, for this sample path 
$\omega$,  we can find a compact subset $\mathcal{C} \subset \mathbb{R}^B$
such that $(\underline{V}^{(k)}, \xi_{out}^{(k)}, \xi_{relay}^{(k)})$ lies inside the compact set 
$\mathcal{C} \times[0,A_1]\times[0,A_2]$ for 
all $k \geq 1$ along this sample path.

By Lemma~\ref{lemma:placement_rate_mean_outage_per_step_continuous_in_V_xio_and_xir} and the fact 
that continuous functions are uniformly continuous over compact sets, we can say that  
$\overline{Q}_{out}(\underline{V}, \xi_{out}, \xi_{relay})$, 
$\overline{\Gamma}(\underline{V}, \xi_{out}, \xi_{relay})$ and 
$\overline{U}(\underline{V}, \xi_{out}, \xi_{relay})$ are uniformly continuous over the 
compact set $\mathcal{C} \times[0,A_1]\times[0,A_2]$. 
Now,  the Euclidean distance between  
$(\underline{V}^{(k)}, \xi_{out}^{(k)}, \xi_{relay}^{(k)})$ and 
$(\underline{V}^*(\xi_{out}^{(k)}, \xi_{relay}^{(k)}), \xi_{out}^{(k)}, \xi_{relay}^{(k)})$ converges to $0$ 
along the sample path $\omega$. Hence, by uniform continuity, we can say that $\lim_{k \rightarrow \infty}|\overline{Q}_{out}(\underline{V}^{(k)}, \xi_{out}^{(k)}, \xi_{relay}^{(k)})-
\overline{Q}_{out}(\underline{V}^*(\xi_{out}^{(k)}, \xi_{relay}^{(k)}), \xi_{out}^{(k)}, \xi_{relay}^{(k)})|=0$ 
and $\lim_{k \rightarrow \infty}|\overline{U}(\underline{V}^{(k)}, \xi_{out}^{(k)}, \xi_{relay}^{(k)})-
\overline{U}(\underline{V}^*(\xi_{out}^{(k)}, \xi_{relay}^{(k)}), \xi_{out}^{(k)}, \xi_{relay}^{(k)})|=0$ along this sample path 
$\omega$. Hence, $\lim_{k \rightarrow \infty}g_1(\underline{V}^{(k-1)},\xi_{out}^{(k-1)},\xi_{relay}^{(k-1)})=0$ and 
$\lim_{k \rightarrow \infty}g_2(\underline{V}^{(k-1)},\xi_{out}^{(k-1)},\xi_{relay}^{(k-1)})=0$ almost surely. 

On the other hand, since $\mathcal{C}$ is bounded, we can say that 
$\{ \underline{V}^{(k)} \}_{k \geq 1}$ is bounded for the chosen $\omega$. 
In a similar way as in the proof of Theorem~\ref{theorem:placement_rate_mean_outage_per_step_continuous_in_xio_and_xir}, 
in case of Lemma~\ref{lemma:placement_rate_mean_outage_per_step_continuous_in_V_xio_and_xir} we can show 
that $g(r,\gamma)$ is continuous in $\overline{V}$, $\xi_{out}$ and $\xi_{relay}$. Now, between the placement of 
the $(k-1)$-st relay and $k$-th relay, at each step, $g(r,\gamma)$ for all $r \in \{1,2,\cdots,B\}, \gamma \in \mathcal{S}$ 
can change at most by an amount $K^* a(k-1-B)$ (for a suitable 
constant $K^*>0$), and hence we can claim that 
$\lim_{k \rightarrow \infty}|\overline{U}'(\underline{V}^{(k-1)}, \xi_{out}^{(k-1)}, \xi_{relay}^{(k-1)})-\overline{U}(\underline{V}^{(k-1)}, \xi_{out}^{(k-1)}, \xi_{relay}^{(k-1)})|=0$, 
$\lim_{k \rightarrow \infty}|\overline{Q}_{out}'(\underline{V}^{(k-1)}, \xi_{out}^{(k-1)}, \xi_{relay}^{(k-1)})-\overline{Q}_{out}(\underline{V}^{(k-1)}, \xi_{out}^{(k-1)}, \xi_{relay}^{(k-1)})|=0$. 
Hence, we obtain that $\lim_{k \rightarrow \infty}l_1(\underline{V}^{(k-1)}, \xi_{out}^{(k-1)}, \xi_{relay}^{(k-1)})=0$ 
and $\lim_{k \rightarrow \infty}l_2(\underline{V}^{(k-1)}, \xi_{out}^{(k-1)}, \xi_{relay}^{(k-1)})=0$.

Also, $g_1(\underline{V}^{(k)},\xi_{out}^{(k)},\xi_{relay}^{(k)})$,  
$g_2(\underline{V}^{(k)},\xi_{out}^{(k)},\xi_{relay}^{(k)})$, 
$l_1(\underline{V}^{(k)},\xi_{out}^{(k)},\xi_{relay}^{(k)})$ and 
$l_2(\underline{V}^{(k)},\xi_{out}^{(k)},\xi_{relay}^{(k)})$ are uniformly bounded across $k \geq 1$, since the 
outage probabilities and placement distances are bounded quantities.

Hence, this condition is satisfied.

{\em Checking Condition~$A5.3.1$:} This condition is easy to check, and done in 
\cite[Appendix~E, Section~C4]{chattopadhyay-etal15measurement-based-impromptu-deployment-arxiv-v1}.\qed

{\em Checking Condition~$A5.3.2$:}
This condition is easy to check, and done in 
\cite[Appendix~E, Section~C4]{chattopadhyay-etal15measurement-based-impromptu-deployment-arxiv-v1}.\qed

\subsubsection{{\bf Finishing the Proof of 
Theorem~\ref{theorem:convergence_OptAsYouGoAdaptiveLearning}}}
\label{subsubsection:finishing_two_timescale_adaptive_learning_with_outage_convergence_proof}

Consider the function 
$h(\xi_{out},\xi_{relay}):=\bigg(\frac{f_1(\xi_{out},\xi_{relay})}{\overline{U}^*(\xi_{out},\xi_{relay})},\frac{f_2(\xi_{out},\xi_{relay})}{\overline{U}^*(\xi_{out},\xi_{relay})} \bigg) 
=\bigg(\frac{\overline{Q}_{out}^*(\xi_{out},\xi_{relay})}{\overline{U}^*(\xi_{out},\xi_{relay})}-\overline{q},\frac{1}{\overline{U}^*(\xi_{out},\xi_{relay})}-\overline{N} \bigg)$ and 
the map:

\footnotesize
\begin{eqnarray}
&& \overline{\Lambda}_{\mathcal{G}}(h(\xi_{out},\xi_{relay})) \nonumber\\
&=& \lim_{0<\beta \rightarrow 0} \frac{\Lambda_{\mathcal{G}}\bigg((\xi_{out},\xi_{relay})+\beta h(\xi_{out},\xi_{relay}))\bigg)-(\xi_{out},\xi_{relay})}{\beta} \nonumber\\
\label{eqn:definition_of_Lambda_G}
\end{eqnarray}
\normalsize

\begin{lemma}\label{lemma:stationary_point_is_optimal_point_adaptive_learning_proof}
If $(\xi_{out},\xi_{relay}) \in [0,A_1] \times [0,A_2]$ is a zero of 
$\overline{\Lambda}_{\mathcal{G}}\bigg(\frac{f_1(\xi_{out},\xi_{relay})}{\overline{U}^*(\xi_{out},\xi_{relay})},\frac{f_2(\xi_{out},\xi_{relay})}{\overline{U}^*(\xi_{out},\xi_{relay})} \bigg)$, 
then $(\underline{V}^*(\xi_{out},\xi_{relay}),\xi_{out},\xi_{relay}) \in \mathcal{K}(\overline{q},\overline{N})$, 
provided that $A_1$ and $A_2$ are chosen using the procedure described in Section~\ref{section:learning-for-pure-as-you-go-deployment-constrained-problem}. 
\end{lemma}
\begin{proof}
 The proof is similar to the proof of 
\cite[Lemma~$9$, Appendix~E, Section~C5]{chattopadhyay-etal15measurement-based-impromptu-deployment-arxiv-v1}.
\end{proof}

Now, by using similar arguments as in  
\cite[Appendix~E, Section~C5]{chattopadhyay-etal15measurement-based-impromptu-deployment-arxiv-v1} and using 
\cite[Theorem~$5.3.1$]{kushner-clark78SA-constrained-unconstrained}, 
We can show that the iterates $(\xi_{out}^{(k)}, \xi_{relay}^{(k)})$ will converge almost surely to 
the set of stationary points of the o.d.e.  
$(\dot{\xi}_{out}(t), \dot{\xi}_{relay}(t))=\overline{\Lambda}_{\mathcal{G}}\bigg(\frac{f_1(\xi_{out}(t),\xi_{relay}(t))}{\overline{U}^*(\xi_{out}(t),\xi_{relay}(t))},\frac{f_2(\xi_{out}(t),\xi_{relay}(t))}{\overline{U}^*(\xi_{out}(t),\xi_{relay}(t))} \bigg)$. 

Using this result and using 
Lemma~\ref{lemma:adaptive_learning_faster_timescale_convergence} and 
Lemma~\ref{lemma:stationary_point_is_optimal_point_adaptive_learning_proof}, we obtain that 
$(\underline{V}^{(k)},\xi_{out}^{(k)},\xi_{relay}^{(k)}) \rightarrow \mathcal{K}(\overline{q},\overline{N})$ 
almost surely, where $k\delta$ can be the distance from the sink or $k$ can be the index of a placed relay node 
(the result holds for both interpretations of $k$). This completes the proof of 
Theorem~\ref{theorem:convergence_OptAsYouGoAdaptiveLearning}. \qed

\end{document}